\documentclass[12pt]{article}
\usepackage[T1]{fontenc}
\usepackage[tracking=true]{microtype}
\usepackage{amsmath,amssymb,amsfonts,amsthm}
\usepackage{bm}
\usepackage{mathrsfs}
\usepackage{graphicx}
\usepackage{natbib}
\usepackage[hypertexnames=false]{hyperref}
\usepackage{xcolor}
\usepackage{enumitem}
\usepackage{algorithm}
\usepackage{algorithmic}
\usepackage{tikz}
\usetikzlibrary{arrows.meta,positioning,calc,decorations.pathreplacing}
\usepackage{booktabs}
\usepackage{multirow}
\usepackage{adjustbox}

\newtheorem{theorem}{Theorem}[section]

\newtheorem{lemma}[theorem]{Lemma}
\newtheorem{corollary}[theorem]{Corollary}

\theoremstyle{definition}

\newtheorem{assumption}[theorem]{Assumption}

\newtheorem{remark}[theorem]{Remark}

\newcommand{\E}{\mathbb{E}}
\newcommand{\Var}{\mathrm{Var}}
\newcommand{\Cov}{\mathrm{Cov}}
\newcommand{\Prob}{\mathbb{P}}

\title{Two-Sample Homogeneity Test via Entropic Optimal Transport}
\author{Yiming Ma, Hang Liu, and Weiwei Zhuang \\
Department of Statistics and Finance, School of Management,\\
University of Science and Technology of China \\
\texttt{\{mayiming,hliu01,weizh\}@ustc.edu.cn}
}
\date{}

\begin{document}

\maketitle

\begin{abstract}

This paper proposes a two-sample homogeneity test based on entropic optimal transport (EOT) maps from a common reference distribution---the uniform law on the unit ball. The test statistic is the squared $L^2$-distance between the two empirical EOT maps. For fixed entropic regularization parameter, we prove that the population map discrepancy is identifiable, derive a functional central limit theorem for the empirical map difference under the null, and establish the Gaussian quadratic-form null limit. We also prove consistency against fixed alternatives and characterize local asymptotic power under
contiguous alternatives. A weighted multiplier bootstrap is proposed to calibrate the non-pivotal null distribution, and its validity is established. Extensive simulations demonstrate that the proposed EOT-map test has reliable finite-sample size control and exhibits competitive power compared with other existing methods. The method is particularly powerful for location alternatives and, beyond a single scalar discrepancy, it provides additional diagnostic information on how the two distributions differ. Finally, a real data application concludes the paper.

\end{abstract}


\section{Introduction}\label{sec:intro}

\subsection{Literature review on two-sample homogeneity testing} Two-sample homogeneity testing is a fundamental problem in statistics. Given
two independent samples from unknown probability distributions \(\mathrm P\) and
\(\mathrm Q\) on \(\mathbb R^d\) (\(d \ge 1\)), the goal is to test
\[
H_0:\mathrm P=\mathrm Q
\qquad\text{against}\qquad
H_1:\mathrm P\neq \mathrm Q .
\]
In the univariate setting, empirical distribution functions, quantile functions,
and rank-based methods provide canonical ways to represent and compare
distributions. In multivariate settings, two-sample testing is still possible
without an ordering, but the absence of a canonical order makes it less clear
how to construct distributional representations that play the same role as
univariate distribution or quantile functions. This has motivated a broad
literature on nonparametric multivariate testing and, in particular, on methods
that compare distributions through distances, kernels, and graphs.

Representative approaches include distance-based methods, such as energy
distances and related Cramér-type statistics
\citep{baringhaus2004new,szekely2013energy}, kernel methods based on
reproducing kernel Hilbert space embeddings and maximum mean discrepancy
\citep{gretton2012kernel}, and graph-based procedures based on nearest
neighbors, matchings, or similarity graphs
\citep{friedman1979multivariate,schilling1986multivariate,henze1988multivariate,
	rosenbaum2005exact,chen2017new}.

These approaches have led to powerful and broadly applicable procedures.
Most of them, however, summarize the difference between two distributions by a
scalar discrepancy, such as a distance, divergence, kernel statistic, rank
statistic, or graph statistic. Such summaries are natural for testing, but they
do not explicitly retain a map-level representation of how each distribution is
organized relative to a common reference domain. This motivates the map-based
perspective developed in this paper: we compare two distributions through their
entropic transport maps from the same reference measure.

Optimal transport (OT) provides another geometric route to distributional comparison
by incorporating the geometry of the sample space
\citep{villani2003topics,villani2009optimal,santambrogio2015optimal,panaretos2020invitation}, leading to Wasserstein-based tests, transport-rank constructions, and distribution-free rank-based procedures \citep{ramdas2017wasserstein,hallin2021distribution,hallin2021wasserstein,HL2024,deb2023multivariate}. Among them, Wasserstein-based
methods have been used in two-sample and goodness-of-fit testing. However, inference
with empirical Wasserstein distances can be delicate in multivariate settings,
owing to nonstandard asymptotics and computational costs. Entropic
regularization offers a smoother and computationally more tractable alternative:
entropic optimal transport can be computed by Sinkhorn-type algorithms
\citep{cuturi2013sinkhorn,peyre2019computational,ma2025MSD,mena2019statistical}, and, for fixed
regularization parameter \(\varepsilon>0\), its potentials and maps admit differentiability-based limit
theory \citep{delbarrio2024central,goldfeld2024limit}.

\subsection{Our contributions}

We exploit the regularized structure of entropic optimal transport to construct
a common-reference map-based two-sample test. Let \(\mathrm U_d\) denote the
uniform distribution on the unit ball \(\mathbb B_d\). For a probability measure
\(\mathrm P\), let \(T_{\varepsilon,\mathrm P}\) be the entropic transport map
from \(\mathrm U_d\) to \(\mathrm P\), defined as the barycentric projection of
the corresponding entropic optimal coupling; see~\eqref{eq:eot-map-barycentric}
for the precise definition. The proposed method compares two distributions
\(\mathrm P\) and \(\mathrm Q\) through the squared
\(L^2(\mathbb B_d,\mathrm U_d;\mathbb R^d)\)-distance between
\(T_{\varepsilon,\mathrm P}\) and \(T_{\varepsilon,\mathrm Q}\).

Our main contributions are as follows:
\begin{itemize}
	\item We propose a two-sample homogeneity test based on common-reference
	entropic OT maps. The resulting statistic is a scalar \(L^2\)-map
	discrepancy, while the underlying vector-valued map difference retains
	directional and geometric information about distributional changes.

	\item We prove an identifiability result showing that, under compact-support
	and fixed-regularization assumptions, equality of the population EOT maps
	implies equality of the underlying distributions. Hence the proposed
	population discrepancy vanishes precisely under the homogeneity null.

	\item We establish the large-sample theory for the proposed statistic. Under
	the null hypothesis, the empirical map difference satisfies a functional
	central limit theorem and the statistic converges to a Gaussian
	quadratic-form limit. We also prove consistency under fixed alternatives and
	characterize local asymptotic power under contiguous alternatives.

	\item We develop a weighted multiplier bootstrap procedure to approximate the
	non-pivotal null distribution and prove its asymptotic validity.

	\item We demonstrate through simulations and a real-data application that the
	method has reliable finite-sample calibration, competitive power, and useful
	map-level diagnostic interpretability.
\end{itemize}

\subsection{Test procedure at a glance}

Figure~\ref{fig:map-discrepancy} summarizes the common-reference construction
underlying the proposed test. The procedure consists of three main steps. First,
the two empirical distributions are represented through entropic transport maps
from the same reference law \(\mathrm U_d\) on the unit ball \(\mathbb B_d\). Second, the two
empirical maps are evaluated on the common reference domain. Third, their
pointwise vector difference is aggregated by a squared
\(L^2(\mathbb B_d,\mathrm U_d;\mathbb R^d)\)-distance to form the test
statistic.

More explicitly, for each reference point \(\mathbf u\in\mathbb B_d\), the two mapped
locations \(T_{\varepsilon,\mathrm P}(\mathbf u)\) and
\(T_{\varepsilon,\mathrm Q}(\mathbf u)\) are compared through the vector
\(
T_{\varepsilon,\mathrm P}(\mathbf u)-T_{\varepsilon,\mathrm Q}(\mathbf u).
\)
The proposed statistic aggregates these pointwise discrepancies over the
reference domain. Thus the original comparison of two probability measures on
\(\mathbb R^d\) is converted into the comparison of two smooth vector-valued
functions on a shared domain, while still yielding a scalar test statistic for
formal inference.

The rest of the paper is organized as follows. Section~2 introduces the
entropic transport map from the common reference distribution and states the
standing assumptions. Section~3 formulates the proposed two-sample testing
procedure, establishes the identifiability of the map-based discrepancy, and
defines the empirical test statistic. Section~4 develops the asymptotic theory,
including the null limit distribution, consistency under fixed alternatives,
local asymptotic power, and the validity of the weighted multiplier bootstrap.
Section~5 presents numerical experiments on finite-sample calibration, power,
comparison with existing two-sample tests, computational cost, tuning stability,
and map-based visualization, followed by a real-data application to Citi Bike
trip records. The proofs of the main theoretical results and additional
operator-level covariance details are collected in the appendices.

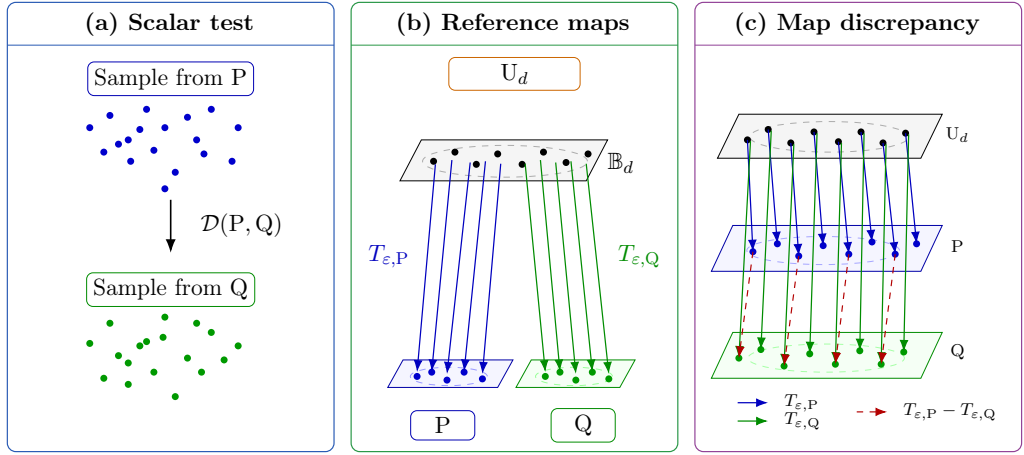
\begin{figure}[!htbp]
	\centering
	\resizebox{0.98 \textwidth}{!}{%
		\begin{tikzpicture}[
			font=\small,
			>=Latex,
			panel/.style={
				draw,
				rounded corners=4pt,
				line width=0.55pt,
				inner sep=0pt
			},
			paneltitle/.style={
				font=\bfseries\small,
				align=center,
				text width=4.2cm
			},
			labelbox/.style={
				draw,
				rounded corners=3pt,
				minimum height=0.48cm,
				align=center,
				inner sep=3pt,
				fill=white,
				font=\small
			},
			bluebox/.style={
				labelbox,
				draw=blue!70!black
			},
			greenbox/.style={
				labelbox,
				draw=green!55!black
			},
			orangebox/.style={
				labelbox,
				draw=orange!80!black
			},
			plane/.style={
				draw,
				line width=0.45pt,
				fill=gray!8
			},
			dashedellipse/.style={
				draw=gray!65,
				dashed,
				line width=0.45pt
			},
			bluepoint/.style={
				circle,
				fill=blue!80!black,
				inner sep=1.15pt
			},
			greenpoint/.style={
				circle,
				fill=green!60!black,
				inner sep=1.15pt
			},
			blackpoint/.style={
				circle,
				fill=black,
				inner sep=1.15pt
			},
			mapblue/.style={
				->,
				blue!75!black,
				line width=0.58pt
			},
			mapgreen/.style={
				->,
				green!55!black,
				line width=0.58pt
			},
			diffred/.style={
				->,
				red!70!black,
				dashed,
				line width=0.55pt
			}
			]
			
			\definecolor{panelblue}{RGB}{40,90,180}
			\definecolor{panelgreen}{RGB}{40,150,70}
			\definecolor{panelpurple}{RGB}{145,65,150}
			
			\def\panelW{5.35}
			\def\panelH{7.35}
			\def\gap{0.28}
			
			\coordinate (A0) at (0,0);
			\coordinate (B0) at ({\panelW+\gap},0);
			\coordinate (C0) at ({2*(\panelW+\gap)},0);
			
			\node[
			panel,
			minimum width=\panelW cm,
			minimum height=\panelH cm,
			draw=panelblue
			] (panelA) at ($(A0)+(\panelW/2,-\panelH/2)$) {};
			
			\node[
			panel,
			minimum width=\panelW cm,
			minimum height=\panelH cm,
			draw=panelgreen
			] (panelB) at ($(B0)+(\panelW/2,-\panelH/2)$) {};
			
			\node[
			panel,
			minimum width=\panelW cm,
			minimum height=\panelH cm,
			draw=panelpurple
			] (panelC) at ($(C0)+(\panelW/2,-\panelH/2)$) {};
			
			\draw[panelblue, line width=0.45pt]
			($(A0)+(0,-0.70)$) -- ($(A0)+(\panelW,-0.70)$);
			\draw[panelgreen, line width=0.45pt]
			($(B0)+(0,-0.70)$) -- ($(B0)+(\panelW,-0.70)$);
			\draw[panelpurple, line width=0.45pt]
			($(C0)+(0,-0.70)$) -- ($(C0)+(\panelW,-0.70)$);
			
			\node[paneltitle] at ($(A0)+(\panelW/2,-0.35)$)
			{(a) Scalar test};
			
			\node[paneltitle] at ($(B0)+(\panelW/2,-0.35)$)
			{(b) Reference maps};
			
			\node[paneltitle] at ($(C0)+(\panelW/2,-0.35)$)
			{(c) Map discrepancy};
			
			
			\node[bluebox, minimum width=2.55cm] at ($(A0)+(\panelW/2,-1.25)$)
			{Sample from \(\mathrm P\)};
			
			\node[greenbox, minimum width=2.55cm] at ($(A0)+(\panelW/2,-4.70)$)
			{Sample from \(\mathrm Q\)};
			
			\foreach \x/\y in {
				1.35/-2.05, 1.68/-1.85, 1.98/-2.25, 2.28/-1.75,
				2.58/-2.05, 2.95/-1.88, 3.34/-1.75, 3.78/-2.05,
				1.58/-2.45, 2.02/-2.60, 2.40/-2.42, 2.75/-2.78,
				3.22/-2.48, 3.68/-2.60, 2.17/-2.08, 3.10/-2.25,
				2.58/-3.05, 1.82/-2.32
			}{
				\node[bluepoint] at ($(A0)+(\x,\y)$) {};
			}
			
			\foreach \x/\y in {
				1.35/-5.58, 1.68/-5.25, 1.98/-5.90, 2.28/-5.55,
				2.58/-5.15, 2.95/-5.82, 3.34/-5.40, 3.78/-5.62,
				1.58/-6.15, 1.98/-6.25, 2.40/-6.05, 2.75/-6.45,
				3.18/-6.05, 3.55/-5.85, 2.17/-5.62, 3.10/-5.25,
				2.55/-5.48, 1.82/-5.78
			}{
				\node[greenpoint] at ($(A0)+(\x,\y)$) {};
			}
			
			\draw[->, black, line width=0.75pt]
			($(A0)+(\panelW/2,-3.23)$)
			--
			($(A0)+(\panelW/2,-4.10)$);
			
			\node[right] at ($(A0)+(\panelW/2+0.35,-3.65)$)
			{\(\mathcal D(\mathrm P,\mathrm Q)\)};
			
			
			\node[orangebox, minimum width=2.15cm] at ($(B0)+(\panelW/2,-1.17)$)
			{\(\mathrm U_d\)};
			
			\coordinate (Bref1) at ($(B0)+(1.15,-2.25)$);
			\coordinate (Bref2) at ($(B0)+(4.20,-2.25)$);
			\coordinate (Bref3) at ($(B0)+(3.85,-2.92)$);
			\coordinate (Bref4) at ($(B0)+(0.80,-2.92)$);
			
			\draw[plane] (Bref1) -- (Bref2) -- (Bref3) -- (Bref4) -- cycle;
			\draw[dashedellipse] ($(B0)+(2.52,-2.60)$) ellipse [x radius=1.35cm, y radius=0.26cm];
			
			\node at ($(B0)+(4.42,-2.63)$) {\(\mathbb B_d\)};
			
			\foreach \x/\y in {
				1.35/-2.60, 1.70/-2.45, 2.05/-2.65, 2.40/-2.48,
				2.80/-2.65, 3.15/-2.45, 3.52/-2.62, 3.88/-2.48
			}{
				\node[blackpoint] at ($(B0)+(\x,\y)$) {};
			}
			
			\coordinate (BP1a) at ($(B0)+(0.85,-5.85)$);
			\coordinate (BP2a) at ($(B0)+(2.65,-5.85)$);
			\coordinate (BP3a) at ($(B0)+(2.40,-6.30)$);
			\coordinate (BP4a) at ($(B0)+(0.60,-6.30)$);
			
			\coordinate (BQ1a) at ($(B0)+(2.95,-5.85)$);
			\coordinate (BQ2a) at ($(B0)+(4.65,-5.85)$);
			\coordinate (BQ3a) at ($(B0)+(4.40,-6.30)$);
			\coordinate (BQ4a) at ($(B0)+(2.70,-6.30)$);
			
			\draw[plane, draw=blue!65!black, fill=blue!4]
			(BP1a) -- (BP2a) -- (BP3a) -- (BP4a) -- cycle;
			\draw[plane, draw=green!55!black, fill=green!4]
			(BQ1a) -- (BQ2a) -- (BQ3a) -- (BQ4a) -- cycle;
			
			\draw[dashedellipse, draw=blue!40]
			($(B0)+(1.60,-6.12)$) ellipse [x radius=0.63cm, y radius=0.15cm];
			\draw[dashedellipse, draw=green!40]
			($(B0)+(3.65,-6.12)$) ellipse [x radius=0.63cm, y radius=0.15cm];
			
			\foreach \x/\y in {
				1.08/-6.12, 1.32/-6.05, 1.57/-6.18, 1.85/-6.05, 2.15/-6.16
			}{
				\node[bluepoint] at ($(B0)+(\x,\y)$) {};
			}
			
			\foreach \x/\y in {
				3.18/-6.12, 3.42/-6.05, 3.68/-6.18, 3.95/-6.05, 4.22/-6.16
			}{
				\node[greenpoint] at ($(B0)+(\x,\y)$) {};
			}
			
			\foreach \xA/\yA/\xB/\yB in {
				1.38/-2.65/1.08/-6.05,
				1.65/-2.58/1.32/-6.04,
				1.95/-2.62/1.57/-6.08,
				2.20/-2.60/1.85/-6.04,
				2.45/-2.64/2.15/-6.08
			}{
				\draw[mapblue] ($(B0)+(\xA,\yA)$) -- ($(B0)+(\xB,\yB)$);
			}
			
			\foreach \xA/\yA/\xB/\yB in {
				2.85/-2.64/3.18/-6.05,
				3.10/-2.58/3.42/-6.04,
				3.35/-2.62/3.68/-6.08,
				3.60/-2.60/3.95/-6.04,
				3.85/-2.64/4.22/-6.08
			}{
				\draw[mapgreen] ($(B0)+(\xA,\yA)$) -- ($(B0)+(\xB,\yB)$);
			}
			
			\node[blue!75!black] at ($(B0)+(0.62,-4.15)$)
			{\(T_{\varepsilon,\mathrm P}\)};
			
			\node[green!55!black] at ($(B0)+(4.72,-4.15)$)
			{\(T_{\varepsilon,\mathrm Q}\)};
			
			\node[bluebox, minimum width=1.05cm] at ($(B0)+(1.50,-6.92)$)
			{\(\mathrm P\)};
			\node[greenbox, minimum width=1.05cm] at ($(B0)+(3.80,-6.92)$)
			{\(\mathrm Q\)};
			
			
			\coordinate (Cref1) at ($(C0)+(0.72,-1.83)$);
			\coordinate (Cref2) at ($(C0)+(4.05,-1.83)$);
			\coordinate (Cref3) at ($(C0)+(3.70,-2.55)$);
			\coordinate (Cref4) at ($(C0)+(0.37,-2.55)$);
			
			\draw[plane] (Cref1) -- (Cref2) -- (Cref3) -- (Cref4) -- cycle;
			\draw[dashedellipse] ($(C0)+(2.20,-2.20)$) ellipse [x radius=1.30cm, y radius=0.25cm];
			
			\node[align=center, font=\scriptsize] at ($(C0)+(4.30,-2.17)$)
			{\(\mathrm U_d\)};
			
			\foreach \x/\y in {
				0.86/-2.25, 1.20/-2.08, 1.57/-2.30, 1.95/-2.12,
				2.35/-2.28, 2.72/-2.12, 3.08/-2.30, 3.45/-2.14
			}{
				\node[blackpoint] at ($(C0)+(\x,\y)$) {};
			}
			
			\coordinate (CP1) at ($(C0)+(0.62,-3.65)$);
			\coordinate (CP2) at ($(C0)+(4.05,-3.65)$);
			\coordinate (CP3) at ($(C0)+(3.70,-4.40)$);
			\coordinate (CP4) at ($(C0)+(0.27,-4.40)$);
			
			\draw[plane, draw=blue!65!black, fill=blue!4]
			(CP1) -- (CP2) -- (CP3) -- (CP4) -- cycle;
			\draw[dashedellipse, draw=blue!35]
			($(C0)+(2.10,-4.06)$) ellipse [x radius=1.25cm, y radius=0.23cm];
			
			\node[align=center, font=\scriptsize] at ($(C0)+(4.30,-3.98)$)
			{\(\mathrm P\)};
			
			\coordinate (CQ1) at ($(C0)+(0.62,-5.40)$);
			\coordinate (CQ2) at ($(C0)+(4.05,-5.40)$);
			\coordinate (CQ3) at ($(C0)+(3.70,-6.15)$);
			\coordinate (CQ4) at ($(C0)+(0.27,-6.15)$);
			
			\draw[plane, draw=green!55!black, fill=green!4]
			(CQ1) -- (CQ2) -- (CQ3) -- (CQ4) -- cycle;
			\draw[dashedellipse, draw=green!35]
			($(C0)+(2.15,-5.82)$) ellipse [x radius=1.28cm, y radius=0.23cm];
			
			\node[align=center, font=\scriptsize] at ($(C0)+(4.30,-5.72)$)
			{\(\mathrm Q\)};
			
			\foreach \i/\xr/\yr/\xp/\yp/\xq/\yq in {
				1/0.86/-2.25/0.95/-4.08/0.72/-5.82,
				2/1.20/-2.08/1.35/-3.95/1.08/-5.68,
				3/1.57/-2.30/1.70/-4.15/1.46/-5.95,
				4/1.95/-2.12/2.10/-3.98/1.88/-5.75,
				5/2.35/-2.28/2.52/-4.12/2.30/-5.92,
				6/2.72/-2.12/2.90/-3.92/2.70/-5.70,
				7/3.08/-2.30/3.28/-4.12/3.05/-5.92,
				8/3.45/-2.14/3.63/-3.95/3.42/-5.72
			}{
				\node[bluepoint] at ($(C0)+(\xp,\yp)$) {};
				\node[greenpoint] at ($(C0)+(\xq,\yq)$) {};
				
				\draw[mapblue] ($(C0)+(\xr,\yr)$) -- ($(C0)+(\xp,\yp)$);
				
				\draw[mapgreen] ($(C0)+(\xr+0.05,\yr-0.02)$) -- ($(C0)+(\xq,\yq)$);
			}
			
			\foreach \xp/\yp/\xq/\yq in {
				0.95/-4.08/0.72/-5.82,
				1.70/-4.15/1.46/-5.95,
				2.52/-4.12/2.30/-5.92,
				3.28/-4.12/3.05/-5.92
			}{
				\draw[diffred] ($(C0)+(\xp,\yp)$) -- ($(C0)+(\xq,\yq)$);
			}
			
			\draw[mapblue] ($(C0)+(0.70,-6.55)$) -- ($(C0)+(1.20,-6.55)$);
			\node[right, font=\scriptsize] at ($(C0)+(1.32,-6.55)$)
			{\(T_{\varepsilon,\mathrm P}\)};
			
			\draw[mapgreen] ($(C0)+(0.70,-6.85)$) -- ($(C0)+(1.20,-6.85)$);
			\node[right, font=\scriptsize] at ($(C0)+(1.32,-6.85)$)
			{\(T_{\varepsilon,\mathrm Q}\)};
			
			\draw[diffred] ($(C0)+(2.65,-6.70)$) -- ($(C0)+(3.15,-6.70)$);
			\node[right, font=\scriptsize] at ($(C0)+(3.27,-6.70)$)
			{\(T_{\varepsilon,\mathrm P}-T_{\varepsilon,\mathrm Q}\)};
			
		\end{tikzpicture}%
	}
	\caption{
		Illustration of the proposed common-reference map comparison. Two
		distributions are represented by entropic transport maps from the same
		reference distribution \(\mathrm U_d\), and their discrepancy is measured
		by comparing the resulting maps over \(\mathbb B_d\).
	}
	\label{fig:map-discrepancy}
\end{figure}

\section{Entropic optimal transport}
\label{sec:eot}

Entropic regularization plays an essential role in the proposed map-based
testing framework. In contrast to unregularized optimal transport, whose
transport map is typically defined only up to source-measure null sets and may
have nonstandard empirical behavior in multivariate settings, the fixed
\(\varepsilon>0\) entropic transport map admits a canonical barycentric version
on the whole reference domain \(\mathbb B_d\). Moreover, the Schrödinger
potentials and the resulting barycentric projection depend smoothly on the
underlying distribution, which enables the Hadamard-differentiability arguments
and functional delta method used in our asymptotic analysis. From a
computational perspective, entropic OT can be solved efficiently by
Sinkhorn-type algorithms~\citep{cuturi2013sinkhorn,peyre2019computational},
which makes the repeated computations required by bootstrap calibration
feasible.

This section introduces the entropic transport map that will be used as the
basic distributional representation throughout the paper. We first define the
population map associated with a probability distribution \(\mathrm P\). We then
describe its empirical counterpart, obtained by replacing only the target
distribution by the empirical measure while keeping the source distribution
fixed. This common-source construction is the key ingredient that allows two
unknown distributions to be compared through maps defined on the same reference
domain.

Let
\(
\mathbb B_d
:=
\{\mathbf u\in\mathbb R^d:\|\mathbf u\|\le 1\}
\)
denote the unit ball in \(\mathbb R^d\). Let \(\mathrm U_d\) be the uniform
probability measure on the unit ball \(\mathbb B_d\). Throughout the paper,
\(\mathrm U_d\) serves as the fixed reference distribution. We choose
\(\mathrm U_d\) because it provides a compact, rotation-invariant, and easily
simulable reference law. The compact reference domain is convenient for
defining the map discrepancy in
\(L^2(\mathbb B_d,\mathrm U_d;\mathbb R^d)\), while rotation invariance avoids
privileging any coordinate direction. In low-dimensional applications, the
unit ball also provides a simple and interpretable domain for visualizing the
estimated map discrepancy.

For a set \(\mathcal X\subset\mathbb R^d\), let \(\mathcal P(\mathcal X)\)
denote the set of Borel probability measures on \(\mathcal X\), and let
\(\operatorname{spt}(\mu)\) denote the support of
\(\mu\in\mathcal P(\mathcal X)\). Throughout, \(\mathrm P\) and \(\mathrm Q\) denote Borel
probability measures on \(\mathbb R^d\) with
\(\operatorname{spt}(\mathrm P)\cup\operatorname{spt}(\mathrm Q)\subseteq\mathcal X\).

For a fixed \(\varepsilon>0\) and the quadratic cost
\(
c(\mathbf{u},\mathbf{x})
=
\frac12\|\mathbf{u}-\mathbf{x}\|^2,
\mathbf{u}\in\mathbb{B}_d, \mathbf{x}\in\mathcal{X},
\)
the entropic optimal transport coupling between \(\mathrm U_d\) and
\(\mathrm P\) is the unique minimizer
\begin{equation}
\label{eq:sinkhorncost}
\pi_{\varepsilon,\mathrm P}
\in
\arg\min_{\pi\in\Pi(\mathrm U_d,\mathrm P)}
\left\{
\int_{\mathbb B_d\times\mathcal X}
\frac12\|\mathbf u-\mathbf x\|^2
\,d\pi(\mathbf u,\mathbf x)
+
\varepsilon\,
\mathrm{KL}\!\left(\pi\,\|\,\mathrm U_d\otimes \mathrm P\right)
\right\},
\end{equation}
Here \(\Pi(\mathrm U_d,\mathrm P)\) is the set of couplings of
\((\mathrm U_d,\mathrm P)\) and \(\mathrm U_d\otimes \mathrm P\) is the product probability measure of \(\mathrm U_d\) and \(\mathrm P\). Here \(\mathrm{KL}(\alpha\|\beta)\) denotes the Kullback--Leibler divergence,
defined by
\[
\mathrm{KL}(\alpha\|\beta)
=
\begin{cases}
	\displaystyle \int \log\left(\frac{d\alpha}{d\beta}\right)\,d\alpha,
	& \alpha\ll \beta,\\[0.6em]
	+\infty,
	& \text{otherwise}.
\end{cases}
\]
Existence and uniqueness of the entropic optimal coupling, the Schrödinger
representation, and regularity properties of the associated potentials are
standard; see, for example,
\citet{leonard2014schrodinger}, \citet{cuturi2013sinkhorn},
\citet{peyre2019computational}, \citet{nutz2023stability}, and
\citet{goldfeld2024limit}.

The entropic regularization makes the optimal coupling strictly positive with
respect to \(\mathrm U_d\otimes \mathrm P\). In particular, the minimizer admits
the Schr\"odinger representation
\[
d\pi_{\varepsilon,\mathrm P}(\mathbf u,\mathbf x)
=
\exp\left\{
\frac{
	\varphi_{\mathrm P}(\mathbf u)
	+
	\psi_{\mathrm P}(\mathbf x)
	-
	\frac12\|\mathbf u-\mathbf x\|^2
}{\varepsilon}
\right\}
\,d\mathrm U_d(\mathbf u)\,d\mathrm P(\mathbf x),
\]
where the Schr\"odinger potentials
\((\varphi_{\mathrm P},\psi_{\mathrm P})\) are unique up to the additive
normalization
\(
(\varphi_{\mathrm P},\psi_{\mathrm P})
\mapsto
(\varphi_{\mathrm P}+C,\psi_{\mathrm P}-C),
C\in\mathbb R.
\)
The marginal constraints are equivalently given by the Schr\"odinger system
\[
e^{-\varphi_{\mathrm P}(\mathbf u)/\varepsilon}
=
\int_{\mathcal X}
\exp\left\{
\frac{
	\psi_{\mathrm P}(\mathbf x)
	-
	\frac12\|\mathbf u-\mathbf x\|^2
}{\varepsilon}
\right\}
\,d\mathrm P(\mathbf x),
\qquad
\mathbf u\in\mathbb B_d,
\]
and
\[
e^{-\psi_{\mathrm P}(\mathbf x)/\varepsilon}
=
\int_{\mathbb B_d}
\exp\left\{
\frac{
	\varphi_{\mathrm P}(\mathbf u)
	-
	\frac12\|\mathbf u-\mathbf x\|^2
}{\varepsilon}
\right\}
\,d\mathrm U_d(\mathbf u),
\qquad
\mathbf x\in\mathcal X.
\]

We now pass from the entropic coupling to a map-valued representation of \(\mathrm P\),
following the barycentric projection viewpoint for entropic transport maps
\citep{pooladian2021entropic,pooladian2022debiaser,rigollet2025sample,klatt2020empirical,goldfeld2024limit}. The entropic transport map associated with \(\mathrm P\) is
defined as the barycentric projection of \(\pi_{\varepsilon,\mathrm P}\):
\begin{equation}\label{eq:eot-map-barycentric}
	T_{\varepsilon,\mathrm P}(\mathbf u)
	:=
	\mathbb E_{\pi_{\varepsilon,\mathrm P}}
	[\mathbf X\mid \mathbf U=\mathbf u],
	\qquad
	\mathbf u\in\mathbb B_d .
	\end{equation}
Here \((\mathbf U,\mathbf X)\sim\pi_{\varepsilon,\mathrm P}\). Equivalently,
\[
T_{\varepsilon,\mathrm P}(\mathbf u)
=
\frac{
	\displaystyle
	\int_{\mathcal X}
	\mathbf x
	\exp\left\{
	\frac{
		\psi_{\mathrm P}(\mathbf x)
		-
		\frac12\|\mathbf u-\mathbf x\|^2
	}{\varepsilon}
	\right\}
	\,d\mathrm P(\mathbf x)
}{
	\displaystyle
	\int_{\mathcal X}
	\exp\left\{
	\frac{
		\psi_{\mathrm P}(\mathbf x)
		-
		\frac12\|\mathbf u-\mathbf x\|^2
	}{\varepsilon}
	\right\}
	\,d\mathrm P(\mathbf x)
}.
\]
Thus \(T_{\varepsilon,\mathrm P}\) maps each reference point
\(\mathbf u\in\mathbb B_d\) to the conditional barycenter of the target
coordinate under the entropic coupling. Under the present quadratic-cost
convention, differentiating the first equation in the Schr\"odinger system
gives the useful identity
\begin{equation}\label{eq:eot-map-gradient}
	T_{\varepsilon,\mathrm P}(\mathbf u)
	=
	\mathbf u-\nabla\varphi_{\mathrm P}(\mathbf u),
	\qquad
	\mathbf u\in\mathbb B_d.
\end{equation}
This identity connects the map to the smoothness of the Schr\"odinger potential
and will be used implicitly in the asymptotic analysis.

\noindent
\textbf{Empirical version.}
The preceding construction is population-level and depends on the unknown
distribution \(\mathrm P\). Given data, we replace \(\mathrm P\) by its empirical
measure while keeping the source distribution equal to the same fixed reference
law \(\mathrm U_d\). Let \(\mathbf X_1,\ldots,\mathbf X_n\) be sampled i.i.d.
from \(\mathrm P\), and define
\(
\mathrm P_n
=
\frac{1}{n}\sum_{i=1}^n \delta_{\mathbf X_i}.
\)
Let \(\pi_{\varepsilon,\mathrm P_n}\) denote the entropic optimal transport
coupling between \(\mathrm U_d\) and \(\mathrm P_n\). The empirical entropic
transport map is defined by the barycentric projection of
\(\pi_{\varepsilon,\mathrm P_n}\):
\[
\widehat T_{\varepsilon,n}(\mathbf u)
:=
T_{\varepsilon,\mathrm P_n}(\mathbf u)
=
\mathbb E_{\pi_{\varepsilon,\mathrm P_n}}
\bigl[
\mathbf X \mid \mathbf U=\mathbf u
\bigr],
\qquad
\mathbf u\in\mathbb B_d .
\]
Equivalently, using the Schr\"odinger representation of
\(\pi_{\varepsilon,\mathrm P_n}\), we may write
\[
\widehat T_{\varepsilon,n}(\mathbf u)
=
\frac{
	\sum_{i=1}^n
	\mathbf X_i
	\exp\left\{
	\frac{
		\widehat\psi_i
		-
		\frac12\|\mathbf u-\mathbf X_i\|^2
	}{\varepsilon}
	\right\}
}{
	\sum_{i=1}^n
	\exp\left\{
	\frac{
		\widehat\psi_i
		-
		\frac12\|\mathbf u-\mathbf X_i\|^2
	}{\varepsilon}
	\right\}
},
\qquad
\mathbf u\in\mathbb B_d,
\]
where \(\widehat\psi_i\) denotes the empirical Schr\"odinger potential evaluated
at \(\mathbf X_i\). Hence \(\widehat T_{\varepsilon,n}\) is a function on the
entire reference domain \(\mathbb B_d\), even though the target measure is
discrete. This feature is important for the two-sample procedure: the two
empirical maps will be evaluated and compared on the same source domain.

We close this section by recording the standing assumptions used in the
fixed-regularization asymptotic theory for empirical entropic OT maps; see
\citet{goldfeld2024limit}. The compact-support assumption provides a common
state space for the regularity and differentiability arguments, while the
fixed-\(\varepsilon\) assumption places the analysis in the regime where the
EOT map admits first-order limit theory.

\begin{assumption}[Compact support]
	\label{ass:compact}
	There exists a closed ball \(\mathcal X\subset\mathbb R^d\) such that
	\(
	\mathbb B_d
	\cup \operatorname{spt}(\mathrm P)
	\cup \operatorname{spt}(\mathrm Q)
	\subset \mathcal X.
	\)
\end{assumption}

\begin{assumption}[Fixed regularization]
	\label{ass:fixed-epsilon}
	The entropic regularization parameter \(\varepsilon>0\) is fixed and does not
	depend on \(n\) or \(m\).
\end{assumption}

\begin{remark}[Role of the assumptions and empirical measures]
	The assumptions above are imposed on the population laws. Although the
	empirical measures \(\mathrm P_n\) and \(\mathrm Q_m\) are discrete, this does
	not affect the definition of the entropic coupling or its barycentric
	projection. For fixed \(\varepsilon>0\), the EOT problem is well posed for
	probability measures supported on the common compact set \(\mathcal X\),
	including empirical measures.
	
	Assumption~\ref{ass:compact} provides a common state space and the uniform
	analytic bounds needed for the differentiability and empirical-process
	arguments. Assumption~\ref{ass:fixed-epsilon} places the analysis in the
	fixed-regularization regime. The results therefore concern the regularized
	EOT map \(T_{\varepsilon,\mathrm P}\) for fixed \(\varepsilon>0\), rather
	than the vanishing-regularization regime \(\varepsilon\to0\).
	
	Population distributions with unbounded support are not covered directly by
	the present theory. Extending the results to such distributions would require
	additional tail conditions or weighted function-space arguments, and is beyond
	the scope of this paper.
\end{remark}

\section{Two-sample homogeneity testing problem}
\label{sec:two-sample}

We now formulate the proposed two-sample testing procedure. The central idea is
to represent each distribution by its entropic transport map from the common
reference distribution \(\mathrm U_d\), and then to compare the resulting maps in
\(L^2(\mathbb B_d,\mathrm U_d;\mathbb R^d)\). We first state the testing problem, then
establish the identifiability property that justifies the map-based comparison,
and finally define the population discrepancy and its empirical plug-in version.

Let \(\mathrm P\) and \(\mathrm Q\) be two Borel probability measures on
\(\mathbb R^d\) satisfying Assumption~\ref{ass:compact}. We consider
the two-sample homogeneity testing problem
\(H_0:\mathrm P=\mathrm Q \quad\text{versus}\quad H_1:\mathrm P\ne \mathrm Q\).

Before defining the test statistic, we verify that comparing the entropic
transport maps \(T_{\varepsilon,\mathrm P}\) and \(T_{\varepsilon,\mathrm Q}\) does not lose information about the underlying distributions.
The following lemma shows that the map \(T_{\varepsilon,\mathrm P}\) uniquely
determines \(\mathrm P\).


\begin{lemma}[Identifiability from the entropic transport map]
	\label{lem:injectivity}
	Under Assumptions~\ref{ass:compact} and~\ref{ass:fixed-epsilon}, all
	Schr\"odinger potentials below are understood as their canonical extensions
	given by the Schr\"odinger system on \(\mathbb B_d\times\mathcal X\). Then
	\[
	\mathrm P=\mathrm Q
	\quad\Longleftrightarrow\quad
	T_{\varepsilon,\mathrm P}(\mathbf u)
	=
	T_{\varepsilon,\mathrm Q}(\mathbf u)
	\quad
	\mathrm U_d\text{-a.e.}
	\]
	as Borel probability measures on \(\mathbb R^d\).
\end{lemma}

Lemma~\ref{lem:injectivity} implies that the original homogeneity problem can
be equivalently expressed in terms of the entropic transport maps:
\(H_0: T_{\varepsilon,\mathrm P}=T_{\varepsilon,\mathrm Q}\ \mathrm U_d\text{-a.e.} \quad\text{versus}\quad H_1: T_{\varepsilon,\mathrm P}\ne T_{\varepsilon,\mathrm Q}\ \mathrm U_d\text{-a.e.}\)
Thus the common reference distribution \(\mathrm U_d\) allows the comparison of
\(\mathrm P\) and \(\mathrm Q\) to be recast as the comparison of two
vector-valued functions on \(\mathbb B_d\).

\begin{remark}[Extension to multiple samples]
	The proposed EOT-map framework can be naturally extended to the
	\(K\)-sample homogeneity problem. Suppose that independent samples are
	drawn from \(\mathrm P_1,\ldots,\mathrm P_K\), with sample sizes
	\(n_1,\ldots,n_K\), and consider
	\[
	H_0:\mathrm P_1=\cdots=\mathrm P_K .
	\]
	Let \(T_{\varepsilon,\mathrm P_k}\) denote the EOT map from the common
	reference distribution \(\mathrm U_d\) to \(\mathrm P_k\). A natural
	population discrepancy is obtained by comparing each map with the
	weighted average map
	\[
	\bar T_{\varepsilon}(\mathbf u)
	=
	\sum_{k=1}^K \lambda_k T_{\varepsilon,\mathrm P_k}(\mathbf u),
	\qquad
	\lambda_k=\frac{n_k}{\sum_{\ell=1}^K n_\ell}.
	\]
	This gives
	\[
	D_{\varepsilon,K}
	=
	\sum_{k=1}^K \lambda_k
	\int_{\mathbb B_d}
	\left\|
	T_{\varepsilon,\mathrm P_k}(\mathbf u)
	-
	\bar T_{\varepsilon}(\mathbf u)
	\right\|^2\,d\mathrm U_d(\mathbf u).
	\]
	When all \(\lambda_k>0\), this discrepancy vanishes if and only if all
	the maps \(T_{\varepsilon,\mathrm P_k}\) are equal \(\mathrm U_d\)-a.e.,
	which, by Lemma~\ref{lem:injectivity}, is equivalent to
	\(\mathrm P_1=\cdots=\mathrm P_K\).
	
	The corresponding empirical statistic is
	\[
	S_{K}
	=
	\sum_{k=1}^K n_k
	\int_{\mathbb B_d}
	\left\|
	\widehat T_{\varepsilon,k}(\mathbf u)
	-
	\widehat{\bar T}_{\varepsilon}(\mathbf u)
	\right\|^2\,d\mathrm U_d(\mathbf u),
	\qquad
	\widehat{\bar T}_{\varepsilon}(\mathbf u)
	=
	\sum_{k=1}^K \lambda_k \widehat T_{\varepsilon,k}(\mathbf u).
	\]
	Natural critical values could be obtained by extending the weighted
	bootstrap or permutation calibration to the pooled \(K\)-sample setting.
	A full theoretical treatment of this extension is left for future work.
\end{remark}

\noindent
\textbf{Population discrepancy.}
Motivated by the preceding equivalence, define the population map discrepancy
\begin{equation}\label{eq:population-discrepancy}
	\mathcal D_{\varepsilon}(\mathrm P,\mathrm Q)
	:=
	\int_{\mathbb B_d}
	\left\|
	T_{\varepsilon,\mathrm P}(\mathbf u)
	-
	T_{\varepsilon,\mathrm Q}(\mathbf u)
	\right\|^2
	\,d\mathrm U_d(\mathbf u).
\end{equation}

By Lemma~\ref{lem:injectivity},
\[
\mathcal D_{\varepsilon}(\mathrm P,\mathrm Q)=0
\quad
\Longleftrightarrow
\quad
\mathrm P=\mathrm Q .
\]
Therefore \(\mathcal D_{\varepsilon}\) is a valid population discrepancy for
the two-sample homogeneity problem. In particular, large values of
\(\mathcal D_{\varepsilon}(\mathrm P,\mathrm Q)\) indicate significant departure from
\(H_0:\mathrm P=\mathrm Q\).

\begin{remark}[Invariance properties]
	The population discrepancy is invariant under common Euclidean
	transformations of the two distributions. Let
	\(g(\mathbf x)=O\mathbf x+\mathbf a\), where \(O\) is an orthogonal
	matrix and \(\mathbf a\in\mathbb R^d\). For
	\(\mathrm P^g:=g_\#\mathrm P\) and
	\(\mathrm Q^g:=g_\#\mathrm Q\), the corresponding entropic transport maps
	satisfy
	\[
	T_{\varepsilon,\mathrm P^g}(\mathbf u)
	=
	O T_{\varepsilon,\mathrm P}(O^\top \mathbf u)+\mathbf a,
	\qquad
	T_{\varepsilon,\mathrm Q^g}(\mathbf u)
	=
	O T_{\varepsilon,\mathrm Q}(O^\top \mathbf u)+\mathbf a,
	\qquad
	\mathbf u\in\mathbb B_d .
	\]
	Hence
	\[
	T_{\varepsilon,\mathrm P^g}(\mathbf u)
	-
	T_{\varepsilon,\mathrm Q^g}(\mathbf u)
	=
	O\{T_{\varepsilon,\mathrm P}(O^\top \mathbf u)
	-
	T_{\varepsilon,\mathrm Q}(O^\top \mathbf u)\}.
	\]
	Since \(\mathrm U_d\) is invariant under orthogonal transformations and the
	Euclidean norm is rotation-invariant, it follows that
	\(
	\mathcal D_{\varepsilon}(\mathrm P^g,\mathrm Q^g)
	=
	\mathcal D_{\varepsilon}(\mathrm P,\mathrm Q).
	\)
	Thus the discrepancy does not depend on a common translation of the two
	distributions or on the choice of orthonormal coordinate system.
\end{remark}

\noindent
\textbf{Empirical test statistic.}
Since \(\mathcal D_{\varepsilon}(\mathrm P,\mathrm Q)\) depends on the unknown
population maps, we estimate them by their empirical counterparts. Let
\[
\mathrm P_n
=
\frac1n\sum_{i=1}^n\delta_{\mathbf X_i},
\qquad
\mathbf X_1,\ldots,\mathbf X_n
\overset{\mathrm{i.i.d.}}{\sim}
\mathrm P,
\]
and
\[
\mathrm Q_m
=
\frac1m\sum_{j=1}^m\delta_{\mathbf Y_j},
\qquad
\mathbf Y_1,\ldots,\mathbf Y_m
\overset{\mathrm{i.i.d.}}{\sim}
\mathrm Q,
\]
where the two samples are independent. Define the empirical entropic transport
maps
\[
\widehat T_{\varepsilon,n}^{\mathrm P}
:=
T_{\varepsilon,\mathrm P_n},
\qquad
\widehat T_{\varepsilon,m}^{\mathrm Q}
:=
T_{\varepsilon,\mathrm Q_m}.
\]
The proposed plug-in statistic is
\begin{equation}\label{eq:eot-test-statistic}
	\mathcal T_{n,m}
	=
	\frac{nm}{n+m}
	\int_{\mathbb B_d}
	\left\|
	\widehat T_{\varepsilon,n}^{\mathrm P}(\mathbf u)
	-
	\widehat T_{\varepsilon,m}^{\mathrm Q}(\mathbf u)
	\right\|^2
	\,d\mathrm U_d(\mathbf u).
\end{equation}
The statistic \eqref{eq:eot-test-statistic} is the squared
\(L^2(\mathbb B_d,\mathrm U_d;\mathbb R^d)\)-norm of the scaled empirical map difference.
The scaling factor \(nm/(n+m)\) is the effective sample size in the two-sample
problem and leads to a nondegenerate null limit under the sample-size balance
condition imposed below. Under \(H_0\), the two population maps coincide
\(\mathrm U_d\)-a.e.; we denote the common map by
\(
T_{\varepsilon,0}
:=
T_{\varepsilon,\mathrm P}
=
T_{\varepsilon,\mathrm Q}.
\)

\begin{remark}[Reference-sample approximation]
	In the theoretical statistic, the source measure is the fixed reference
	distribution \(\mathrm U_d\), and the outer integral is taken with respect to
	\(\mathrm U_d\). In numerical implementation, this integral is typically
	approximated using independent reference points
	\(
	\mathbf U_1,\ldots,\mathbf U_N
	\overset{\mathrm{i.i.d.}}{\sim}
	\mathrm U_d .
	\)
	For two-sample testing, it is natural to use the same reference points when
	evaluating the two empirical maps. This shared reference sample reduces Monte
	Carlo variability in the map difference, because the source-side approximation
	error is common to both empirical maps and is expected to cancel to first
	order under the null. The asymptotic theory developed below, however, is for
	the ideal statistic with the fixed source measure \(\mathrm U_d\).
	
	A joint limit theory that also accounts for a random empirical source measure
could be obtained by applying the functional delta method to the two-marginal
EOT map functional
\(
(\mu,\nu)\mapsto T_{\varepsilon,(\mu,\nu)},
\)
where \(\mu\) denotes the source measure and \(\nu\) denotes the target measure.
In the main theory of this paper, \(\mu=\mathrm U_d\) is fixed and only the
target distribution is estimated. In the Monte Carlo implementation, however,
\(\mathrm U_d\) is approximated by empirical reference points, which would
introduce additional Gaussian terms from the source empirical process. We do
not pursue this extension here. This distinction motivates the
reference-sample sensitivity analysis in Appendix~\ref{app:sensitivity-computation}.
	
\end{remark}

\begin{remark}[Goodness-of-fit as a known-null analogue]
	If the second distribution is a fully specified null model \(\mathrm P_0\),
	the two-sample statistic reduces to the goodness-of-fit statistic
	\(
	\mathcal T_n^{\mathrm{gof}}
	=
	n
	\int_{\mathbb B_d}
	\left\|
	T_{\varepsilon,\mathrm P_n}(\mathbf u)
	-
	T_{\varepsilon,\mathrm P_0}(\mathbf u)
	\right\|^2
	\,d\mathrm U_d(\mathbf u).
	\)
	Under \(H_0:\mathrm P=\mathrm P_0\), the same continuous-mapping argument gives
	\(
	\mathcal T_n^{\mathrm{gof}}
	\rightsquigarrow
	\left\|
	\mathbb G_\varepsilon
	\right\|_{L^2(\mathbb B_d,\mathrm U_d;\mathbb R^d)}^2 .
	\)
	This is the \(m\to\infty\) known-null analogue of
	\(\mathcal T_{n,m}\). Since the focus of this paper is the two-sample problem,
	the goodness-of-fit case is not pursued separately.
\end{remark}

\section{Asymptotic theory of the EOT-map test}
\label{sec:asymptotic-theory}

This section establishes the asymptotic properties of the proposed test. The
analysis proceeds in four steps. First, we derive a functional central limit
theorem for the empirical entropic transport map difference under the
homogeneity null, which yields the null distribution of the statistic as a
Gaussian quadratic form. Second, we prove divergence under fixed alternatives
and hence consistency. Third, we study contiguous local alternatives and
characterize the resulting local asymptotic power. Finally, we develop a
weighted multiplier bootstrap to approximate the non-pivotal null distribution
and justify the resulting bootstrap test.

\subsection{Null limit distribution and consistency}
\label{subsec:null-consistency}

We begin with the asymptotic behavior under the homogeneity null. Throughout
this section, the two sample sizes are allowed to grow at possibly different
rates, subject to the following standard balance condition.

\begin{assumption}[Sample-size balance]
	\label{ass:sample-size-balance}
	As \(n,m\to\infty\), the relative sample sizes satisfy
	\(
	\frac{n}{n+m}
	\to
	\lambda
	\in(0,1).
	\)
\end{assumption}

	Under \(H_0\), the two empirical maps estimate the same
	population map \(T_{\varepsilon,0}\). Hence the leading fluctuation of their
	difference is obtained by combining two independent one-sample EOT-map limits.
	The following theorem gives the resulting two-sample functional central limit
	theorem. The proof relies on the EOT-map central limit theorem derived from
	the Hadamard differentiability of the Schr\"odinger potentials; see
	\citep{goldfeld2024limit}.
	
	We first clarify the notation used below. For an integer \(s\ge 1\),
	\(C^{s-1}(\mathbb B_d;\mathbb R^d)\) denotes the space of vector-valued
	functions \(f:\mathbb B_d\to\mathbb R^d\) whose coordinate functions have
	continuous partial derivatives up to order \(s-1\). We also write
	\(
	L^2(\mathbb B_d,\mathrm U_d;\mathbb R^d)
	=
	\left\{
	f:\mathbb B_d\to\mathbb R^d:
	\int_{\mathbb B_d}
	\|f(\mathbf u)\|^2\,d\mathrm U_d(\mathbf u)<\infty
	\right\}.
	\)
	Unless otherwise stated, weak convergence of EOT maps is understood in
	\(L^2(\mathbb B_d,\mathrm U_d;\mathbb R^d)\). We write
\(\rightsquigarrow\) for weak convergence in the indicated function space.
For bootstrap processes, we write
\(\rightsquigarrow_{\mathbb P}\) for conditional weak convergence in probability,
where the conditioning is on the observed data
\(({\mathbf X}_1,\ldots,{\mathbf X}_n,{\mathbf Y}_1,\ldots,{\mathbf Y}_m)\).
	
	In the sequel, \(\mathbb G_{\varepsilon}^{(2)}\) denotes the Gaussian limit
	arising from the scaled difference between the two empirical EOT maps based on
	the samples from \(\mathrm P\) and \(\mathrm Q\). Here and throughout, the
	superscript \((2)\) indicates a two-sample, or equivalently two-population,
	quantity. We denote by \(\mathcal K_{\varepsilon}^{(2)}\) the covariance operator of \(\mathbb G_{\varepsilon}^{(2)}\).

\begin{theorem}[Two-sample functional CLT for empirical entropic transport maps]
	\label{thm:two-sample-functional-clt}
	Suppose that Assumptions~\ref{ass:compact}, \ref{ass:fixed-epsilon}, and
	\ref{ass:sample-size-balance} hold. Then, under
		\(H_0:\mathrm P=\mathrm Q\), for every fixed \(s\in\mathbb N_+\),
	\[
	\sqrt{\frac{nm}{n+m}}
	\left(
	\widehat T_{\varepsilon,n}^{\mathrm P}
	-
	\widehat T_{\varepsilon,m}^{\mathrm Q}
	\right)
	\rightsquigarrow
	\mathbb G_{\varepsilon}^{(2)}
	\qquad
	\text{in } C^{s-1}(\mathcal X;\mathbb R^d),
	\]
	where \(\mathbb G_{\varepsilon}^{(2)}\) is a mean-zero Gaussian random element
	in \(C^{s-1}(\mathcal X;\mathbb R^d)\). Consequently,
	\[
	\sqrt{\frac{nm}{n+m}}
	\left(
	\widehat T_{\varepsilon,n}^{\mathrm P}
	-
	\widehat T_{\varepsilon,m}^{\mathrm Q}
	\right)
	\rightsquigarrow
	\mathbb G_{\varepsilon}^{(2)}
	\qquad
	\text{in } L^2(\mathbb B_d,\mathrm U_d;\mathbb R^d).
	\]
\end{theorem}

The statistic \(\mathcal T_{n,m}\) is the squared
\(L^2(\mathbb B_d,\mathrm U_d;\mathbb R^d)\)-norm of the centered empirical map difference.
Therefore, the null limit of the scalar statistic follows directly from
Theorem~\ref{thm:two-sample-functional-clt} by the continuous mapping theorem.

\begin{theorem}[Null limit distribution of the two-sample statistic]
	\label{thm:null-limit}
	Suppose that the conditions of
	Theorem~\ref{thm:two-sample-functional-clt} hold. Under
	\(H_0:\mathrm P=\mathrm Q\),
	\[
	\mathcal T_{n,m}
	\rightsquigarrow
	\int_{\mathbb B_d}
	\left\|
	\mathbb G_{\varepsilon}^{(2)}(\mathbf u)
	\right\|^2
	\,d\mathrm U_d(\mathbf u).
	\]
	Moreover, if \(\mathcal K_{\varepsilon}^{(2)}\) denotes the covariance
	operator of \(\mathbb G_{\varepsilon}^{(2)}\) in
	\(L^2(\mathbb B_d,\mathrm U_d;\mathbb R^d)\), then
	\[
	\int_{\mathbb B_d}
	\left\|
	\mathbb G_{\varepsilon}^{(2)}(\mathbf u)
	\right\|^2
	\,d\mathrm U_d(\mathbf u)
	\overset{d}{=}
	\sum_{k=1}^{\infty}\omega_k^{(2)}Z_k^2,
	\]
	where \((\omega_k^{(2)})_{k\geq 1}\) are the nonnegative eigenvalues of
	\(\mathcal K_{\varepsilon}^{(2)}\), and \((Z_k)_{k\geq 1}\) are independent
	standard normal random variables.
\end{theorem}

The limiting law in Theorem~\ref{thm:null-limit} is generally non-pivotal,
because it depends on the covariance operator of the limiting EOT map process.
The next remark clarifies the two-sample covariance structure and relates it to
the one-sample EOT map limit.

\begin{remark}[Covariance operator in the null limit]
	\label{rem:null-covariance}
	The weights \((\omega_k^{(2)})_{k\ge1}\) in
	Theorem~\ref{thm:null-limit} are the eigenvalues of the covariance operator
	of the limiting Gaussian field \(\mathbb G_\varepsilon^{(2)}\). Under the
	null hypothesis, write the common distribution as
	\(
	\mathrm P=\mathrm Q=\mathrm P_0.
	\)
	Then
	\(
	\mathbb G_\varepsilon^{(2)}
	=
	\sqrt{1-\lambda}\,\mathbb G_\varepsilon^{\mathrm P}
	-
	\sqrt{\lambda}\,\mathbb G_\varepsilon^{\mathrm Q},
	\)
	where \(\mathbb G_\varepsilon^{\mathrm P}\) and
	\(\mathbb G_\varepsilon^{\mathrm Q}\) are independent copies of the
	one-sample EOT-map Gaussian limit associated with the common law
	\(\mathrm P_0\).
	
	Let \(\mathcal K_{\varepsilon,\mathrm P_0}\) denote the covariance operator
	of this one-sample Gaussian limit in
	\(L^2(\mathbb B_d,\mathrm U_d;\mathbb R^d)\). Then the covariance operator
	\(\mathcal K_\varepsilon^{(2)}\) of \(\mathbb G_\varepsilon^{(2)}\) satisfies
	\(
	\mathcal K_\varepsilon^{(2)}
	=
	(1-\lambda)\mathcal K_{\varepsilon,\mathrm P_0}
	+
	\lambda\mathcal K_{\varepsilon,\mathrm P_0}
	=
	\mathcal K_{\varepsilon,\mathrm P_0}.
	\)
	Thus, with the normalization \(nm/(n+m)\), the two-sample null covariance
	spectrum coincides with the one-sample EOT-map covariance spectrum under
	the common null distribution \(\mathrm P_0\). The linearized
	Schr\"odinger-system representation of \(\mathcal K_{\varepsilon,\mathrm P_0}\)
	is given in Appendix~\ref{app:linearized-covariance-operator}.
\end{remark}

Next we turn to fixed alternatives, in which the two empirical maps
converge to two distinct population maps whenever \(\mathrm P\ne\mathrm Q\).
Consequently, the unscaled squared map distance converges to the population
discrepancy \(\mathcal D_\varepsilon(\mathrm P,\mathrm Q)\), while the statistic
itself diverges because the effective sample size \(nm/(n+m)\) tends to
infinity.

\begin{theorem}[Divergence under fixed alternatives]
	\label{thm:consistency}
	Suppose that Assumptions~\ref{ass:compact}, \ref{ass:fixed-epsilon}, and
	\ref{ass:sample-size-balance} hold. Let \(\mathbb P_{\mathrm P,\mathrm Q}\)
	denote the joint law of two independent samples
	\({\mathbf X}_1,\ldots,{\mathbf X}_n\overset{\mathrm{i.i.d.}}{\sim}\mathrm P\) and
	\({\mathbf Y}_1,\ldots,{\mathbf Y}_m\overset{\mathrm{i.i.d.}}{\sim}\mathrm Q\). Under any fixed
	alternative \(\mathrm P\ne \mathrm Q\), we have
	\(
	\frac{\mathcal T_{n,m}}{nm/(n+m)}
	\overset{\mathbb P_{\mathrm P,\mathrm Q}}{\longrightarrow}
	\mathcal D_\varepsilon(\mathrm P,\mathrm Q).
	\)
	In particular,
	\(
	\mathcal T_{n,m}\to\infty \) in \(\mathbb P_{\mathrm P,\mathrm Q}\text{-probability}.
	\)
\end{theorem}

By Lemma~\ref{lem:injectivity}, the population discrepancy satisfies
\(\mathcal D_\varepsilon(\mathrm P,\mathrm Q)>0\) whenever
\(\mathrm P\ne\mathrm Q\). Hence Theorem~\ref{thm:consistency} implies that
any test rejecting for sufficiently large values of \(\mathcal T_{n,m}\) is
consistent against fixed alternatives, provided that the critical values remain
bounded in probability. The bootstrap procedure introduced below will provide
such a calibration under the null.

\subsection{Local alternatives}
\label{subsec:local-alternatives}

We next examine the behavior of the proposed test under contiguous local
alternatives. This analysis describes the sensitivity of the statistic to
departures from the null that shrink at the same rate as the stochastic
fluctuations of the empirical maps. Suppose that
Assumptions~\ref{ass:compact}, \ref{ass:fixed-epsilon}, and
\ref{ass:sample-size-balance} hold, and let \(\mathrm P_0\) denote the common
null distribution. Set \(r_{n,m}:=nm/(n+m)\).
Let \(h_{\mathrm P}\) and \(h_{\mathrm Q}\) be bounded measurable functions. For
sufficiently large \(n,m\), consider local alternatives of the form
\begin{equation}\label{eq:local-alternatives}
	\frac{d\mathrm P_{n,m,h_{\mathrm P}}}{d\mathrm P_0}
	=
	1+\frac{h_{\mathrm P}}{\sqrt{r_{n,m}}},
	\qquad
	\frac{d\mathrm Q_{n,m,h_{\mathrm Q}}}{d\mathrm P_0}
	=
	1+\frac{h_{\mathrm Q}}{\sqrt{r_{n,m}}}.
\end{equation}
Since \(\mathrm P_{n,m,h_{\mathrm P}}\) and
\(\mathrm Q_{n,m,h_{\mathrm Q}}\) are probability measures, necessarily
\(
\int h_{\mathrm P}\,d\mathrm P_0=0
\) and \(
\int h_{\mathrm Q}\,d\mathrm P_0=0.
\)

The centering conditions ensure that
\(h_{\mathrm P}\mathrm P_0\) and \(h_{\mathrm Q}\mathrm P_0\) are signed
perturbations with total mass zero. The use of the effective sample size
\(r_{n,m}\) matches the normalization in the two-sample statistic, so that these
alternatives generate a nondegenerate deterministic drift in the limiting map
process.

The following result shows that, under such local alternatives \eqref{eq:local-alternatives}, the empirical
map difference has the same Gaussian fluctuation as under the null, shifted by
a deterministic first-order drift determined by the difference
\(h_{\mathrm P}-h_{\mathrm Q}\).

\begin{theorem}[Local asymptotic distribution]
	\label{thm:local-power}
	Suppose that the assumptions of
	Theorem~\ref{thm:two-sample-functional-clt} hold. Let
	\(
	\mathbf X_1,\ldots,\mathbf X_n
	\overset{\mathrm{i.i.d.}}{\sim}
	\mathrm P_{n,m,h_{\mathrm P}}\) and \(
	\mathbf Y_1,\ldots,\mathbf Y_m
	\overset{\mathrm{i.i.d.}}{\sim}
	\mathrm Q_{n,m,h_{\mathrm Q}},
	\)
	with the two samples independent. Let \(\dot T_{\varepsilon,\mathrm P_0}\) denote the Hadamard derivative of the map
	functional \(\mathrm P\mapsto T_{\varepsilon,\mathrm P}\) at \(\mathrm P_0\). Then
	\[
	\sqrt{\frac{nm}{n+m}}
	\left(
	\widehat T_{\varepsilon,n}^{\mathrm P}
	-
	\widehat T_{\varepsilon,m}^{\mathrm Q}
	\right)
	\rightsquigarrow
	\mathbb G_{\varepsilon}^{(2)}
	+
	\Delta_{\varepsilon,h_{\mathrm P},h_{\mathrm Q}}^{(2)}
	\]
	in \(L^2(\mathbb B_d,\mathrm U_d;\mathbb R^d)\), where
	\[
	\mathbb G_{\varepsilon}^{(2)}
	=
	\sqrt{1-\lambda}\,
	\dot T_{\varepsilon,\mathrm P_0}
	\bigl[
	\mathbb G_{\mathrm P_0}^{\mathrm P}
	\bigr]
	-
	\sqrt{\lambda}\,
	\dot T_{\varepsilon,\mathrm P_0}
	\bigl[
	\mathbb G_{\mathrm P_0}^{\mathrm Q}
	\bigr]
	\]
	and
	\[
	\Delta_{\varepsilon,h_{\mathrm P},h_{\mathrm Q}}^{(2)}
	=
	\dot T_{\varepsilon,\mathrm P_0}
	\bigl[
	(h_{\mathrm P}-h_{\mathrm Q})\mathrm P_0
	\bigr].
	\]
	Consequently,
	\[
	\mathcal T_{n,m}
	\rightsquigarrow
	\int_{\mathbb B_d}
	\left\|
	\mathbb G_{\varepsilon}^{(2)}(\mathbf u)
	+
	\Delta_{\varepsilon,h_{\mathrm P},h_{\mathrm Q}}^{(2)}(\mathbf u)
	\right\|^2
	\,d\mathrm U_d(\mathbf u).
	\]
\end{theorem}

The local drift in Theorem~\ref{thm:local-power} depends only on the difference
between the two tangent directions. Thus local alternatives with
\(h_{\mathrm P}=h_{\mathrm Q}\) correspond to a common first-order perturbation
of both samples and are indistinguishable at the two-sample level. In contrast,
when \(h_{\mathrm P}-h_{\mathrm Q}\) is nonzero, the limiting statistic is a
noncentral Gaussian quadratic form. Thus the local behavior of the test is governed by the image of the tangent
difference \((h_{\mathrm P}-h_{\mathrm Q})\mathrm P_0\) under the derivative of the EOT-map functional.

We now translate this limiting distribution into a local power statement. Let
\(
\mathcal H
:=
L^2(\mathbb B_d,\mathrm U_d;\mathbb R^d).
\)
For a nominal level \(\alpha\in(0,1)\), let \(c_\alpha\) be the
\((1-\alpha)\)-quantile of the null limit
\(\|\mathbb G_{\varepsilon}^{(2)}\|_{\mathcal H}^2\), chosen so that
\(
\Prob
\left(
\|\mathbb G_{\varepsilon}^{(2)}\|_{\mathcal H}^2
>
c_\alpha
\right)
=
\alpha.
\)

\begin{corollary}[Local asymptotic power]
	\label{cor:local-power}
	Suppose that Assumptions~\ref{ass:compact}, \ref{ass:fixed-epsilon}, and
	\ref{ass:sample-size-balance} hold. Let \(G:=\mathbb G_{\varepsilon}^{(2)}\) and \(\Delta:=
	\Delta_{\varepsilon,h_{\mathrm P},h_{\mathrm Q}}^{(2)}\).
	Under the local alternatives in Theorem~\ref{thm:local-power}, the
	asymptotic rejection probability of the level-\(\alpha\) test is
	\[
	\beta(h_{\mathrm P},h_{\mathrm Q})
	=
	\Prob
	\left(
	\|G+\Delta\|_{\mathcal H}^2>c_\alpha
	\right),
	\]
	which satisfies
	\(\beta(h_{\mathrm P},h_{\mathrm Q})
	\ge
	\alpha\). Moreover, if \(h_{\mathrm P}\ne h_{\mathrm Q}\) \(\mathrm P_0\text{-a.s.}\),
	then \(\beta(h_{\mathrm P},h_{\mathrm Q}) >
	\alpha\).
\end{corollary}

Corollary~\ref{cor:local-power} shows that the test has nontrivial local power
against \(r_{n,m}^{-1/2}\)-perturbations for which the two samples move in
different tangent directions. The inequality
\(\beta(h_{\mathrm P},h_{\mathrm Q})\ge\alpha\) follows from the Anderson
inequality for centered Gaussian measures, while the strict inequality reflects
the nonzero deterministic drift induced by
\(h_{\mathrm P}-h_{\mathrm Q}\).

\subsection{Weighted multiplier bootstrap}
\label{subsec:weighted-bootstrap}

The null limit in Theorem~\ref{thm:null-limit} is generally non-pivotal, since
it depends on the unknown covariance operator of the limiting EOT map process.
Rather than estimating this operator and its eigenvalues directly, we use a
weighted multiplier bootstrap to approximate the null distribution of
\(\mathcal T_{n,m}\). The bootstrap procedure reweights the two empirical target
measures separately while keeping the common source measure fixed. 

Let
\(
\xi_1,\ldots,\xi_n,
\qquad
\zeta_1,\ldots,\zeta_m
\)
be two independent sequences of nonnegative i.i.d.\ multiplier variables,
independent of the data, satisfying
\(\E[\xi_i]=\E[\zeta_j]=1,\ \Var(\xi_i)=\Var(\zeta_j)=1\),
and
\(\E[|\xi_i|^{2+\eta}]<\infty,\ \E[|\zeta_j|^{2+\eta}]<\infty\)
for some \(\eta>0\). A standard choice is the exponential multiplier with mean
one.

Define the normalized weights
\(W_i^{\mathrm P,*} = \xi_i / \sum_{\ell=1}^n \xi_\ell,\ i=1,\ldots,n\),
and
\(W_j^{\mathrm Q,*} = \zeta_j / \sum_{\ell=1}^m \zeta_\ell,\ j=1,\ldots,m\).
The weighted empirical measures are
\[
\mathrm P_n^*
=
\sum_{i=1}^n
W_i^{\mathrm P,*}\delta_{\mathbf X_i},
\qquad
\mathrm Q_m^*
=
\sum_{j=1}^m
W_j^{\mathrm Q,*}\delta_{\mathbf Y_j}.
\]
Let
\(\widehat T_{\varepsilon,n}^{\mathrm P,*} := T_{\varepsilon,\mathrm P_n^*},\
\widehat T_{\varepsilon,m}^{\mathrm Q,*} := T_{\varepsilon,\mathrm Q_m^*}\)
be the corresponding weighted-bootstrap entropic transport maps. Such weighted bootstrap schemes are standard for empirical processes and smooth
functionals thereof
\citep{praestgaard1993exchangeably,vaart1998asymptotic,vaartwellner1996weak,kosorok2008introduction,liu2022restimators}.

To mimic the null fluctuation of the empirical map difference, we center each
bootstrap map at its original empirical counterpart and define
\begin{equation}\label{eq:bootstrap-process-statistic}
	\mathbb Z_{n,m}^*
	:=
	\sqrt{\frac{nm}{n+m}}
	\left\{
	\left(
	\widehat T_{\varepsilon,n}^{\mathrm P,*}
	-
	\widehat T_{\varepsilon,n}^{\mathrm P}
	\right)
	-
	\left(
	\widehat T_{\varepsilon,m}^{\mathrm Q,*}
	-
	\widehat T_{\varepsilon,m}^{\mathrm Q}
	\right)
	\right\}.
\end{equation}

The bootstrap analogue of the test statistic is then
\[
\mathcal T_{n,m}^*
:=
\int_{\mathbb B_d}
\left\|
\mathbb Z_{n,m}^*(\mathbf u)
\right\|^2
\,d\mathrm U_d(\mathbf u).
\]

The resulting procedure is summarized in Algorithm~\ref{alg:eot-map-two-sample-test}
in the Appendix. The following theorem states that the conditional distribution of the bootstrap
process consistently estimates the null distribution of the limiting Gaussian
map process.

\begin{theorem}[Weighted bootstrap validity]
	\label{thm:weighted-bootstrap}
	Suppose that the assumptions of
	Theorem~\ref{thm:two-sample-functional-clt} hold, and let
	\(\mathbb Z_{n,m}^*\) and \(\mathcal T_{n,m}^*\) be constructed by the
	weighted bootstrap procedure described in \eqref{eq:bootstrap-process-statistic}.
	Then, under
	\(H_0:\mathrm P=\mathrm Q=\mathrm P_0\), conditionally on the data,
	\[
	\mathbb Z_{n,m}^*
	\rightsquigarrow_{\Prob}
	\mathbb G_\varepsilon^{(2)}
	\qquad
	\text{in }
	L^2(\mathbb B_d,\mathrm U_d;\mathbb R^d),
	\]
	where \(\mathbb G_\varepsilon^{(2)}\) is the same Gaussian limit as in
	Theorem~\ref{thm:two-sample-functional-clt}. Consequently,
	\[
	\mathcal T_{n,m}^*
	\rightsquigarrow_{\Prob}
	\int_{\mathbb B_d}
	\left\|
	\mathbb G_\varepsilon^{(2)}(\mathbf u)
	\right\|^2
	\,d\mathrm U_d(\mathbf u).
	\]
\end{theorem}

Theorem~\ref{thm:weighted-bootstrap} justifies using the conditional quantiles
of \(\mathcal T_{n,m}^*\) as critical values. The next corollary records the
resulting asymptotic size under the null, local power under contiguous
alternatives, and consistency under fixed alternatives.

\begin{corollary}[Asymptotic size and power of the bootstrap test]
	\label{cor:bootstrap-size-power}
	Suppose that Assumptions~\ref{ass:compact}, \ref{ass:fixed-epsilon}, and
	\ref{ass:sample-size-balance} hold. Assume further that, under
	\(H_0:\mathrm P=\mathrm Q=\mathrm P_0\), the common distribution
	\(\mathrm P_0\) is not a Dirac measure.
	
	Let \(c_{n,m,1-\alpha}^{*}\) be the conditional \((1-\alpha)\)-quantile of
	\(\mathcal T_{n,m}^{*}\), and define
	\[
	\phi_{n,m,\alpha}^{*}
	:=
	\mathbf 1
	\left\{
	\mathcal T_{n,m}>c_{n,m,1-\alpha}^{*}
	\right\}.
	\]
	Then, under \(H_0:\mathrm P=\mathrm Q=\mathrm P_0\),
	\[
	\Prob_{H_0}
	\left(
	\phi_{n,m,\alpha}^{*}=1
	\right)
	\to
	\alpha .
	\]
	Under the local alternatives of Theorem~\ref{thm:local-power},
	\[
	\Prob_{n,m,h_{\mathrm P},h_{\mathrm Q}}
	\left(
	\phi_{n,m,\alpha}^{*}=1
	\right)
	\to
	\beta(h_{\mathrm P},h_{\mathrm Q}),
	\]
	where \(\beta(h_{\mathrm P},h_{\mathrm Q})\) is defined in
	Corollary~\ref{cor:local-power}. Finally, under any fixed alternative
	\(\mathrm P\ne\mathrm Q\),
	\[
	\Prob_{\mathrm P,\mathrm Q}
	\left(
	\phi_{n,m,\alpha}^{*}=1
	\right)
	\to
	1 .
	\]
\end{corollary}

\section{Numerical experiments}
\label{sec:numerical-experiments}

This section evaluates the finite-sample performance and computational properties of
the proposed EOT-map two-sample test. The numerical study is organized around
four questions. First, we examine whether the weighted bootstrap provides
reliable calibration under the homogeneity null. Second, we evaluate power
against representative multivariate alternatives, including location, scale,
dependence, mixture, and nonlinear deformation changes. Third, we compare the
proposed method with distance-, kernel-, graph-, rank-, and
transport-based two-sample procedures, including energy distance, MMD, Sinkhorn
divergence, nearest-neighbor graph tests, OT-rank Energy, and Wasserstein-based
tests. Finally, beyond rejection probabilities, we illustrate two diagnostic and computational advantages of the method: the map-level diagnostic information provided by the
estimated vector field and the ability to control computation through the
number of reference points used in the Monte Carlo approximation.

\subsection{Simulation setup}
\label{subsec:implementation-details}
Although the asymptotic theory is developed under compact-support assumptions,
the simulations also include Gaussian and Student-\(t\) settings. For each finite sample, the empirical measures are
finitely supported, so the empirical EOT problems are well defined. We implement
the proposed EOT-map two-sample test using
Algorithm~\ref{alg:eot-map-two-sample-test} in the Appendix, with a common Monte
Carlo reference sample for the integral over \(\mathrm U_d\).

Unless otherwise stated, tests are conducted at level \(\alpha=0.05\), with
balanced samples \(m=n\), \(N=n\) reference points, and exponential bootstrap
weights with mean one and variance one. The regularization parameter is chosen
by a median-distance pilot rule. Let
\(\mathcal Z_{n,m}=\{ \mathbf X_1,\ldots, \mathbf X_n, \mathbf Y_1,\ldots, \mathbf Y_m\}\) denote the pooled
sample. We set
\[
\varepsilon
=
c_\varepsilon
\operatorname{median}
\left\{
\|a-b\|^2:
a,b\in\mathcal Z_{n,m},\ a\ne b
\right\},
\]
with \(c_\varepsilon=0.2\) in the main experiments. Sensitivity analyses for \(\varepsilon\), \(N\), and computation time are reported
in Appendix~\ref{app:sensitivity-computation}, where we show that smaller
reference samples can substantially reduce computation time while maintaining
stable rejection probabilities.

The empirical EOT plans are computed using the Sinkhorn algorithm implemented
in the Python Optimal Transport package, with maximum \(1000\) iterations and
stopping tolerance \(10^{-9}\). All simulations were performed on a workstation
with AMD EPYC 7763 processors operating at 2.45 GHz and 512 GB RAM, using 61
threads across Monte Carlo repetitions.

\subsection{Empirical size under the null}
\label{subsec:type1-error}

We first examine the finite-sample calibration of the proposed bootstrap test
under the homogeneity null \(H_0:\mathrm P=\mathrm Q\). To cover different
distributional shapes, we consider three null distributions: the standard
Gaussian distribution \(N(0,I_d)\), a Student-\(t\) distribution with
\(\nu=5\) degrees of freedom, and the symmetric Gaussian mixture
\(\frac12 N(\mu,I_d)+\frac12 N(-\mu,I_d),\ \boldsymbol{\mu}=(1,0,\ldots,0)^\top\).
These choices represent light-tailed unimodal, heavy-tailed, and multimodal
settings, respectively. In each case, the two samples are generated independently
from the same distribution.

We conduct experiments for \(d\in\{2,5,10\}\) and balanced sample sizes
\(n=m\in\{100,500,1000,2000\}\). For each configuration, we use \(R=1000\)
Monte Carlo replications and \(B=500\) bootstrap repetitions, and set the number
of reference points to \(N=n\). We report the empirical rejection probability at
the nominal level \(\alpha=0.05\).

\begin{table}[!htbp]
	\centering
	\resizebox{\textwidth}{!}{%
	\begin{tabular}{llcccc}
		\toprule
		Null distribution
		& Dimension
		& \(n=m=100\)
		& \(n=m=500\)
		& \(n=m=1000\)
		& \(n=m=2000\) \\
		\midrule
		\multirow{3}{*}{Gaussian}
		& \(d=2\)
		& \(0.057\) & \(0.050\) & \(0.055\) & \(0.048\) \\
		& \(d=5\)
		& \(0.052\) & \(0.063\) & \(0.050\) & \(0.060\) \\
		& \(d=10\)
		& \(0.051\) & \(0.045\) & \(0.063\) & \(0.053\) \\
		\midrule
		\multirow{3}{*}{Student-\(t\)}
		& \(d=2\)
		& \(0.049\) & \(0.038\) & \(0.044\) & \(0.045\) \\
		& \(d=5\)
		& \(0.035\) & \(0.042\) & \(0.047\) & \(0.046\) \\
		& \(d=10\)
		& \(0.057\) & \(0.047\) & \(0.052\) & \(0.048\) \\
		\midrule
		\multirow{3}{*}{Gaussian mixture}
		& \(d=2\)
		& \(0.050\) & \(0.053\) & \(0.052\) & \(0.049\) \\
		& \(d=5\)
		& \(0.056\) & \(0.054\) & \(0.056\) & \(0.058\) \\
		& \(d=10\)
		& \(0.067\) & \(0.060\) & \(0.057\) & \(0.050\) \\
		\bottomrule
	\end{tabular}%
	}
	\caption{
		Empirical size of the EOT-map two-sample test under the homogeneity null.
		The nominal level is \(\alpha=0.05\). Entries report empirical rejection
		probabilities. 
	}
	\label{tab:type1-size}
\end{table}
Table~\ref{tab:type1-size} reports the empirical rejection probabilities under
the homogeneity null. Overall, the proposed bootstrap test is well calibrated
across the considered null distributions, dimensions, and sample sizes. Most
empirical sizes are close to the nominal level \(0.05\), with fluctuations of the
order expected from Monte Carlo error. The Gaussian null shows stable size
control across all dimensions. The Student-\(t\) null is slightly conservative in
a few configurations, especially when \(d=5\) and \(n=m=100\), but remains close
to the target level as the sample size increases. The Gaussian mixture null also
exhibits satisfactory calibration, with the largest empirical size \(0.067\)
occurring in the smallest-sample, highest-dimensional configuration. Overall,
these results indicate that the weighted multiplier bootstrap provides reliable
finite-sample calibration for the proposed EOT-map statistic.

\subsection{Power under fixed alternatives}
\label{subsec:power-fixed-alternatives}

We next evaluate the finite-sample power of the proposed test under three
classes of fixed alternatives. Throughout this experiment, the baseline
distribution is \(\mathrm P=N(0,I_d)\), and the second sample is generated from
one of the following alternatives:
\[
\begin{array}{ll}
	\text{location shift:}
	& \mathrm Q=N(\delta \mathbf e_1,I_d),
	\quad
	\delta\in\{0,0.1,0.2,\ldots,1.0\}, \\[2mm]
	\text{scale change:}
	& \mathrm Q=N(0,\sigma^2 I_d),
	\quad
	\sigma\in\{1.00,1.05,1.10,\ldots,1.60\}, \\[2mm]
	\text{correlation change:}
	& \mathrm Q=N(0,\Sigma_\rho),
	\quad
	\rho\in\{0,0.1,0.2,\ldots,0.9\}.
\end{array}
\]
Here \(\mathbf e_1=(1,0,\ldots,0)^\top\), and \(\Sigma_\rho\) has unit diagonal entries
and all off-diagonal entries equal to \(\rho\). The null cases correspond to
\(\delta=0\), \(\sigma=1\), and \(\rho=0\), respectively. These alternatives are
used to assess sensitivity to changes in location, marginal scale, and
dependence structure.

The experiments are conducted for \(d\in\{2,5\}\) and balanced sample sizes
\(n=m\in\{100,500\}\). For each configuration, we use \(R=1000\) Monte Carlo
replications and \(B=300\) bootstrap repetitions, set \(N=n\), and report the
empirical rejection probability at the nominal level \(\alpha=0.05\).

\begin{figure}[!htbp]
	\centering
	\includegraphics[width=0.98\textwidth]{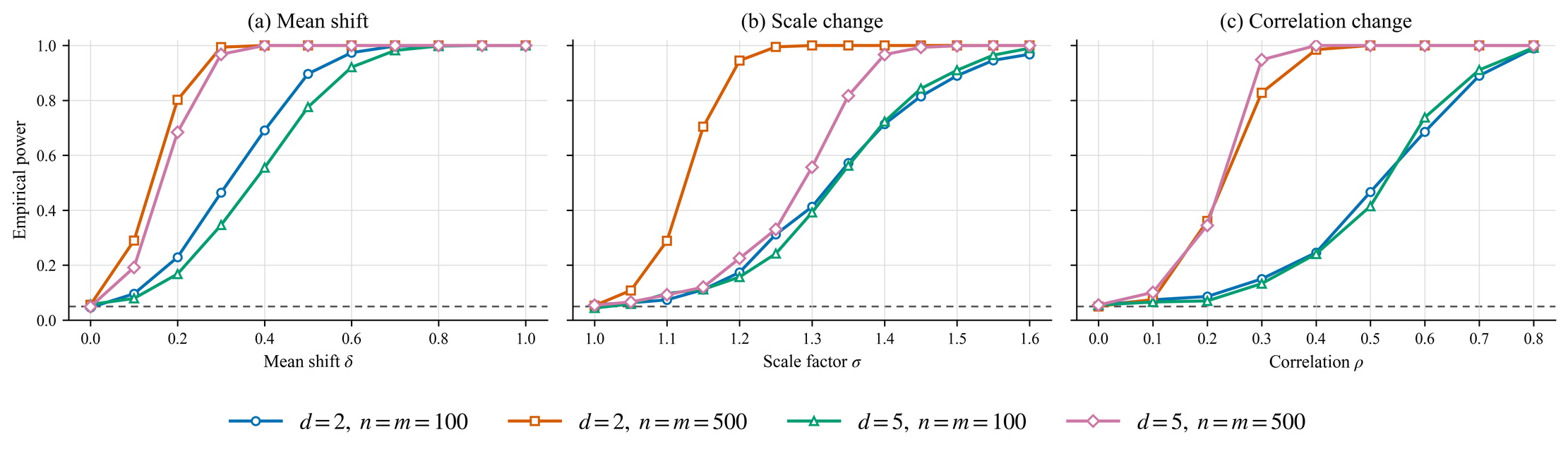}
	\caption{
		Empirical power under fixed alternatives. The nominal level is
		\(\alpha=0.05\). The panels correspond to mean-shift, scale-change,
		and correlation-change alternatives. The horizontal dashed line marks
		the nominal level. 
	}
	\label{fig:power-fixed-alternatives}
\end{figure}

Figure~\ref{fig:power-fixed-alternatives} displays the empirical power curves
under the three fixed alternatives. At the null values
\(\delta=0\), \(\sigma=1\), and \(\rho=0\), the rejection probabilities are close
to the nominal level, in agreement with the size results in
Table~\ref{tab:type1-size}. Away from the null, the power generally increases
with the signal strength and with the sample size.

The location-shift experiments show particularly strong power in the considered
settings. For \(n=m=500\), the power is
already close to one for moderate mean shifts in both dimensions. For
\(n=m=100\), the power still increases steadily with \(\delta\), although the
increase is slower in dimension \(d=5\). This suggests that the integrated
map-discrepancy statistic is particularly sensitive to coherent directional
changes in the target distribution.

The scale alternative is more challenging. In dimension \(d=2\), the test has
high power for moderate scale changes, especially when \(n=m=500\). By contrast,
when \(d=5\) and \(n=m=100\), the power remains low over the considered range
of \(\sigma\). The substantial improvement at \(n=m=500\) indicates that this
difficulty is mainly a finite-sample effect.

For the correlation alternative, the power increases sharply with \(\rho\). The
test attains high power for moderate correlations when \(n=m=500\), while the
small-sample curves increase more gradually. Compared with the scale-change
case, the higher-dimensional setting is not uniformly harder here, since a
common correlation perturbation changes many covariance entries
simultaneously. Overall, these results show that the proposed test is sensitive
to location, scale, and dependence changes, with clear power gains as the
sample size increases.

\subsection{Comparison with existing tests}
\label{subsec:comparison-existing-tests}

We compare the proposed EOT-map test with several standard two-sample
procedures: the energy test \citep{baringhaus2004new,szekely2013energy},
the Gaussian-kernel maximum mean discrepancy (MMD) test
\citep{gretton2012kernel}, the Sinkhorn divergence test
\citep{cuturi2013sinkhorn,genevay2019sample}, the Schilling
\(k\)-nearest-neighbor graph test \citep{schilling1986multivariate},
the Wasserstein test \citep{ramdas2017wasserstein}, and the OT-rank
Energy test of \citet{deb2023multivariate}. All methods are calibrated at
the nominal level \(\alpha=0.05\). The energy, MMD, Sinkhorn, \(k\)-NN,
Wasserstein, and OT-rank Energy tests use permutation calibration with
\(300\) random permutations, whereas the EOT-map test uses the weighted
bootstrap with \(300\) bootstrap repetitions
\citep{praestgaard1993exchangeably,vaartwellner1996weak,kosorok2008introduction}.
The EOT-map statistic is computed with \(N=200\) reference points sampled
from \(\mathrm U_d\) on \(\mathbb B_d\).

\noindent\textbf{Null calibration.}
Table~\ref{tab:benchmark-size} shows that all methods have reasonable size
control under the Gaussian, uniform, and Student-\(t\) nulls. The EOT-map sizes
range from \(0.039\) to \(0.064\), close to the nominal level given Monte Carlo
error.

\begin{table}[!htbp]
	\centering
	\resizebox{0.98\textwidth}{!}{
		\begin{tabular}{llccccccc}
			\toprule
			Null distribution
			& Dimension
			& Energy
			& MMD
			& Sinkhorn
			& \(k\)-NN
			& Wasserstein
			& OT-rank
			& EOT-map \\
			\midrule
			\multirow{2}{*}{Gaussian}
			& \(d=2\)
			& 0.046 & 0.037 & 0.044 & 0.057 & 0.046 & 0.046 & 0.039 \\
			& \(d=5\)
			& 0.047 & 0.048 & 0.061 & 0.047 & 0.056 & 0.052 & 0.051 \\
			\midrule
			\multirow{2}{*}{Uniform}
			& \(d=2\)
			& 0.049 & 0.055 & 0.054 & 0.050 & 0.059 & 0.052 & 0.055 \\
			& \(d=5\)
			& 0.055 & 0.053 & 0.063 & 0.051 & 0.062 & 0.066 & 0.064 \\
			\midrule
			\multirow{2}{*}{Student-\(t\)}
			& \(d=2\)
			& 0.052 & 0.054 & 0.051 & 0.038& 0.046& 0.052 & 0.050 \\
			& \(d=5\)
			& 0.042 & 0.049 & 0.034& 0.041 & 0.060 & 0.037 & 0.064 \\
			\bottomrule
		\end{tabular}
	}
	\caption{
		Empirical size of the two-sample tests.
		The nominal level is \(\alpha=0.05\), \(n=m=200\), and each entry is based on
		\(R=1000\) Monte Carlo replications.
	}
	\label{tab:benchmark-size}
\end{table}

\noindent\textbf{Power under alternatives.}
Figure~\ref{fig:location-power-comparison} summarizes the location-shift
experiment. EOT-map is among the strongest methods, especially for
weak-to-moderate shifts in dimension \(d=5\), where the location change induces
a coherent directional displacement between the two estimated maps.

\begin{figure}[!htbp]
	\centering
	\includegraphics[width=0.98\textwidth]{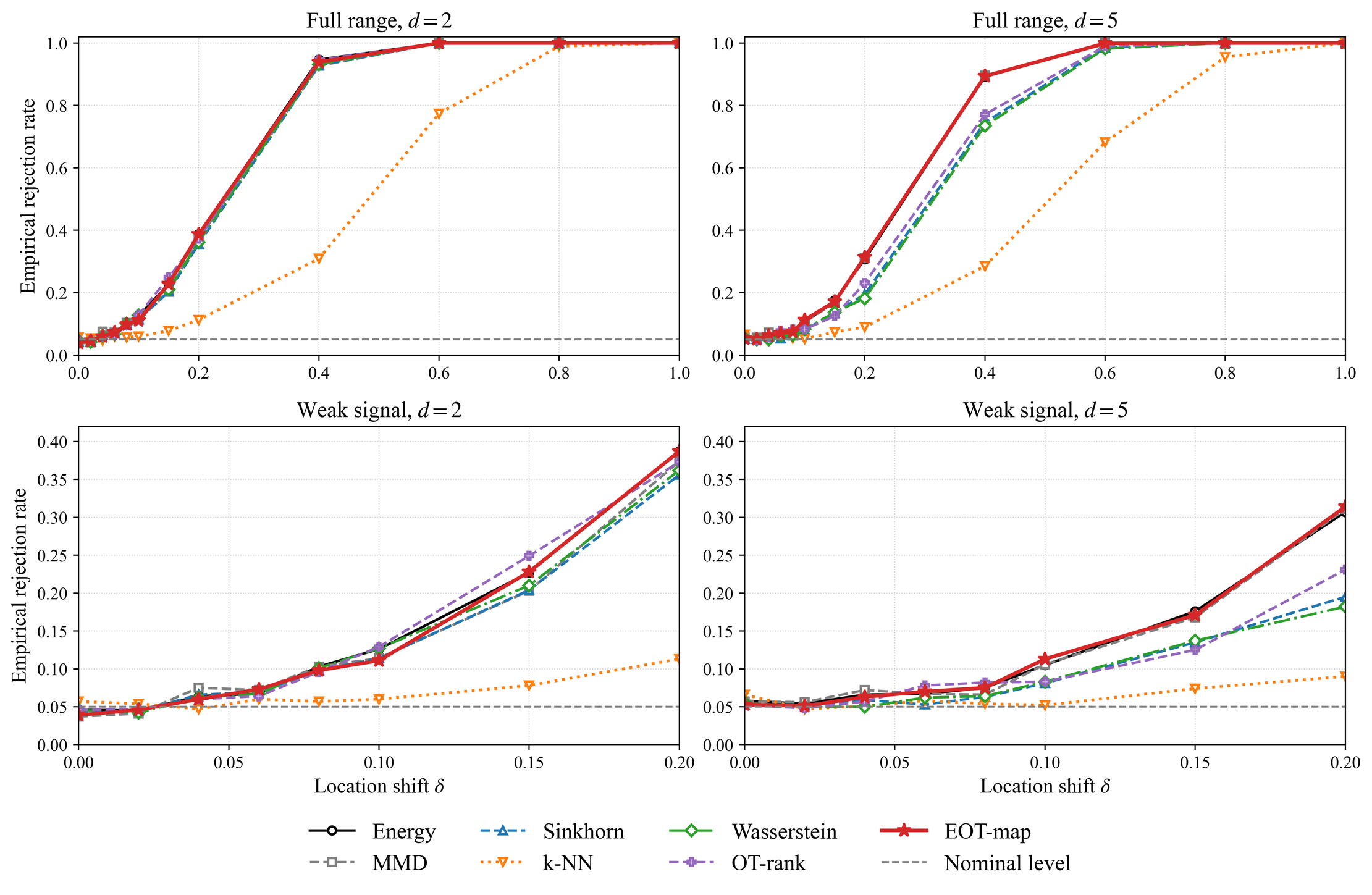}
	\caption{
		Empirical rejection rates under location alternatives
		\(\mathrm Q=N(\delta \mathbf e_1,I_d)\) with \(n=m=200\).
		The top panels show the full range of the location shift \(\delta\),
		and the bottom panels zoom in on the weak-signal region.
		The horizontal dashed line marks the nominal level \(\alpha=0.05\).
	}
	\label{fig:location-power-comparison}
\end{figure}

Table~\ref{tab:benchmark-power-other} gives representative power results for
scale, correlation, mixture, and nonlinear deformation alternatives. The full
benchmark table is reported in Appendix~\ref{app:benchmark-comparison-details}.

\begin{table}[!htbp]
	\centering
	\resizebox{0.98\textwidth}{!}{
		\begin{tabular}{lllccccccc}
			\toprule
			Alternative
			& Dimension
			& Parameter
			& Energy
			& MMD
			& Sinkhorn
			& \(k\)-NN
			& Wasserstein
			& OT-rank
			& EOT-map \\
			\midrule
			Scale & \(d=2\) & \(\sigma=1.2\) & 0.365 & 0.567 & 0.576 & 0.222 & 0.584 & 0.124 & 0.614 \\
			Scale & \(d=5\) & \(\sigma=1.2\) & 0.632 & 0.797 & 0.645 & 0.351 & 0.579 & 0.096 & 0.414 \\
			Correlation & \(d=2\) & \(\rho=0.4\) & 0.160 & 0.297 & 0.721 & 0.322 & 0.722 & 0.260 & 0.730 \\
			Correlation & \(d=5\) & \(\rho=0.2\) & 0.134 & 0.186 & 0.722 & 0.347 & 0.700 & 0.072 & 0.561 \\
			Mixture & \(d=2\) & \(\delta=1.0\) & 0.785 & 0.938 & 0.968 & 0.616 & 0.966 & 0.352 & 0.974 \\
			Mixture & \(d=5\) & \(\delta=1.0\) & 0.301 & 0.440 & 0.739 & 0.444 & 0.708 & 0.065 & 0.588 \\
			Nonlinear deformation & \(d=2\) & \(\tau=0.40\) & 0.156 & 0.292 & 0.636 & 0.297 & 0.638 & 0.233 & 0.639 \\
			Nonlinear deformation & \(d=5\) & \(\tau=0.70\) & 0.180 & 0.232 & 0.878 & 0.689 & 0.855 & 0.130 & 0.618 \\
			\bottomrule
		\end{tabular}
	}
	\caption{
		Representative empirical power under scale, correlation, mixture, and
		nonlinear deformation alternatives. The nominal level is \(\alpha=0.05\),
		\(n=m=200\), and each entry is based on \(R=1000\) Monte Carlo replications.
		Full results are reported in Appendix~\ref{app:benchmark-comparison-details}.
	}
	\label{tab:benchmark-power-other}
\end{table}

The EOT-map test is strongest or close to strongest under location alternatives,
and for two-dimensional nonlinear deformation it is nearly indistinguishable
from Sinkhorn divergence and Wasserstein. In some high-dimensional scale,
correlation, mixture, and nonlinear deformation settings, Sinkhorn divergence,
Wasserstein, or MMD can be more sensitive. Thus, EOT-map should be viewed as a
competitive and interpretable complement to scalar discrepancy tests rather than
as a uniformly dominant replacement.

The runtime results in Table~\ref{tab:benchmark-runtime} show that the current
EOT-map implementation is slower than Energy, MMD, \(k\)-NN, and OT-rank Energy,
but faster than permutation-calibrated Sinkhorn divergence; detailed timings are
reported in Appendix~\ref{app:benchmark-comparison-details}.

A distinctive advantage of the EOT-map test is that, in addition to a rejection
decision, it provides the estimated vector field
\(\widehat T_{\varepsilon,n}^{\mathrm P}(\mathbf u) - \widehat T_{\varepsilon,m}^{\mathrm Q}(\mathbf u)\),
which can be visualized to reveal whether the distributional difference is
primarily a translation, a radial expansion, a dependence deformation, a local
nonlinear distortion, or a more complex shape change.

\subsection{Map-level diagnostics}
\label{subsec:map-visualization}
Having established that the proposed statistic is competitive as a test, we
next illustrate the additional diagnostic information available before the
final \(L^2\)-aggregation. Specifically, we visualize the estimated map
discrepancy in two dimensions and provide a coordinate-wise decomposition of
the statistic. The purpose of this experiment is not to assess rejection
probabilities, but to demonstrate how the map-based comparison reveals the
directional and geometric structure of distributional differences.

For \(d=2\), we plot the estimated vector field
\(\widehat D_{\varepsilon}(\mathbf u) = \widehat T_{\varepsilon,n}^{\mathrm P}(\mathbf u) - \widehat T_{\varepsilon,m}^{\mathrm Q}(\mathbf u),\ \mathbf u\in\mathbb B_2\).
The background color represents the magnitude \(\|\widehat D_{\varepsilon}(\mathbf u)\|\), while the
arrows indicate the local direction of the map discrepancy over the common
reference domain.

We consider four representative alternatives, all with baseline distribution
\(\mathrm P=N(0,I_2)\). The target distribution \(\mathrm Q\) is respectively
given by a mean shift \(N(0.4\mathbf e_1,I_2)\), a scale change \(N(0,1.4^2I_2)\), a
correlation change \(N(0,\Sigma_{0.6})\), and the nonlinear perturbation
defined by drawing
\(\mathbf Y_0=(Y_{0,1},Y_{0,2})^\top\sim N(0,I_2)\) and setting
\[
\mathbf Y
=
\mathbf Y_0
+
0.2
\begin{pmatrix}
	\sin(5Y_{0,2})\\
	\sin(5Y_{0,1})
\end{pmatrix}.
\]
The distribution of \(\mathbf Y\) is then used as the nonlinear perturbation
alternative.

The sample size is \(n=m=2000\). The reference points are placed on a regular
grid over \(\mathbb B_2\), with additional boundary points used to fill the disk.
The regularization parameter is selected by the same median-distance rule as in
the benchmark comparison, with \(c_\varepsilon=0.2\). Each panel uses its own
color scale for the discrepancy intensity.

Figure~\ref{fig:map-visualization} displays the resulting map discrepancies.
The mean-shift alternative produces an almost constant horizontal displacement
field, matching the perturbation in the first coordinate. The scale-change
alternative yields a radial outward field, reflecting the larger dispersion of
the target distribution. The correlation-change alternative produces a
shearing pattern, revealing the change in dependence structure. The nonlinear
local perturbation generates a spatially varying deformation field over
\(\mathbb B_2\).

To summarize coordinate-wise contributions, write
\(\widehat D_{\varepsilon}
=(\widehat D_{\varepsilon,1},\ldots,\widehat D_{\varepsilon,d})^\top\). The statistic
can be decomposed as
\begin{equation}\label{eq:cor_contribution}
	\mathcal T_{n,m}
	=
	\sum_{k=1}^d
	\mathcal T_{n,m}^{(k)},
	\qquad
	\mathcal T_{n,m}^{(k)}
	=
	\frac{nm}{n+m}
	\int_{\mathbb B_d}
	\widehat D_{\varepsilon,k}(\mathbf u)^2\,d\mathrm U_d(\mathbf u),
\end{equation}
and we report the normalized contribution
\(r_k=\mathcal T_{n,m}^{(k)}/\mathcal T_{n,m}\). These ratios provide a simple
diagnostic for identifying which coordinate directions contribute most to the
overall map discrepancy.

The reported ratios in Figure~\ref{fig:map-visualization} are consistent with
the geometry of the four alternatives. For the mean shift, the first coordinate
dominates, with \(r_1=0.990\) and \(r_2=0.010\). For the isotropic scale change,
the contributions are nearly balanced, with \(r_1=0.486\) and \(r_2=0.514\).
The correlation alternative gives similarly balanced marginal contributions,
but the vector field reveals a non-axis-aligned shearing structure that is not
captured by the coordinate ratios alone. For the nonlinear local perturbation,
the first coordinate contributes more strongly, with \(r_1=0.779\) and
\(r_2=0.221\), while the vector field shows that the discrepancy is spatially
nonlinear.

Thus, the visualization and the coordinate-wise decomposition provide
complementary information. The vector field shows where and in which direction
the two distributions differ, whereas the ratios \(r_k\) summarize the marginal
coordinate contributions to the scalar statistic. This illustrates that, although
the final test statistic is an \(L^2\)-type scalar summary, the EOT-map
comparison retains interpretable directional information that is not directly
available from scalar discrepancy tests such as energy distance, MMD, or
Sinkhorn divergence.

\begin{figure}[!htbp]
	\centering
	\includegraphics[width=0.85\textwidth]{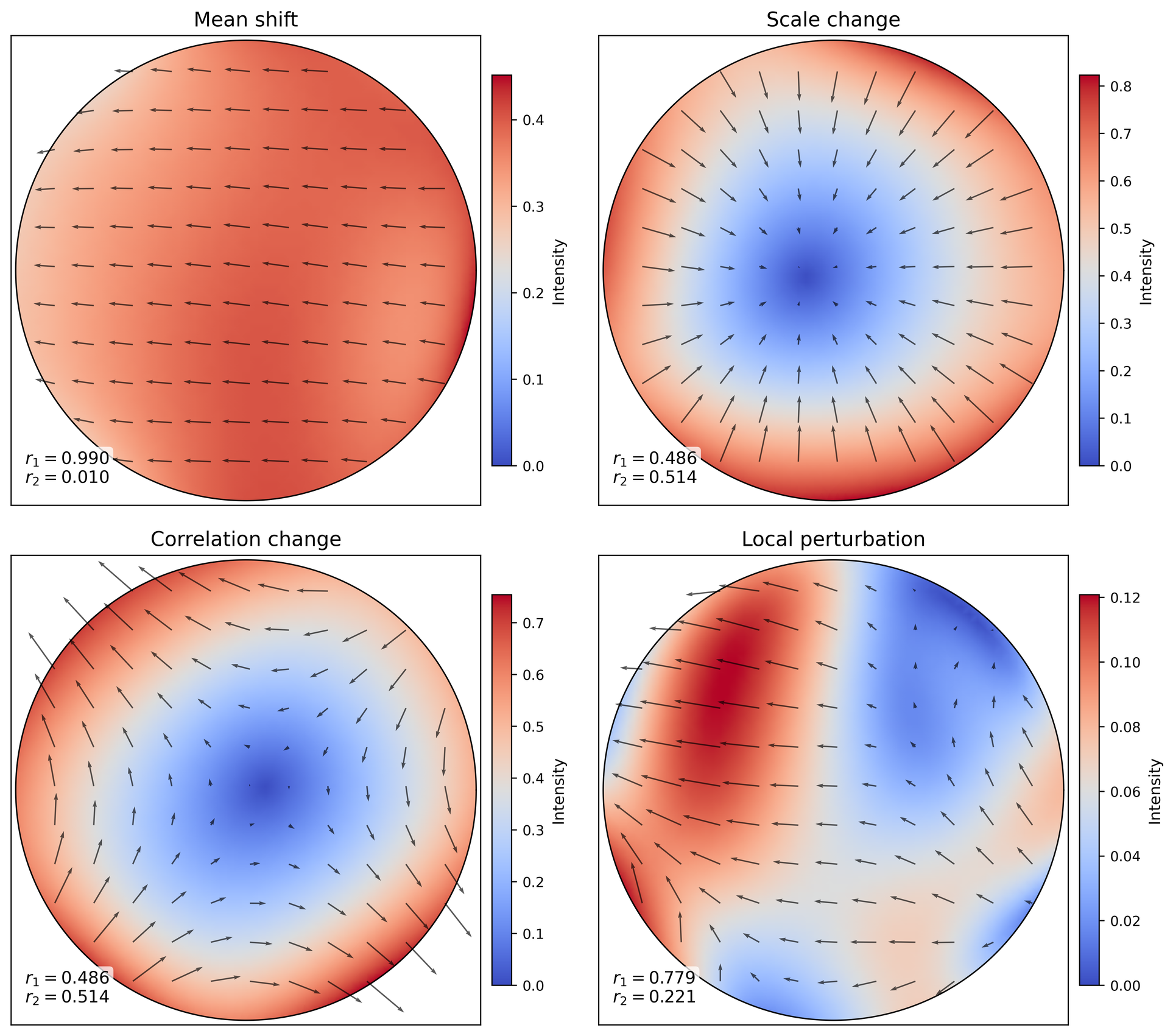}
	\caption{
		Visualization of the estimated EOT-map discrepancy in dimension
		\(d=2\). The background color represents the discrepancy intensity
		\(\|\widehat D_{\varepsilon}(\mathbf u)\|\), and the arrows indicate the direction of the
		estimated map difference
		\(\widehat D_{\varepsilon}(\mathbf u)=\widehat T_{\varepsilon,n}^{\mathrm P}(\mathbf u)
		-\widehat T_{\varepsilon,m}^{\mathrm Q}(\mathbf u)\). Each panel reports the
		coordinate-wise contribution ratios \(r_1\) and \(r_2\). 
	}
	\label{fig:map-visualization}
\end{figure}

Additional sensitivity analyses for the regularization parameter, the reference-sample size, and computation time are reported in Appendix~\ref{app:sensitivity-computation}. Overall, the proposed test is stable over a range of regularization choices, and moderate reference-sample sizes provide reliable approximations while reducing computation time. \mbox{}

\section{Real-data application}
\label{sec:realdata}

We illustrate the proposed EOT-map two-sample test using the Citi Bike trip
history data from Jersey City in March 2026. The data contain individual bike
trips with trip times, station information, geographic coordinates, and rider
type. We examine whether the spatial distributions of trip starting locations
differ between member and casual riders.

Let \(\mathrm P_{\mathrm{mem}}\) and \(\mathrm P_{\mathrm{cas}}\) denote the distributions of
the starting-location coordinates for member and casual riders, respectively.
We test
\(H_0: \mathrm P_{\mathrm{mem}} = \mathrm P_{\mathrm{cas}} \quad \text{versus} \quad H_1: \mathrm P_{\mathrm{mem}} \neq \mathrm P_{\mathrm{cas}}\).
This comparison is motivated by potentially different usage patterns across
rider types: member rides are often associated with regular commuting, whereas
casual rides may be more related to occasional or leisure travel.

We preprocess the trip records by removing observations with missing rider type
or missing starting-location coordinates. The two-dimensional starting-location
coordinates are then standardized, and balanced samples of \(5000\) member rides
and \(5000\) casual rides are retained for the analysis.

We compare the EOT-map test with the same benchmark procedures as in the
simulation study. The implementation follows the settings described in
Section~\ref{sec:numerical-experiments}: the EOT-map test is calibrated by the weighted
bootstrap with \(300\) bootstrap repetitions, while the permutation-calibrated
benchmark tests use \(300\) random permutations.

Applied to the balanced sample of \(5000\) rides per group, all tests produce
extremely small \(p\)-values and reject the null hypothesis. Because the sample
size is large, this result mainly confirms a detectable distributional
difference between member and casual riders. We therefore focus on a subsampling
analysis to assess whether the finding remains stable at smaller sample sizes.

For each \(n=m\in\{500,1000,2000,3000\}\), we draw \(30\) balanced subsamples
from the member and casual groups and repeat all tests. The results are reported
in Table~\ref{tab:citibike-subsampling}. At \(n=m=500\), the EOT-map test
rejects in \(83.3\%\) of the subsamples, with median \(p\)-value \(0.015\). The
MMD and Sinkhorn tests have the same rejection rate, whereas the energy and
\(k\)-NN tests reject in \(76.7\%\) of the subsamples. Once \(n=m=1000\), the
EOT-map, energy, MMD, and Sinkhorn tests reject in all subsamples, and this
pattern persists for \(n=m=2000\) and \(n=m=3000\). The \(k\)-NN test is less
stable, although its rejection rate increases from \(76.7\%\) to \(93.3\%\) as
the subsample size grows. Overall, the subsampling results indicate a persistent
and statistically robust difference between the starting-location distributions
of member and casual riders.

Figure~\ref{fig:citibike-density} visualizes the spatial distributions of the
two groups. Both rider types are concentrated around the downtown and waterfront
areas of Jersey City, but their relative intensities differ across locations.
The signed density-difference map highlights regions where one group has higher
relative start-trip intensity than the other. These localized patterns are
consistent with the formal testing results and suggest that the rejection is not
driven solely by a global location shift.

To summarize the coordinate-wise contributions to the map discrepancy in \eqref{eq:cor_contribution}, we consider the Citi Bike data with \(d=2\). The overall test statistic equals \(74.020\), and the coordinate-wise contributions are \(31.802\) and \(42.218\), yielding contribution ratios of \(0.430\) and \(0.570\), respectively. As the two coordinates represent standardized longitude and latitude, both the east--west and north--south components drive the detected discrepancy, with the second coordinate exhibiting a slightly larger contribution.
Figure~\ref{fig:citibike-eotmap} visualizes the corresponding EOT-map discrepancy
on the source unit disk. For this application, the displayed map difference uses
\(\mathrm P=\mathrm P_{\mathrm{mem}}\) and \(\mathrm Q=\mathrm P_{\mathrm{cas}}\), so the member-minus-casual
direction follows the same \(\mathrm P-\mathrm Q\) convention.

\begin{table}[!htbp]
	\centering
	\begin{tabular}{llccc}
		\toprule
		\(n=m\) & Method & Rejection rate & Median \(p\)-value & Mean \(p\)-value \\
		\midrule
		500  & Energy      & 0.767 & 0.015 & 0.041 \\
		500  & MMD         & 0.833 & 0.007 & 0.031 \\
		500  & Sinkhorn    & 0.833 & 0.012 & 0.045 \\
		500  & \(k\)-NN    & 0.767 & 0.015 & 0.098 \\
		500  & Wasserstein & 0.867 & 0.008 & 0.037 \\
		500  & OT-rank     & 0.633 & 0.023 & 0.075 \\
		500  & EOT-map     & 0.833 & 0.015 & 0.037 \\
		\midrule
		1000 & Energy      & 1.000 & 0.000 & 0.005 \\
		1000 & MMD         & 1.000 & 0.000 & 0.004 \\
		1000 & Sinkhorn    & 1.000 & 0.000 & 0.005 \\
		1000 & \(k\)-NN    & 0.833 & 0.000 & 0.056 \\
		1000 & Wasserstein & 1.000 & 0.000 & 0.004 \\
		1000 & OT-rank     & 0.967 & 0.000 & 0.007 \\
		1000 & EOT-map     & 1.000 & 0.000 & 0.004 \\
		\midrule
		2000 & Energy      & 1.000 & 0.000 & 0.000 \\
		2000 & MMD         & 1.000 & 0.000 & 0.000 \\
		2000 & Sinkhorn    & 1.000 & 0.000 & 0.000 \\
		2000 & \(k\)-NN    & 0.900 & 0.000 & 0.060 \\
		2000 & Wasserstein & 1.000 & 0.000 & 0.000 \\
		2000 & OT-rank     & 1.000 & 0.000 & 0.000 \\
		2000 & EOT-map     & 1.000 & 0.000 & 0.000 \\
		\midrule
		3000 & Energy      & 1.000 & 0.000 & 0.000 \\
		3000 & MMD         & 1.000 & 0.000 & 0.000 \\
		3000 & Sinkhorn    & 1.000 & 0.000 & 0.000 \\
		3000 & \(k\)-NN    & 0.933 & 0.000 & 0.018 \\
		3000 & Wasserstein & 1.000 & 0.000 & 0.000 \\
		3000 & OT-rank     & 1.000 & 0.000 & 0.000 \\
		3000 & EOT-map     & 1.000 & 0.000 & 0.000 \\
		\bottomrule
	\end{tabular}
	\caption{Subsampling results for the Citi Bike data. Each entry is based on
		\(30\) balanced subsamples. Values reported as 0.000 are below numerical resolution.}
	\label{tab:citibike-subsampling}
\end{table}

\begin{figure}[!htbp]
	\centering
	\includegraphics[width=0.32\textwidth]{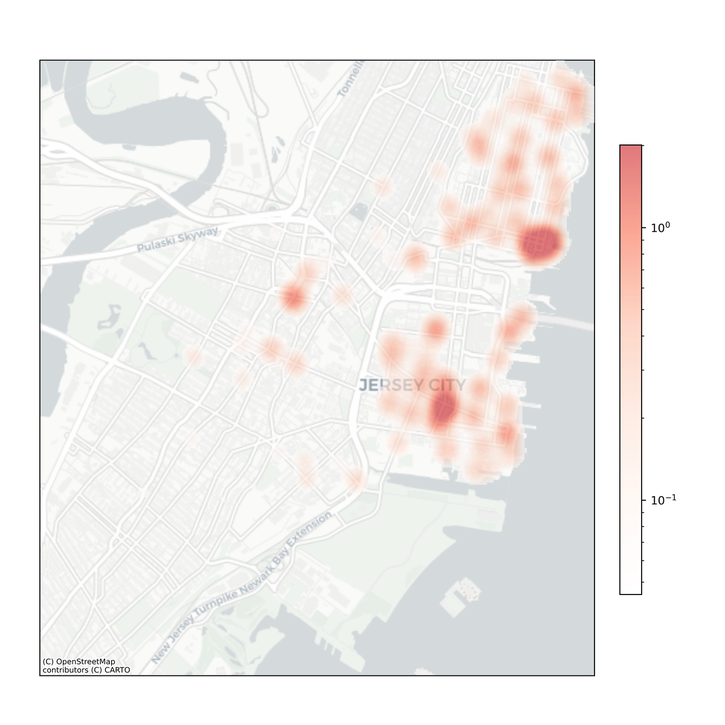}
	\includegraphics[width=0.32\textwidth]{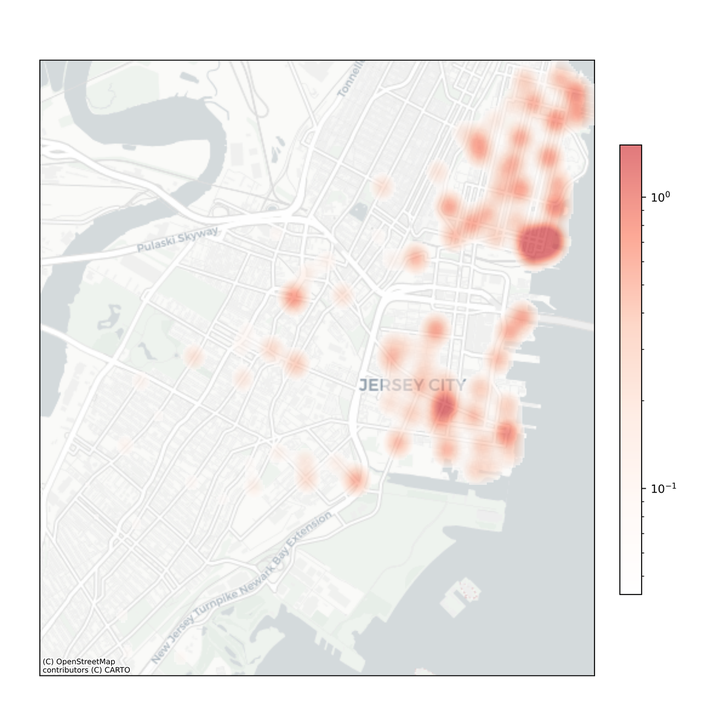}
	\includegraphics[width=0.32\textwidth]{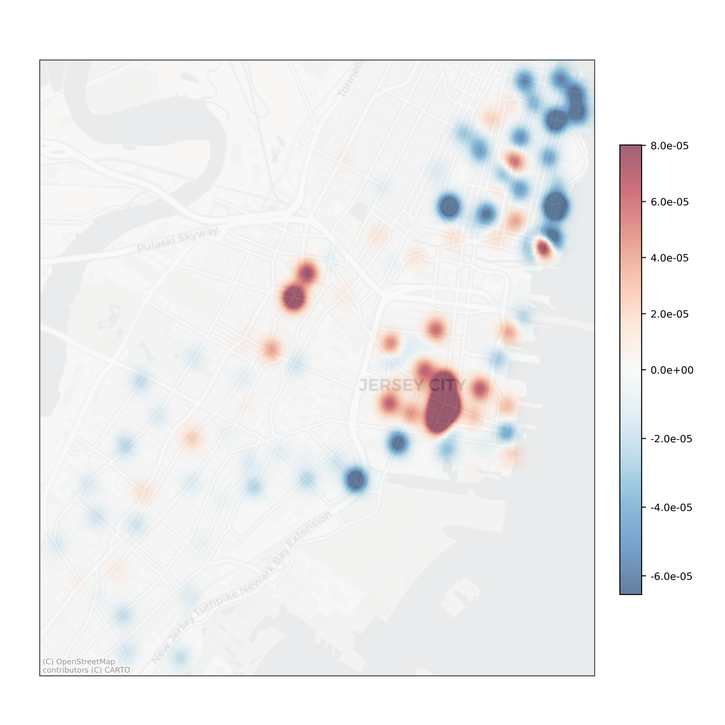}
	\caption{Spatial distributions of Citi Bike starting locations. Left:
		member riders. Middle: casual riders. Right: signed density difference
		between member and casual riders. Red regions indicate higher relative
		start-trip intensity for member riders, whereas blue regions indicate higher
		relative start-trip intensity for casual riders.}
	\label{fig:citibike-density}
\end{figure}

\begin{figure}[!htbp]
	\centering
	\includegraphics[width=0.58\textwidth]{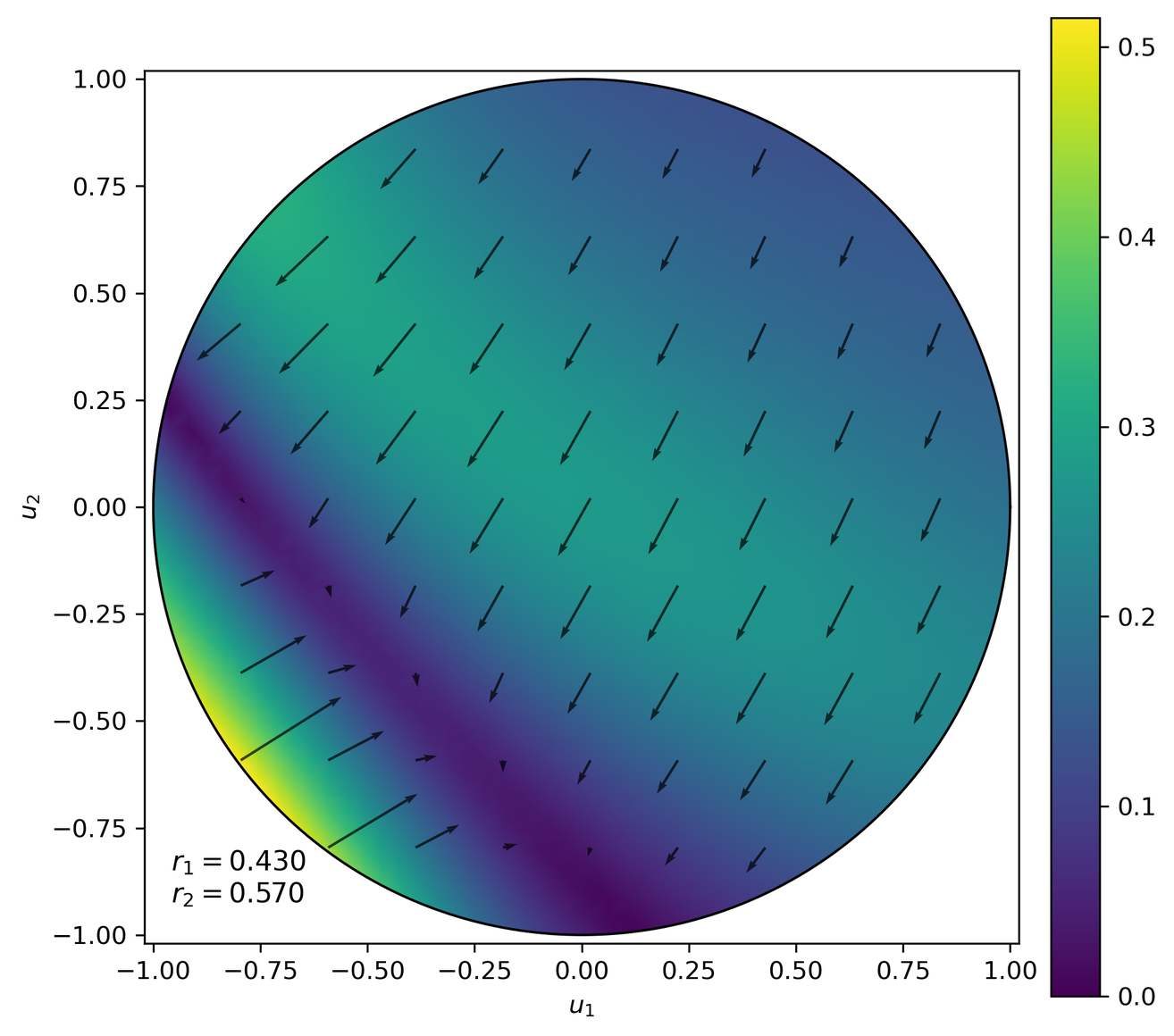}
	\caption{Empirical EOT-map discrepancy on the source unit disk. The color
		represents the magnitude
		\(\|\widehat T_{\varepsilon,\mathrm{mem}}(\mathbf u)
		-\widehat T_{\varepsilon,\mathrm{cas}}(\mathbf u)\|\), and the arrows represent the
		direction of the map difference. The coordinate contribution ratios are
		\(0.430\) and \(0.570\), respectively.}
	\label{fig:citibike-eotmap}
\end{figure}

\section{Conclusion and discussion}
\label{sec:conclusion}

This paper proposed a two-sample testing procedure based on common-reference
entropic optimal transport maps. The method represents each distribution as a
smooth vector-valued map from the fixed reference law \(\mathrm U_d\) on
\(\mathbb B_d\), and compares the two distributions through the squared
\(L^2(\mathbb B_d,\mathrm U_d;\mathbb R^d)\) distance between the resulting
maps.

For fixed \(\varepsilon>0\), we established the identifiability of the
population EOT map, derived the null limit distribution of the proposed
statistic, proved consistency under fixed alternatives, characterized local
asymptotic power, and justified a weighted multiplier bootstrap for
calibration.

Extensive simulations show that the proposed test is well calibrated and competitive
with standard distance-, kernel-, graph-, rank-, and transport-based two-sample
tests. The sensitivity analysis further shows that moderate numbers of
reference points can provide stable empirical size and power, offering a simple
way to control computational cost. Beyond testing, the estimated map difference
provides an interpretable vector-field summary of distributional changes. The
Citi Bike application further illustrates this diagnostic value by detecting
differences between member and casual riders and decomposing the map
discrepancy into coordinate-wise contributions.

\section*{Acknowledgments}
Yiming Ma and Hang Liu contributed equally to this work. Correspondence should be addressed to Weiwei Zhuang (\href{mailto:weizh@ustc.edu.cn}{weizh@ustc.edu.cn}).
	The authors acknowledge the use of ChatGPT for language polishing and 
	code-development assistance during the preparation of this manuscript.
	Hang Liu's
	research was supported by the National Natural Science Foundation of China
	(NSFC) grant 12401372 and University of Science and Technology of China
	grant 2024ycjg13. Weiwei Zhuang's research was supported by the National
	Natural Science Foundation of China (NSFC) grant 72571262.

\bibliographystyle{plainnat}
\bibliography{TST.bib}

\newpage
\appendix

\section{ Proofs of the main results}

\begin{proof}[Proof of Lemma \ref{lem:injectivity}]
	
	We first prove the direct implication. If \(\mathrm P=\mathrm Q\), then the
	two entropic optimal transport problems from \(\mathrm U_d\) to the target
	measure are identical. By uniqueness of the entropic optimal coupling, their
	barycentric projections coincide.  It remains to prove the converse implication. 
	
	For the cost convention
	\[
	c(\mathbf u,\mathbf x)=\frac12\|\mathbf u-\mathbf x\|^2,
	\]
	the entropic transport map satisfies
	\[
	T_{\varepsilon,\mathrm P}(\mathbf{u})
	=
	\mathbf{u}
	-
	\nabla \varphi_{\mathrm P}(\mathbf{u}),
	\]
	and similarly for \(\mathrm Q\).
	
	Therefore,
	\[
	T_{\varepsilon,\mathrm P}
	=
	T_{\varepsilon,\mathrm Q}
	\quad
	\mathrm U_d\text{-a.e.}
	\]
	implies
	\[
	\nabla \varphi_{\mathrm P}
	=
	\nabla \varphi_{\mathrm Q}
	\quad
	\mathrm U_d\text{-a.e.}
	\]
	Since the potentials are smooth on \(\operatorname{int}\mathbb{B}_d\), the
	difference
	\[
	\nabla(\varphi_{\mathrm P}-\varphi_{\mathrm Q})
	\]
	is continuous. Hence it vanishes everywhere on
	\(\operatorname{int}\mathbb{B}_d\). Since \(\operatorname{int}\mathbb{B}_d\) is
	connected, there exists a constant \(C\in\mathbb{R}\) such that
	\[
	\varphi_{\mathrm P}(\mathbf{u})
	=
	\varphi_{\mathrm Q}(\mathbf{u}) + C,
	\qquad
	\mathbf{u}\in\operatorname{int}\mathbb{B}_d.
	\]
	
	By the Schrödinger system,
	\[
	e^{-\varphi_{\mathrm P}(\mathbf{u})/\varepsilon}
	=
	\int_{\mathbb{R}^d}
	\exp\left\{
	-\frac{\|\mathbf{u}-\mathbf{x}\|^2}{2\varepsilon}
	\right\}
	e^{\psi_{\mathrm P}(\mathbf{x})/\varepsilon}
	\,d\mathrm P(\mathbf{x}),
	\]
	and similarly for \(\mathrm Q\). Therefore
	\[
	\int_{\mathbb{R}^d}
	\exp\left\{
	-\frac{\|\mathbf{u}-\mathbf{x}\|^2}{2\varepsilon}
	\right\}
	\,d\Lambda_{\mathrm P}(\mathbf{x})
	=
	e^{-C/\varepsilon}
	\int_{\mathbb{R}^d}
	\exp\left\{
	-\frac{\|\mathbf{u}-\mathbf{x}\|^2}{2\varepsilon}
	\right\}
	\,d\Lambda_{\mathrm Q}(\mathbf{x}),
	\]
	where
	\[
	d\Lambda_{\mathrm P}(\mathbf{x})
	:=
	e^{\psi_{\mathrm P}(\mathbf{x})/\varepsilon}\,d\mathrm P(\mathbf{x}),
	\qquad
	d\Lambda_{\mathrm Q}(\mathbf{x})
	:=
	e^{\psi_{\mathrm Q}(\mathbf{x})/\varepsilon}\,d\mathrm Q(\mathbf{x}).
	\]
	Because \(\mathrm P\) and \(\mathrm Q\) are compactly supported and the
	potentials are finite on the corresponding supports, \(\Lambda_{\mathrm P}\)
	and \(\Lambda_{\mathrm Q}\) are finite measures.
	
	Let
	\[
	\Gamma
	:=
	\Lambda_{\mathrm P}
	-
	e^{-C/\varepsilon}\Lambda_{\mathrm Q}.
	\]
	Then \(\Gamma\) is a finite signed measure and
	\[
	\int_{\mathbb{R}^d}
	\exp\left\{
	-\frac{\|\mathbf{u}-\mathbf{x}\|^2}{2\varepsilon}
	\right\}
	\,d\Gamma(\mathbf{x})
	=
	0,
	\qquad
	\mathbf{u}\in\operatorname{int}\mathbb{B}_d.
	\]
	
	The left-hand side is a real-analytic function of \(\mathbf u\). Since it
	vanishes on the open set \(\operatorname{int}\mathbb{B}_d\), the identity
	theorem for real-analytic functions implies that it vanishes on all of
	\(\mathbb R^d\). Hence
	\[
	G_\varepsilon * \Gamma=0
	\quad\text{on }\mathbb R^d,
	\]
	where
	\[
	G_\varepsilon(\mathbf x)
	=
	\exp\left\{
	-\frac{\|\mathbf x\|^2}{2\varepsilon}
	\right\}.
	\]
	Taking Fourier transforms yields
	\[
	\widehat G_\varepsilon(\boldsymbol\xi)\,
	\widehat\Gamma(\boldsymbol\xi)
	=
	0.
	\]
	Moreover,
	\[
	\widehat G_\varepsilon(\boldsymbol\xi)
	=
	(2\pi\varepsilon)^{d/2}
	\exp\left\{
	-\frac{\varepsilon\|\boldsymbol\xi\|^2}{2}
	\right\}
	>
	0
	\qquad
	\text{for every }\boldsymbol\xi\in\mathbb R^d.
	\]
	Therefore
	\[
	\widehat\Gamma(\boldsymbol\xi)=0
	\quad\text{for all }\boldsymbol\xi\in\mathbb R^d.
	\]
	By uniqueness of the Fourier transform for finite signed measures,
	\[
	\Gamma=0.
	\]
	Hence
	\[
	\Lambda_{\mathrm P}
	=
	e^{-C/\varepsilon}\Lambda_{\mathrm Q}.
	\]
	
	Finally, using the second equation in the Schrödinger system,
	\[
	e^{\psi_{\mathrm P}(\mathbf x)/\varepsilon}
	=
	\left[
	\int_{\mathbb B_d}
	\exp\left\{
	\frac{
		\varphi_{\mathrm P}(\mathbf u)-\frac12\|\mathbf u-\mathbf x\|^2
	}{\varepsilon}
	\right\}
	\,d\mathrm U_d(\mathbf u)
	\right]^{-1},
	\]
	and similarly for \(\mathrm Q\), the identity
	\[
	\varphi_{\mathrm P}
	=
	\varphi_{\mathrm Q}+C
	\]
	implies
	\[
	e^{\psi_{\mathrm P}(\mathbf x)/\varepsilon}
	=
	e^{-C/\varepsilon}
	e^{\psi_{\mathrm Q}(\mathbf x)/\varepsilon}.
	\]
	Substituting this into
	\[
	\Lambda_{\mathrm P}
	=
	e^{-C/\varepsilon}\Lambda_{\mathrm Q}
	\]
	gives
	\[
	e^{-C/\varepsilon}
	e^{\psi_{\mathrm Q}(\mathbf x)/\varepsilon}
	d\mathrm P(\mathbf x)
	=
	e^{-C/\varepsilon}
	e^{\psi_{\mathrm Q}(\mathbf x)/\varepsilon}
	d\mathrm Q(\mathbf x).
	\]
	Since
	\[
	e^{-C/\varepsilon}
	e^{\psi_{\mathrm Q}(\mathbf x)/\varepsilon}>0,
	\]
	we obtain
	\[
	d\mathrm P(\mathbf x)=d\mathrm Q(\mathbf x).
	\]
	Therefore \(\mathrm P=\mathrm Q\).
\end{proof}

\begin{proof}[Proof of Theorem \ref{thm:two-sample-functional-clt}]
	Under \(H_0\), write the common distribution as \(\mathrm P_0\), so that
	\[
	\mathrm P=\mathrm Q=\mathrm P_0
	\]
	and
	\[
	T_{\varepsilon,\mathrm P}
	=
	T_{\varepsilon,\mathrm Q}
	=
	T_{\varepsilon,0}.
	\]
	By the one-sample empirical EOT map CLT,
	\[
	\sqrt n
	\left(
	\widehat T_{\varepsilon,n}^{\mathrm P}
	-
	T_{\varepsilon,0}
	\right)
	\rightsquigarrow
	\mathbb G_{\varepsilon}^{\mathrm P}
	\qquad
	\text{in }C^{s-1}(\mathcal X;\mathbb R^d),
	\]
	and
	\[
	\sqrt m
	\left(
	\widehat T_{\varepsilon,m}^{\mathrm Q}
	-
	T_{\varepsilon,0}
	\right)
	\rightsquigarrow
	\mathbb G_{\varepsilon}^{\mathrm Q}
	\qquad
	\text{in }C^{s-1}(\mathcal X;\mathbb R^d),
	\]
	where \(\mathbb G_{\varepsilon}^{\mathrm P}\) and
	\(\mathbb G_{\varepsilon}^{\mathrm Q}\) are independent copies of the same
	mean-zero Gaussian limit. Independence follows from the independence of the
	two samples.
	
	Therefore,
	\[
	\sqrt{\frac{nm}{n+m}}
	\left(
	\widehat T_{\varepsilon,n}^{\mathrm P}
	-
	\widehat T_{\varepsilon,m}^{\mathrm Q}
	\right)
	\]
	can be decomposed as
	\[
	\sqrt{\frac{m}{n+m}}
	\sqrt n
	\left(
	\widehat T_{\varepsilon,n}^{\mathrm P}
	-
	T_{\varepsilon,0}
	\right)
	-
	\sqrt{\frac{n}{n+m}}
	\sqrt m
	\left(
	\widehat T_{\varepsilon,m}^{\mathrm Q}
	-
	T_{\varepsilon,0}
	\right).
	\]
	Since \(n/(n+m)\to\lambda\), we have
	\[
	\sqrt{\frac{m}{n+m}}\to \sqrt{1-\lambda},
	\qquad
	\sqrt{\frac{n}{n+m}}\to \sqrt{\lambda}.
	\]
	Hence, by Slutsky's theorem,
	\[
	\sqrt{\frac{nm}{n+m}}
	\left(
	\widehat T_{\varepsilon,n}^{\mathrm P}
	-
	\widehat T_{\varepsilon,m}^{\mathrm Q}
	\right)
	\rightsquigarrow
	\sqrt{1-\lambda}\,\mathbb G_{\varepsilon}^{\mathrm P}
	-
	\sqrt{\lambda}\,\mathbb G_{\varepsilon}^{\mathrm Q}
	=: \mathbb G_{\varepsilon}^{(2)}.
	\]
	This proves the convergence in \(C^{s-1}(\mathcal X;\mathbb R^d)\).
	
	Since \(\mathbb B_d\subset\mathcal X\) and \(s\ge1\), the restriction map
	\[
	\iota:
	C^{s-1}(\mathcal X;\mathbb R^d)
	\to
	L^2(\mathbb B_d,\mathrm U_d;\mathbb R^d),
	\qquad
	\iota(f)=f|_{\mathbb B_d},
	\]
	is continuous. Indeed,
	\[
	\|\iota(f)\|_{L^2(\mathbb B_d,\mathrm U_d;\mathbb R^d)}
	\le
	\|f\|_{\infty,\mathcal X}
	\le
	\|f\|_{C^{s-1}(\mathcal X;\mathbb R^d)}.
	\]
	Therefore, the \(L^2\) convergence follows from the continuous mapping theorem.
\end{proof}

\begin{proof}[Proof of Theorem \ref{thm:null-limit}]
	By Theorem~\ref{thm:two-sample-functional-clt},
	\[
	\sqrt{\frac{nm}{n+m}}
	\left(
	\widehat T_{\varepsilon,n}^{\mathrm P}
	-
	\widehat T_{\varepsilon,m}^{\mathrm Q}
	\right)
	\rightsquigarrow
	\mathbb G_{\varepsilon}^{(2)}
	\qquad
	\text{in }L^2(\mathbb B_d,\mathrm U_d;\mathbb R^d).
	\]
	Define the functional
	\[
	\Psi(f)
	=
	\int_{\mathbb B_d}
	\|f(\mathbf u)\|^2
	\,d\mathrm U_d(\mathbf u)
	=
	\|f\|_{L^2(\mathbb B_d,\mathrm U_d;\mathbb R^d)}^2 .
	\]
	The map
	\[
	\Psi:L^2(\mathbb B_d,\mathrm U_d;\mathbb R^d)\to\mathbb R
	\]
	is continuous. Hence, by the continuous mapping theorem,
	\[
	\mathcal T_{n,m}
	=
	\Psi\left(
	\sqrt{\frac{nm}{n+m}}
	\left(
	\widehat T_{\varepsilon,n}^{\mathrm P}
	-
	\widehat T_{\varepsilon,m}^{\mathrm Q}
	\right)
	\right)
	\rightsquigarrow
	\Psi(\mathbb G_{\varepsilon}^{(2)}).
	\]
	This proves the first claim.
	
	For the second claim, since \(\mathbb G_{\varepsilon}^{(2)}\) is a mean-zero
	Gaussian random element in the Hilbert space
	\(L^2(\mathbb B_d,\mathrm U_d;\mathbb R^d)\), its covariance operator
	\(\mathcal K_{\varepsilon}^{(2)}\) is self-adjoint, positive, and trace-class.
	By the Karhunen--Loève expansion,
	\[
	\mathbb G_{\varepsilon}^{(2)}
	=
	\sum_{k=1}^{\infty}
	\sqrt{\omega_k^{(2)}}Z_k e_k^{(2)}
	\quad
	\text{in }L^2(\mathbb B_d,\mathrm U_d;\mathbb R^d),
	\]
	where \((e_k^{(2)})_{k\geq 1}\) is an orthonormal system of eigenfunctions of
	\(\mathcal K_{\varepsilon}^{(2)}\), \((\omega_k^{(2)})_{k\geq 1}\) are the
	corresponding eigenvalues, and \((Z_k)_{k\geq 1}\) are i.i.d.\ standard normal
	variables. Therefore,
	\[
	\|\mathbb G_{\varepsilon}^{(2)}\|_{L^2(\mathbb B_d,\mathrm U_d;\mathbb R^d)}^2
	\overset{d}{=}
	\sum_{k=1}^{\infty}
	\omega_k^{(2)}Z_k^2.
	\]
	This proves the result.
\end{proof}

\begin{proof}[Proof of Theorem \ref{thm:consistency}]
	By the \(L^2(\mathbb B_d,\mathrm U_d;\mathbb R^d)\)-consistency of the
	empirical entropic transport maps,
	\[
	\widehat T_{\varepsilon,n}^{\mathrm P}
	-
	\widehat T_{\varepsilon,m}^{\mathrm Q}
	\overset{\Prob}{\longrightarrow}
	T_{\varepsilon,\mathrm P}
	-
	T_{\varepsilon,\mathrm Q}.
	\]
	The continuous mapping theorem applied to the squared \(L^2\)-norm gives
	\[
	\frac{\mathcal T_{n,m}}{nm/(n+m)}
	\overset{\Prob}{\longrightarrow}
	\mathcal D_\varepsilon(\mathrm P,\mathrm Q).
	\]
	By Lemma~\ref{lem:injectivity}, \(\mathrm P\ne\mathrm Q\) implies
	\[
	\mathcal D_\varepsilon(\mathrm P,\mathrm Q)>0.
	\]
	Since Assumption~\ref{ass:sample-size-balance} implies
	\[
	\frac{nm}{n+m}\to\infty,
	\]
	we obtain
	\[
	\mathcal T_{n,m}\to\infty
	\]
	in probability under \(\mathbb P_{\mathrm P,\mathrm Q}\).
\end{proof}

\begin{proof}[Proof of Theorem \ref{thm:local-power}]
	Under the stated local alternatives, the local empirical-process expansion gives
	\[
	\sqrt{\frac{nm}{n+m}}
	\left[
	(\mathrm P_n-\mathrm P_0)
	-
	(\mathrm Q_m-\mathrm P_0)
	\right]
	\rightsquigarrow
	\sqrt{1-\lambda}\,\mathbb G_{\mathrm P_0}^{\mathrm P}
	-
	\sqrt{\lambda}\,\mathbb G_{\mathrm P_0}^{\mathrm Q}
	+
	(h_{\mathrm P}-h_{\mathrm Q})\mathrm P_0,
	\]
	where \(\mathbb G_{\mathrm P_0}^{\mathrm P}\) and
	\(\mathbb G_{\mathrm P_0}^{\mathrm Q}\) are independent
	\(\mathrm P_0\)-Brownian bridges. Since
	\[
	\Phi(\mathrm R)
	:=
	T_{\varepsilon,\mathrm R}
	\]
	is Hadamard differentiable at \(\mathrm P_0\), the functional delta method yields
	\[
	\sqrt{\frac{nm}{n+m}}
	\left(
	\widehat T_{\varepsilon,n}^{\mathrm P}
	-
	\widehat T_{\varepsilon,m}^{\mathrm Q}
	\right)
	\rightsquigarrow
	\mathbb G_{\varepsilon}^{(2)}
	+
	\dot T_{\varepsilon,\mathrm P_0}
	\bigl[
	(h_{\mathrm P}-h_{\mathrm Q})\mathrm P_0
	\bigr]
	\]
	in \(L^2(\mathbb B_d,\mathrm U_d;\mathbb R^d)\). The convergence of
	\(\mathcal T_{n,m}\) follows from the continuous mapping theorem applied to the
	squared \(L^2\)-norm.
\end{proof}

The proof of Corollary~\ref{cor:local-power} relies on the following two technical lemmas.
\begin{lemma}[Injectivity of the linearized entropic transport map]
	\label{lem:linearized-injectivity}
	Let \(\mathrm P_0\) be supported on a compact set
	\(\mathcal X\subset\mathbb R^d\), and assume that \(\mathrm P_0\) satisfies the
	regularity conditions imposed in the main text. Fix \(\varepsilon>0\), and let
	\((\varphi_0,\psi_0)\) be Schrödinger potentials associated with the entropic
	optimal transport problem between \(\mathrm U_d\) and \(\mathrm P_0\), for the
	quadratic cost
	\[
	c(\mathbf u,\mathbf x)=\frac12\|\mathbf u-\mathbf x\|^2.
	\]
	Define
	\[
	K_0(\mathbf u,\mathbf x)
	=
	\exp\left\{
	\frac{
		\varphi_0(\mathbf u)+\psi_0(\mathbf x)-\frac12\|\mathbf u-\mathbf x\|^2
	}{\varepsilon}
	\right\},
	\qquad
	(\mathbf u,\mathbf x)\in\mathbb B_d\times\mathcal X.
	\]
	Let
	\[
	\gamma=h\,\mathrm P_0,
	\qquad
	\int h\,d\mathrm P_0=0,
	\]
	where \(h\) is bounded and measurable. If
	\[
	\dot T_{\varepsilon,\mathrm P_0}[\gamma]=0
	\qquad
	\text{in }
	L^2(\mathbb B_d,\mathrm U_d;\mathbb R^d),
	\]
	then
	\[
	\gamma=0.
	\]
	Equivalently,
	\[
	\dot T_{\varepsilon,\mathrm P_0}[h\mathrm P_0]=0
	\quad
	\Longrightarrow
	\quad
	h=0
	\quad
	\mathrm P_0\text{-a.s.}
	\]
\end{lemma}

\begin{proof}
	Let \(\mathrm P_t=\mathrm P_0+t\gamma\), for \(t\) sufficiently small, and let
	\((\varphi_t,\psi_t)\) denote Schrödinger potentials associated with
	\((\mathrm U_d,\mathrm P_t)\). By Hadamard differentiability of the
	Schrödinger potentials, there exist first-order perturbations
	\(a_\gamma\) and \(b_\gamma\) such that
	\[
	\varphi_t=\varphi_0+t a_\gamma+o(t),
	\qquad
	\psi_t=\psi_0+t b_\gamma+o(t).
	\]
	The perturbations satisfy the linearized Schrödinger system
	\begin{equation}\label{eq:linearized-schrodinger-a}
		a_\gamma(\mathbf u)
		+
		\int_{\mathcal X} K_0(\mathbf u,\mathbf x)b_\gamma(\mathbf x)\,d\mathrm P_0(\mathbf x)
		=
		-\varepsilon
		\int_{\mathcal X} K_0(\mathbf u,\mathbf x)\,d\gamma(\mathbf x),
	\end{equation}
	and
	\begin{equation}\label{eq:linearized-schrodinger-b}
		b_\gamma(\mathbf x)
		+
		\int_{\mathbb B_d}K_0(\mathbf u,\mathbf x)a_\gamma(\mathbf u)\,d\mathrm U_d(\mathbf u)
		=
		0.
	\end{equation}
	
	Moreover,
	\begin{equation}\label{eq:normalization}
		\int_{\mathcal X}K_0(\mathbf u,\mathbf x)\,d\mathrm P_0(\mathbf x)=1,
		\qquad
		\int_{\mathbb B_d}K_0(\mathbf u,\mathbf x)\,d\mathrm U_d(\mathbf u)=1.
	\end{equation}

	For the cost convention \(c(\mathbf u,\mathbf x)=\frac12\|\mathbf u-\mathbf x\|^2\), the entropic map satisfies
	\[
	T_{\varepsilon,\mathrm P}(\mathbf u)=\mathbf u-\nabla\varphi_{\mathrm P}(\mathbf u).
	\]
	Therefore,
	\[
	\dot T_{\varepsilon,\mathrm P_0}[\gamma](\mathbf u)
	=
	-\nabla a_\gamma(\mathbf u).
	\]
	If \(\dot T_{\varepsilon,\mathrm P_0}[\gamma]=0\) in
	\(L^2(\mathbb B_d,\mathrm U_d;\mathbb R^d)\), then
	\[
	\nabla a_\gamma=0
	\qquad
	\mathrm U_d\text{-a.e. on }\mathbb B_d.
	\]
	Since \(a_\gamma\) belongs to a sufficiently smooth function space, in
	particular \(a_\gamma\in C^1(\operatorname{int}\mathbb B_d)\), continuity gives
	\[
	\nabla a_\gamma(\mathbf u)=0
	\qquad
	\text{for all }\mathbf u\in\operatorname{int}(\mathbb B_d).
	\]
	Because \(\operatorname{int}(\mathbb B_d)\) is connected, there exists
	\(C\in\mathbb R\) such that
	\[
	a_\gamma(\mathbf u)=C
	\qquad
	\mathbf u\in\operatorname{int}(\mathbb B_d).
	\]
	Since \(\mathrm U_d(\partial\mathbb B_d)=0\), this implies
	\(a_\gamma=C\) \(\mathrm U_d\)-a.e. on \(\mathbb B_d\).
	
	Using \eqref{eq:linearized-schrodinger-b} and the second normalization in \eqref{eq:normalization}, we get
	\[
	b_\gamma(\mathbf x)=-C
	\qquad
	\mathrm P_0\text{-a.e.}
	\]
	Substituting \(a_\gamma=C\) and \(b_\gamma=-C\) into \eqref{eq:linearized-schrodinger-a}, and using the first
	normalization in \eqref{eq:normalization}, yields
	
	\[
	\int_{\mathcal X}K_0(\mathbf u,\mathbf x)\,d\gamma(\mathbf x)=0
	\qquad
	\text{for all }\mathbf u\in\operatorname{int}(\mathbb B_d).
	\]
	
	Since
	\[
	K_0(\mathbf u,\mathbf x)
	=
	\exp\left\{\frac{\varphi_0(\mathbf u)}{\varepsilon}\right\}
	\exp\left\{
	\frac{\psi_0(\mathbf x)-\frac12\|\mathbf u-\mathbf x\|^2}{\varepsilon}
	\right\},
	\]
	and \(\exp\{\varphi_0(\mathbf u)/\varepsilon\}>0\), we have
	\[
	\int_{\mathcal X}
	\exp\left\{
	-\frac{\|\mathbf u-\mathbf x\|^2}{2\varepsilon}
	\right\}
	\exp\left\{
	\frac{\psi_0(\mathbf x)}{\varepsilon}
	\right\}
	\,d\gamma(\mathbf x)
	=
	0
	\]
	for all \(\mathbf u\in\operatorname{int}(\mathbb B_d)\).
	
	Define the finite signed measure
	\[
	d\widetilde\gamma(\mathbf x)
	=
	\exp\left\{
	\frac{\psi_0(\mathbf x)}{\varepsilon}
	\right\}
	\,d\gamma(\mathbf x).
	\]
	This measure is finite because \(\gamma=h\,\mathrm P_0\), \(h\) is bounded, and
	\(\exp\{\psi_0/\varepsilon\}\) is bounded on the compact set \(\mathcal X\).
	Then
	\[
	F(\mathbf u)
	:=
	\int_{\mathcal X}
	\exp\left\{
	-\frac{\|\mathbf u-\mathbf x\|^2}{2\varepsilon}
	\right\}
	\,d\widetilde\gamma(\mathbf x)
	=
	0
	\]
	for all \(\mathbf u\in\operatorname{int}(\mathbb B_d)\).
	
	The function \(F\) is real analytic on \(\mathbb R^d\). Since it vanishes on the
	nonempty open set \(\operatorname{int}(\mathbb B_d)\), the identity theorem for
	real analytic functions gives
	\[
	F(\mathbf u)=0
	\qquad
	\text{for all }\mathbf u\in\mathbb R^d.
	\]
	Equivalently,
	\[
	G_\varepsilon*\widetilde\gamma=0
	\qquad
	\text{on }\mathbb R^d,
	\]
	where
	\[
	G_\varepsilon(z)
	=
	\exp\left\{
	-\frac{\|z\|^2}{2\varepsilon}
	\right\}.
	\]
	Taking Fourier transforms,
	\[
	\widehat G_\varepsilon(\boldsymbol\xi)\widehat{\widetilde\gamma}(\boldsymbol\xi)=0.
	\]
	Since
	\[
	\widehat G_\varepsilon(\boldsymbol\xi)
	=
	(2\pi\varepsilon)^{d/2}
	\exp\left\{
	-\frac{\varepsilon\|\boldsymbol\xi\|^2}{2}
	\right\}>0
	\qquad
	\text{for all }\boldsymbol\xi\in\mathbb R^d,
	\]
	we obtain
	\[
	\widehat{\widetilde\gamma}(\boldsymbol\xi)=0
	\qquad
	\text{for all }\boldsymbol\xi\in\mathbb R^d.
	\]
	By uniqueness of Fourier transforms for finite signed measures,
	\[
	\widetilde\gamma=0.
	\]
	Finally, the density
	\[
	\exp\left\{
	\frac{\psi_0(\mathbf x)}{\varepsilon}
	\right\}
	\]
	is strictly positive, so \(\gamma=0\). This proves the claim.
\end{proof}
\begin{lemma}[Nondegeneracy of nonzero local drift directions]
	\label{lem:drift-nondegeneracy}
	Let
	\[
	\mathcal H
	:=
	L^2(\mathbb B_d,\mathrm U_d;\mathbb R^d).
	\]  
	By the Hadamard differentiability assumption used in
	Theorem~\ref{thm:local-power}, the derivative of the entropic transport map at
	\(\mathrm P_0\) induces a continuous linear operator
	\[
	\dot\Phi_{\varepsilon,0}:
	L^2_0(\mathrm P_0)
	\to
	\mathcal H,
	\qquad
	\dot\Phi_{\varepsilon,0}h
	:=
	\dot T_{\varepsilon,\mathrm P_0}[h\,\mathrm P_0],
	\]
	where
	\[
	L^2_0(\mathrm P_0)
	:=
	\left\{
	h\in L^2(\mathrm P_0):
	\int h\,d\mathrm P_0=0
	\right\}.
	\]
	Let \(G:=\mathbb G_{\varepsilon}^{(2)}\). If
	\[
	\Delta
	=
	\dot\Phi_{\varepsilon,0}(h_{\mathrm P}-h_{\mathrm Q})
	\ne0
	\qquad
	\text{in }\mathcal H,
	\]
	then
	\[
	\Var
	\left(
	\left\langle
	G,\Delta
	\right\rangle_{\mathcal H}
	\right)
	>0.
	\]
	Equivalently,
	\[
	\left\langle
	\mathcal K_\varepsilon^{(2)}\Delta,\Delta
	\right\rangle_{\mathcal H}
	>0.
	\]
\end{lemma}

\begin{proof}
	Let
	\[
	r:=h_{\mathrm P}-h_{\mathrm Q},
	\qquad
	\Delta=\dot\Phi_{\varepsilon,0} r.
	\]
	By the Hadamard differentiability assumption used in
	Theorem~\ref{thm:local-power},
	\[
	\dot\Phi_{\varepsilon,0}:
	L^2_0(\mathrm P_0)\to\mathcal H
	\]
	is a continuous linear operator. Hence it admits a Hilbert-space adjoint
	\[
	\dot\Phi_{\varepsilon,0}^{\,*}:
	\mathcal H\to L^2_0(\mathrm P_0).
	\]
	For any \(f\in\mathcal H\), the representation of \(G\) and the definition of
	the adjoint imply
	\[
	\left\langle G,f\right\rangle_{\mathcal H}
	=
	\sqrt{1-\lambda}\,
	\mathbb G_{\mathrm P_0}^{\mathrm P}
	\left(
	\dot\Phi_{\varepsilon,0}^{\,*} f
	\right)
	-
	\sqrt{\lambda}\,
	\mathbb G_{\mathrm P_0}^{\mathrm Q}
	\left(
	\dot\Phi_{\varepsilon,0}^{\,*} f
	\right).
	\]
	Since \(\mathbb G_{\mathrm P_0}^{\mathrm P}\) and
	\(\mathbb G_{\mathrm P_0}^{\mathrm Q}\) are independent isonormal Gaussian
	processes on \(L^2_0(\mathrm P_0)\), it follows that
	\[
	\Var
	\left(
	\left\langle G,f\right\rangle_{\mathcal H}
	\right)
	=
	\left\|
	\dot\Phi_{\varepsilon,0}^{\,*} f
	\right\|_{L^2(\mathrm P_0)}^2.
	\]
	Taking \(f=\Delta=\dot\Phi_{\varepsilon,0}r\), we obtain
	\[
	\Var
	\left(
	\left\langle G,\Delta\right\rangle_{\mathcal H}
	\right)
	=
	\left\|
	\dot\Phi_{\varepsilon,0}^{\,*}
	\dot\Phi_{\varepsilon,0}r
	\right\|_{L^2(\mathrm P_0)}^2.
	\]
	If this variance were zero, then
	\[
	\dot\Phi_{\varepsilon,0}^{\,*}
	\dot\Phi_{\varepsilon,0}r=0.
	\]
	Taking the inner product with \(r\) in \(L^2(\mathrm P_0)\) gives
	\[
	0
	=
	\left\langle
	\dot\Phi_{\varepsilon,0}^{\,*}
	\dot\Phi_{\varepsilon,0}r,r
	\right\rangle_{L^2(\mathrm P_0)}
	=
	\left\|
	\dot\Phi_{\varepsilon,0}r
	\right\|_{\mathcal H}^2
	=
	\|\Delta\|_{\mathcal H}^2,
	\]
	which contradicts \(\Delta\ne0\). Therefore
	\[
	\Var
	\left(
	\left\langle G,\Delta\right\rangle_{\mathcal H}
	\right)>0.
	\]
	Finally, by the definition of the covariance operator of \(G\),
	\[
	\Var
	\left(
	\left\langle G,\Delta\right\rangle_{\mathcal H}
	\right)
	=
	\left\langle
	\mathcal K_\varepsilon^{(2)}\Delta,\Delta
	\right\rangle_{\mathcal H}.
	\]
	Hence
	\[
	\left\langle
	\mathcal K_\varepsilon^{(2)}\Delta,\Delta
	\right\rangle_{\mathcal H}>0.
	\]
\end{proof}

\begin{proof}[Proof of Corollary \ref{cor:local-power}] 

	By Theorem~\ref{thm:local-power},
	\[
	\mathcal T_{n,m}
	\rightsquigarrow
	\|G+\Delta\|_{\mathcal H}^2
	\]
	under the stated local alternatives. Hence the asymptotic rejection probability
	is
	\[
	\beta(h_{\mathrm P},h_{\mathrm Q})
	=
	\Prob
	\left(
	\|G+\Delta\|_{\mathcal H}^2>c_\alpha
	\right).
	\]
	
	Let
	\[
	B_\alpha
	:=
	\left\{
	f\in\mathcal H:
	\|f\|_{\mathcal H}^2\le c_\alpha
	\right\}.
	\]
	The set \(B_\alpha\) is closed, convex, and symmetric. Since \(G\) is a centered
	Gaussian random element in \(\mathcal H\), Anderson's inequality gives
	\[
	\Prob(G+\Delta\in B_\alpha)
	\le
	\Prob(G\in B_\alpha).
	\]
	Taking complements yields
	\[
	\beta(h_{\mathrm P},h_{\mathrm Q})
	=
	\Prob
	\left(
	\|G+\Delta\|_{\mathcal H}^2>c_\alpha
	\right)
	\ge
	\Prob
	\left(
	\|G\|_{\mathcal H}^2>c_\alpha
	\right)
	=
	\alpha .
	\]
	
	It remains to prove the strict statement. By definition,
	\[
	\Delta
	=
	\dot T_{\varepsilon,\mathrm P_0}
	\left[
	(h_{\mathrm P}-h_{\mathrm Q})\mathrm P_0
	\right].
	\]
	If
	\[
	h_{\mathrm P}\ne h_{\mathrm Q}
	\quad
	\mathrm P_0\text{-a.s.},
	\]
	then
	\[
	(h_{\mathrm P}-h_{\mathrm Q})\mathrm P_0
	\ne 0
	\]
	is a signed measure. Since
	\[
	\int (h_{\mathrm P}-h_{\mathrm Q})\,d\mathrm P_0=0,
	\]
	Lemma~\ref{lem:linearized-injectivity} implies
	\[
	\Delta\ne0
	\qquad
	\text{in }\mathcal H .
	\]
	
	Therefore, \(\Delta\) is a nonzero shift in the range of
	\(\dot T_{\varepsilon,\mathrm P_0}\). Under the corresponding strict Anderson
	inequality condition (see \citep{lewandowski1995anderson}), by Lemma~\ref{lem:drift-nondegeneracy},
	\[
	\Prob(G+\Delta\in B_\alpha)
	<
	\Prob(G\in B_\alpha).
	\]
	Taking complements gives
	\[
	\beta(h_{\mathrm P},h_{\mathrm Q})
	=
	\Prob
	\left(
	\|G+\Delta\|_{\mathcal H}^2>c_\alpha
	\right)
	>
	\Prob
	\left(
	\|G\|_{\mathcal H}^2>c_\alpha
	\right)
	=
	\alpha .
	\]
	This proves the claim.
\end{proof}
\begin{proof}[Proof of Theorem \ref{thm:weighted-bootstrap}]
	Under \(H_0\), write the common distribution as \(\mathrm P_0\). By the weighted
	bootstrap empirical-process CLT, conditionally on the data,
	\[
	\sqrt n(\mathrm P_n^*-\mathrm P_n)
	\rightsquigarrow_{\Prob}
	\mathbb G_{\mathrm P_0}^{\mathrm P},
	\qquad
	\sqrt m(\mathrm Q_m^*-\mathrm Q_m)
	\rightsquigarrow_{\Prob}
	\mathbb G_{\mathrm P_0}^{\mathrm Q},
	\]
	in the corresponding empirical-process tangent spaces, where
	\(\mathbb G_{\mathrm P_0}^{\mathrm P}\) and
	\(\mathbb G_{\mathrm P_0}^{\mathrm Q}\) are independent
	\(\mathrm P_0\)-Brownian bridges. The independence follows from the independence
	of the two samples and of the two multiplier sequences.
	
	Since
	\[
	\Phi(\mathrm R)=T_{\varepsilon,\mathrm R}
	\]
	is Hadamard differentiable at \(\mathrm P_0\), the conditional functional delta
	method yields
	\[
	\sqrt n
	\left(
	\widehat T_{\varepsilon,n}^{\mathrm P,*}
	-
	\widehat T_{\varepsilon,n}^{\mathrm P}
	\right)
	\rightsquigarrow_{\Prob}
	\dot T_{\varepsilon,\mathrm P_0}
	\left[
	\mathbb G_{\mathrm P_0}^{\mathrm P}
	\right],
	\]
	and
	\[
	\sqrt m
	\left(
	\widehat T_{\varepsilon,m}^{\mathrm Q,*}
	-
	\widehat T_{\varepsilon,m}^{\mathrm Q}
	\right)
	\rightsquigarrow_{\Prob}
	\dot T_{\varepsilon,\mathrm P_0}
	\left[
	\mathbb G_{\mathrm P_0}^{\mathrm Q}
	\right]
	\]
	in \(L^2(\mathbb B_d,\mathrm U_d;\mathbb R^d)\).
	
	By definition,
	\[
	\begin{aligned}
		\mathbb Z_{n,m}^*
		=
		\sqrt{\frac{m}{n+m}}\,
		\sqrt n
		\left(
		\widehat T_{\varepsilon,n}^{\mathrm P,*}
		-
		\widehat T_{\varepsilon,n}^{\mathrm P}
		\right)
		-
		\sqrt{\frac{n}{n+m}}\,
		\sqrt m
		\left(
		\widehat T_{\varepsilon,m}^{\mathrm Q,*}
		-
		\widehat T_{\varepsilon,m}^{\mathrm Q}
		\right).
	\end{aligned}
	\]
	By Assumption~\ref{ass:sample-size-balance},
	\[
	\sqrt{\frac{m}{n+m}}
	\to
	\sqrt{1-\lambda},
	\qquad
	\sqrt{\frac{n}{n+m}}
	\to
	\sqrt{\lambda}.
	\]
	Thus, by conditional Slutsky's theorem,
	\[
	\mathbb Z_{n,m}^*
	\rightsquigarrow_{\Prob}
	\sqrt{1-\lambda}\,
	\dot T_{\varepsilon,\mathrm P_0}
	\left[
	\mathbb G_{\mathrm P_0}^{\mathrm P}
	\right]
	-
	\sqrt{\lambda}\,
	\dot T_{\varepsilon,\mathrm P_0}
	\left[
	\mathbb G_{\mathrm P_0}^{\mathrm Q}
	\right].
	\]
	The right-hand side is precisely
	\[
	\mathbb G_\varepsilon^{(2)}.
	\]
	Therefore,
	\[
	\mathbb Z_{n,m}^*
	\rightsquigarrow_{\Prob}
	\mathbb G_\varepsilon^{(2)}
	\qquad
	\text{in }
	L^2(\mathbb B_d,\mathrm U_d;\mathbb R^d).
	\]
	
	Finally, the map
	\[
	f
	\mapsto
	\int_{\mathbb B_d}
	\|f(\mathbf u)\|^2\,d\mathrm U_d(\mathbf u)
	\]
	is continuous on \(L^2(\mathbb B_d,\mathrm U_d;\mathbb R^d)\). The conditional
	continuous mapping theorem gives
	\[
	\mathcal T_{n,m}^*
	\rightsquigarrow_{\Prob}
	\int_{\mathbb B_d}
	\left\|
	\mathbb G_\varepsilon^{(2)}(\mathbf u)
	\right\|^2
	\,d\mathrm U_d(\mathbf u).
	\]
\end{proof}

The proof of Corollary~\ref{cor:bootstrap-size-power} relies on two auxiliary results.
\begin{lemma}[Bootstrap critical values under fixed alternatives]
	\label{lem:bootstrap-tightness-fixed-alt}
	Suppose that Assumptions~\ref{ass:compact}, \ref{ass:fixed-epsilon}, and \ref{ass:sample-size-balance} hold for a fixed pair \((\mathrm P,\mathrm Q)\) with \(\mathrm P\ne \mathrm Q\). Assume further that the multiplier variables satisfy the moment conditions stated in Section~\ref{subsec:weighted-bootstrap}.
	Then, under \(\mathbb P_{\mathrm P,\mathrm Q}\),
	\[
	\mathcal T_{n,m}^{*}
	=
	O_{\mathbb P}^{*}(1)
	\]
	in \(\mathbb P_{\mathrm P,\mathrm Q}\)-probability. Consequently, for every
	\(\alpha\in(0,1)\), the conditional \((1-\alpha)\)-quantile
	\(c_{n,m,1-\alpha}^{*}\) satisfies
	\[
	c_{n,m,1-\alpha}^{*}
	=
	O_{\mathbb P_{\mathrm P,\mathrm Q}}(1).
	\]
\end{lemma}

\begin{proof}
	By the weighted bootstrap empirical-process tightness and the bootstrap
	functional delta method applied to
	\[
	\Phi:\mathrm R\mapsto T_{\varepsilon,\mathrm R}
	\]
	at \(\mathrm P\) and at \(\mathrm Q\), respectively, we have
	\[
	\sqrt n
	\left(
	\widehat T_{\varepsilon,n}^{\mathrm P,*}
	-
	\widehat T_{\varepsilon,n}^{\mathrm P}
	\right)
	=
	O_{\mathbb P}^{*}(1),
	\qquad
	\sqrt m
	\left(
	\widehat T_{\varepsilon,m}^{\mathrm Q,*}
	-
	\widehat T_{\varepsilon,m}^{\mathrm Q}
	\right)
	=
	O_{\mathbb P}^{*}(1)
	\]
	in
	\[
	\mathcal H:=L^2(\mathbb B_d,\mathrm U_d;\mathbb R^d),
	\]
	in \(\mathbb P_{\mathrm P,\mathrm Q}\)-probability.
	
	By definition,
	\[
	\begin{aligned}
		\mathbb Z_{n,m}^{*}
		=
		\sqrt{\frac{m}{n+m}}\,
		\sqrt n
		\left(
		\widehat T_{\varepsilon,n}^{\mathrm P,*}
		-
		\widehat T_{\varepsilon,n}^{\mathrm P}
		\right)
		-
		\sqrt{\frac{n}{n+m}}\,
		\sqrt m
		\left(
		\widehat T_{\varepsilon,m}^{\mathrm Q,*}
		-
		\widehat T_{\varepsilon,m}^{\mathrm Q}
		\right).
	\end{aligned}
	\]
	By Assumption~\ref{ass:sample-size-balance}, the two deterministic
	coefficients are bounded. Therefore
	\[
	\mathbb Z_{n,m}^{*}
	=
	O_{\mathbb P}^{*}(1)
	\qquad
	\text{in }\mathcal H,
	\]
	in \(\mathbb P_{\mathrm P,\mathrm Q}\)-probability. Since
	\[
	\mathcal T_{n,m}^{*}
	=
	\left\|
	\mathbb Z_{n,m}^{*}
	\right\|_{\mathcal H}^{2},
	\]
	the continuity of the squared norm yields
	\[
	\mathcal T_{n,m}^{*}
	=
	O_{\mathbb P}^{*}(1)
	\]
	in \(\mathbb P_{\mathrm P,\mathrm Q}\)-probability.
	
	It remains to translate conditional tightness of
	\(\mathcal T_{n,m}^{*}\) into tightness of the conditional quantiles.
	By conditional tightness, for every \(\eta>0\), there exists
	\(M<\infty\) such that, for all sufficiently large \(n,m\),
	\[
	\mathbb P_{\mathrm P,\mathrm Q}
	\left\{
	\mathbb P^{*}
	\left(
	\mathcal T_{n,m}^{*}>M
	\right)
	>\alpha
	\right\}
	<\eta .
	\]
	On the event
	\[
	\mathbb P^{*}
	\left(
	\mathcal T_{n,m}^{*}>M
	\right)
	\le \alpha,
	\]
	the conditional \((1-\alpha)\)-quantile satisfies
	\[
	c_{n,m,1-\alpha}^{*}
	\le M.
	\]
	Therefore,
	\[
	\mathbb P_{\mathrm P,\mathrm Q}
	\left(
	c_{n,m,1-\alpha}^{*}>M
	\right)
	<\eta
	\]
	for all sufficiently large \(n,m\). This proves
	\[
	c_{n,m,1-\alpha}^{*}
	=
	O_{\mathbb P_{\mathrm P,\mathrm Q}}(1).
	\]
\end{proof}

\begin{proof}[Proof of Corollary \ref{cor:bootstrap-size-power}]
	We first show that \(L_0\) has no atoms. Since \(\mathrm P_0\) is not a Dirac
	measure, there exists a Borel set \(A\) such that
	\[
	0<\mathrm P_0(A)<1.
	\]
	Define
	\[
	h_A
	:=
	\mathbf 1_A-\mathrm P_0(A).
	\]
	Then \(h_A\) is bounded,
	\[
	\int h_A\,d\mathrm P_0=0,
	\]
	and
	\[
	h_A\ne0
	\qquad
	\mathrm P_0\text{-a.s.}
	\]
	By Lemma~\ref{lem:linearized-injectivity}, 
	\[
	\dot T_{\varepsilon,\mathrm P_0}
	[h_A\mathrm P_0]
	\ne0
	\qquad
	\text{in }\mathcal H.
	\]
	Hence the linearized EOT map
	\[
	\dot\Phi_{\varepsilon,0}h
	:=
	\dot T_{\varepsilon,\mathrm P_0}[h\mathrm P_0]
	\]
	is not the zero operator.
	
	Under \(H_0\), we may write
	\[
	\mathbb G_{\varepsilon}^{(2)}
	=
	\dot\Phi_{\varepsilon,0}[\mathbb W],
	\]
	where
	\[
	\mathbb W
	:=
	\sqrt{1-\lambda}\,
	\mathbb G_{\mathrm P_0}^{\mathrm P}
	-
	\sqrt{\lambda}\,
	\mathbb G_{\mathrm P_0}^{\mathrm Q}.
	\]
	Since \(\mathbb G_{\mathrm P_0}^{\mathrm P}\) and
	\(\mathbb G_{\mathrm P_0}^{\mathrm Q}\) are independent
	\(\mathrm P_0\)-Brownian bridges, \(\mathbb W\) is again a
	\(\mathrm P_0\)-Brownian bridge. Because
	\(\dot\Phi_{\varepsilon,0}\) is not the zero operator, there exists
	\(f\in\mathcal H\) such that
	\[
	\dot\Phi_{\varepsilon,0}^{\,*}f
	\ne0
	\qquad
	\text{in }L_0^2(\mathrm P_0).
	\]
	Therefore
	\[
	\begin{aligned}
		\Var
		\left(
		\left\langle
		\mathbb G_{\varepsilon}^{(2)},f
		\right\rangle_{\mathcal H}
		\right)
		&=
		\Var
		\left(
		\mathbb W
		\left(
		\dot\Phi_{\varepsilon,0}^{\,*}f
		\right)
		\right)\\
		&=
		\left\|
		\dot\Phi_{\varepsilon,0}^{\,*}f
		\right\|_{L^2(\mathrm P_0)}^2
		>0.
	\end{aligned}
	\]
	Thus \(\mathbb G_{\varepsilon}^{(2)}\) is nondegenerate, equivalently,
	\(
	\mathcal K_{\varepsilon}^{(2)}\ne0.
	\)
	
By Theorem~\ref{thm:null-limit}, the null limit admits the spectral representation
\(
L_0
:=
\|\mathbb G_{\varepsilon}^{(2)}\|_{\mathcal H}^2
\overset{d}{=}
\sum_{k=1}^{\infty}\omega_k^{(2)} Z_k^2,
\)
Since \(\mathcal K_{\varepsilon}^{(2)}\neq 0\), there exists \(k_0\ge 1\) such that
\(\omega_{k_0}^{(2)}>0\). Hence
\[
L_0
\overset{d}{=}
\omega_{k_0}^{(2)} Z_{k_0}^2 + R,
\qquad
R
:=
\sum_{k\neq k_0}\omega_k^{(2)} Z_k^2,
\]
where \(R\) is independent of \(Z_{k_0}\). Since
\(\omega_{k_0}^{(2)}Z_{k_0}^2\) has a continuous distribution, for every
\(t\in\mathbb R\),
\[
\mathbb P(L_0=t)
=
\mathbb E\left[
\mathbb P\left(
\omega_{k_0}^{(2)}Z_{k_0}^2=t-R
\,\middle|\, R
\right)
\right]
=0.
\]
Therefore \(L_0\) has no atoms. In particular, \(c_\alpha\) is a continuity
point of the distribution of \(L_0\), and
\(\mathbb P(L_0>c_\alpha)=\alpha\).

	By Theorem~\ref{thm:weighted-bootstrap},
	\[
	\mathcal T_{n,m}^{*}
	\rightsquigarrow_{\Prob}
	L_0
	\]
	conditionally on the data. Since \(c_\alpha\) is a continuity point of the
	distribution of \(L_0\), the convergence of conditional distributions implies
	\[
	c_{n,m,1-\alpha}^{*}
	\overset{\Prob}{\to}
	c_\alpha .
	\]
	
	Under \(H_0\), Theorem~\ref{thm:null-limit} gives
	\[
	\mathcal T_{n,m}
	\rightsquigarrow
	L_0 .
	\]
	Hence, by Slutsky's theorem,
	\[
	\Prob_{H_0}
	\left(
	\phi_{n,m,\alpha}^{*}=1
	\right)
	=
	\Prob_{H_0}
	\left(
	\mathcal T_{n,m}>c_{n,m,1-\alpha}^{*}
	\right)
	\to
	\Prob(L_0>c_\alpha)
	=
	\alpha .
	\]
	
	Under the local alternatives of Theorem~\ref{thm:local-power},
	\[
	\mathcal T_{n,m}
	\rightsquigarrow
	\left\|
	\mathbb G_{\varepsilon}^{(2)}
	+
	\Delta_{\varepsilon,h_{\mathrm P},h_{\mathrm Q}}^{(2)}
	\right\|_{\mathcal H}^{2}.
	\]
	Together with
	\[
	c_{n,m,1-\alpha}^{*}
	\overset{\Prob}{\to}
	c_\alpha,
	\]
	Slutsky's theorem yields
	\[
	\Prob_{n,m,h_{\mathrm P},h_{\mathrm Q}}
	\left(
	\phi_{n,m,\alpha}^{*}=1
	\right)
	\to
	\Prob
	\left(
	\left\|
	\mathbb G_{\varepsilon}^{(2)}
	+
	\Delta_{\varepsilon,h_{\mathrm P},h_{\mathrm Q}}^{(2)}
	\right\|_{\mathcal H}^{2}
	>
	c_\alpha
	\right).
	\]
	
	Finally, under any fixed alternative \(\mathrm P\ne\mathrm Q\),
	Theorem~\ref{thm:consistency} gives
	\[
	\mathcal T_{n,m}\to\infty
	\qquad
	\text{in probability}.
	\]
	By Lemma~\ref{lem:bootstrap-tightness-fixed-alt},
	\[
	c_{n,m,1-\alpha}^{*}
	=
	O_{\mathbb P_{\mathrm P,\mathrm Q}}(1).
	\]
	Therefore,
	\[
	\Prob_{\mathrm P,\mathrm Q}
	\left(
	\mathcal T_{n,m}>c_{n,m,1-\alpha}^{*}
	\right)
	\to
	1,
	\]
	which proves the fixed-alternative consistency of the bootstrap test.
\end{proof}

\section{Linearized representation of the covariance operator}
\label{app:linearized-covariance-operator}

This section gives an operator-level representation of the covariance
operator appearing in Theorem~\ref{thm:null-limit}. The representation is not a
closed-form formula for the eigenvalues, but it identifies the covariance
operator as the Brownian-bridge covariance propagated through the derivative of
the entropic transport map functional.

Let \(H_0:\mathrm P=\mathrm Q=\mathrm P_0\) hold. Let
\((\varphi_0,\psi_0)\) be normalized Schrödinger potentials associated with
\((\mathrm U_d,\mathrm P_0)\), and define
\[
	K_0(\mathbf u,\mathbf x)
	:=
	\exp\left\{
	\frac{
		\varphi_0(\mathbf u)+\psi_0(\mathbf x)-\frac12\|\mathbf u-\mathbf x\|^2
	}{\varepsilon}
	\right\}.
	\]
	Then \(K_0\) is the density of the entropic optimal coupling with respect to
	\(\mathrm U_d\otimes \mathrm P_0\), and it satisfies the marginal normalizations
	\(\int K_0(\mathbf u,\mathbf x)\,d\mathrm P_0(\mathbf x)=1,\ \int K_0(\mathbf u,\mathbf x)\,d\mathrm U_d(\mathbf u)=1\).

For a finite signed perturbation \(\gamma\) of \(\mathrm P_0\) satisfying
\(\gamma(\mathbb R^d)=0\),
define
\[
	\eta_\gamma(\mathbf u)
	:=
	\int K_0(\mathbf u,\mathbf x)\,d\gamma(\mathbf x).
\]
Let \(\mathbb G_{\mathrm P_0}\) denote the \(\mathrm P_0\)-Brownian bridge.
Then
\(\eta_{\mathbb G_{\mathrm P_0}}(\mathbf u) = \mathbb G_{\mathrm P_0}(K_0(\mathbf u,\cdot))\)
is a mean-zero Gaussian process indexed by \(\mathbf u\in\mathbb B_d\). Its covariance
kernel is
\[
C_\eta(\mathbf u,\mathbf v)
:=
\Cov
\left(
\eta_{\mathbb G_{\mathrm P_0}}(\mathbf u),
\eta_{\mathbb G_{\mathrm P_0}}(\mathbf v)
\right).
\]
By the Brownian bridge covariance identity,
\[
\Cov
\left(
\mathbb G_{\mathrm P_0}(f),
\mathbb G_{\mathrm P_0}(g)
\right)
=
\int fg\,d\mathrm P_0
-
\int f\,d\mathrm P_0
\int g\,d\mathrm P_0,
\]
	taking \(f(\mathbf x)=K_0(\mathbf u,\mathbf x)\) and \(g(\mathbf x)=K_0(\mathbf v,\mathbf x)\) gives
	\[
	C_\eta(\mathbf u,\mathbf v)
	=
	\int K_0(\mathbf u,\mathbf x)K_0(\mathbf v,\mathbf x)\,d\mathrm P_0(\mathbf x)-1.
\]

We next describe the derivative of the EOT map. Denote by \(a_\gamma\) and
\(b_\gamma\) the first-order perturbations of \(\varphi_0\) and \(\psi_0\),
respectively, in the target direction \(\gamma\). Formally,
\[
\varphi_t
=
\varphi_0+t a_\gamma+o(t),
\qquad
\psi_t
=
\psi_0+t b_\gamma+o(t),
\qquad
\mathrm P_t
=
\mathrm P_0+t\gamma+o(t).
\]
Differentiating the Schrödinger system at \((\mathrm U_d,\mathrm P_0)\) gives
\[
	a_\gamma(\mathbf u)
	+
	\int K_0(\mathbf u,\mathbf x)b_\gamma(\mathbf x)\,d\mathrm P_0(\mathbf x)
	=
	-\varepsilon
	\int K_0(\mathbf u,\mathbf x)\,d\gamma(\mathbf x),
	\]
	and
	\[
	b_\gamma(\mathbf x)
	+
	\int K_0(\mathbf u,\mathbf x)a_\gamma(\mathbf u)\,d\mathrm U_d(\mathbf u)
	=
	0.
	\]
	Eliminating \(b_\gamma\) yields
	\(\mathcal A_0a_\gamma = -\varepsilon\eta_\gamma\),
	where
	\[
	(\mathcal A_0a)(\mathbf u)
	=
	a(\mathbf u)
	-
	\int K_0(\mathbf u,\mathbf x)
	\left[
	\int K_0(\mathbf v,\mathbf x)a(\mathbf v)\,d\mathrm U_d(\mathbf v)
	\right]
	d\mathrm P_0(\mathbf x).
\]
The operator \(\mathcal A_0\) is considered on the normalized subspace
\[
L^2_0(\mathbb B_d,\mathrm U_d)
=
\left\{
a\in L^2(\mathbb B_d,\mathrm U_d):
\int a\,d\mathrm U_d=0
\right\}.
\]
By Lemma~\ref{lem:reduced-schrodinger-invertibility}, \(\mathcal A_0\) is
boundedly invertible on this subspace. Hence
\(a_\gamma = -\varepsilon\mathcal A_0^{-1}\eta_\gamma\).

Since, under the present cost convention,
	\[
	T_{\varepsilon,\mathrm P}(\mathbf u)
	=
	\mathbf u-\nabla\varphi_{\mathrm P}(\mathbf u),
	\]
	the first-order derivative of the entropic transport map at \(\mathrm P_0\) in
	direction \(\gamma\) is
	\[
	\dot T_{\varepsilon,\mathrm P_0}[\gamma](\mathbf u)
	=
	-\nabla a_\gamma(\mathbf u)
	=
	\varepsilon
	\nabla\mathcal A_0^{-1}\eta_\gamma(\mathbf u).
\]
Therefore the one-sample Gaussian limit admits the representation
\(\mathbb G_\varepsilon = \varepsilon\nabla\mathcal A_0^{-1}\eta_{\mathbb G_{\mathrm P_0}}\).

Consequently, the covariance operator of the one-sample Gaussian EOT map limit
is
\[
\mathcal K_\varepsilon
=
\varepsilon^2
\nabla\mathcal A_0^{-1}
\mathcal C_\eta
(\mathcal A_0^{-1})^\ast
\nabla^\ast .
\]

Finally, under the two-sample null hypothesis and the sample-size balance
condition,
\(\mathbb G_\varepsilon^{(2)} = \sqrt{1-\lambda}\,\mathbb G_\varepsilon^{\mathrm P} - \sqrt{\lambda}\,\mathbb G_\varepsilon^{\mathrm Q}\),
where \(\mathbb G_\varepsilon^{\mathrm P}\) and
\(\mathbb G_\varepsilon^{\mathrm Q}\) are independent copies of
\(\mathbb G_\varepsilon\). Therefore the two-sample covariance operator is
\(\mathcal K_\varepsilon^{(2)} = (1-\lambda)\mathcal K_\varepsilon + \lambda\mathcal K_\varepsilon = \mathcal K_\varepsilon\).
Combining the previous displays gives the operator-level representation
\[
\mathcal K_\varepsilon^{(2)}
=
\varepsilon^2
\nabla\mathcal A_0^{-1}
\mathcal C_\eta
(\mathcal A_0^{-1})^\ast
\nabla^\ast .
\]
Hence the weights \((\omega_k^{(2)})_{k\ge1}\) in
Theorem~\ref{thm:null-limit} are the eigenvalues of this operator.

\begin{lemma}[Invertibility of the reduced linearized Schrödinger operator]
	\label{lem:reduced-schrodinger-invertibility}
	Let \(K_0\) be defined as above, and define
	\[
	(\mathcal M_0 a)(\mathbf u)
	:=
	\int_{\mathcal X} K_0(\mathbf u,\mathbf x)
	\biggl[
	\int_{\mathbb B_d} K_0(\mathbf v,\mathbf x)a(\mathbf v)\,d\mathrm U_d(\mathbf v)
	\biggr]
	d\mathrm P_0(\mathbf x).
	\]
	\[
	\mathcal A_0:=I-\mathcal M_0 .
	\]
	Then \(\mathcal M_0\) is compact, self-adjoint, and positive on
	\(L^2(\mathbb B_d,\mathrm U_d)\), and
	\[
	\ker(\mathcal A_0)
	=
	\{\text{constant functions}\}.
	\]
	Consequently, the restricted operator
	\[
	\mathcal A_0:
	L^2_0(\mathbb B_d,\mathrm U_d)
	\to
	L^2_0(\mathbb B_d,\mathrm U_d),
	\]
	\[
	L^2_0(\mathbb B_d,\mathrm U_d)
	:=
	\bigl\{
	a\in L^2(\mathbb B_d,\mathrm U_d):
	\int_{\mathbb B_d} a\,d\mathrm U_d=0
	\bigr\},
	\]
	is boundedly invertible.
\end{lemma}

\begin{proof}
	Define
	\[
	(L_0a)(\mathbf x)
	:=
	\int_{\mathbb B_d}K_0(\mathbf v,\mathbf x)a(\mathbf v)\,d\mathrm U_d(\mathbf v),
	\qquad
	a\in L^2(\mathbb B_d,\mathrm U_d).
	\]
	Then \(L_0:L^2(\mathbb B_d,\mathrm U_d)\to L^2(\mathcal X,\mathrm P_0)\)
	is an integral operator with bounded kernel \(K_0\). Since \(K_0\) is continuous
	on the compact set \(\mathbb B_d\times\mathcal X\), it is square integrable with
	respect to \(\mathrm U_d\otimes\mathrm P_0\). Hence \(L_0\) is Hilbert--Schmidt
	and therefore compact.
	
	The adjoint of \(L_0\) is
	\[
	(L_0^\ast b)(\mathbf u)
	=
	\int_{\mathcal X}K_0(\mathbf u,\mathbf x)b(\mathbf x)\,d\mathrm P_0(\mathbf x),
	\qquad
	b\in L^2(\mathcal X,\mathrm P_0).
	\]
	Therefore
	\[
	\mathcal M_0=L_0^\ast L_0.
	\]
	It follows that \(\mathcal M_0\) is compact and self-adjoint. 
	
	Moreover, for
	\(a\in L^2(\mathbb B_d,\mathrm U_d)\),
	\[
	\langle \mathcal M_0a,a\rangle_{L^2(\mathrm U_d)}
	=
	\|L_0a\|_{L^2(\mathrm P_0)}^2
	\ge 0,
	\]
	so \(\mathcal M_0\) is positive.
	
	We next identify the kernel of
	\[
	\mathcal A_0=I-\mathcal M_0.
	\]
	Suppose that \(\mathcal A_0a=0\). Then \(\mathcal M_0a=a\), and hence
	\[
	\|a\|_{L^2(\mathrm U_d)}^2
	=
	\langle \mathcal M_0a,a\rangle_{L^2(\mathrm U_d)}
	=
	\|L_0a\|_{L^2(\mathrm P_0)}^2 .
	\]
	Using the marginal normalization
	\[
	\int_{\mathbb B_d}K_0(\mathbf u,\mathbf x)\,d\mathrm U_d(\mathbf u)=1,
	\]
	Jensen's inequality gives, for each \(\mathbf x\),
	\[
	\left[
	\int_{\mathbb B_d}K_0(\mathbf u,\mathbf x)a(\mathbf u)\,d\mathrm U_d(\mathbf u)
	\right]^2
	\le
	\int_{\mathbb B_d}K_0(\mathbf u,\mathbf x)a(\mathbf u)^2\,d\mathrm U_d(\mathbf u).
	\]
	Integrating with respect to \(\mathrm P_0\) yields
	\[
	\|L_0a\|_{L^2(\mathrm P_0)}^2
	\le
	\int_{\mathcal X}\int_{\mathbb B_d}
	K_0(\mathbf u,\mathbf x)a(\mathbf u)^2\,d\mathrm U_d(\mathbf u)\,d\mathrm P_0(\mathbf x).
	\]
	By the other marginal normalization,
	\[
	\int_{\mathcal X}K_0(\mathbf u,\mathbf x)\,d\mathrm P_0(\mathbf x)=1,
	\]
	the right-hand side equals
	\[
	\int_{\mathbb B_d}a(\mathbf u)^2\,d\mathrm U_d(\mathbf u)
	=
	\|a\|_{L^2(\mathrm U_d)}^2.
	\]
	Since equality must hold throughout, equality holds in Jensen's inequality for
	\(\mathrm P_0\)-almost every \(\mathbf x\). For such \(\mathbf x\), the measure
	\(K_0(\mathbf u,\mathbf x)\,d\mathrm U_d(\mathbf u)\) is equivalent to \(\mathrm U_d\), because
	\(K_0(\mathbf u,\mathbf x)>0\). Equality in Jensen's inequality therefore implies that
	\(a\) is \(\mathrm U_d\)-almost surely constant. Hence
	\[
	\ker(\mathcal A_0)
	=
	\{\text{constant functions}\}.
	\]
	
	It remains to restrict the operator to the mean-zero subspace. For every
	\(a\in L^2(\mathbb B_d,\mathrm U_d)\),
	\[
	\begin{aligned}
		\int_{\mathbb B_d}(\mathcal M_0a)(\mathbf u)\,d\mathrm U_d(\mathbf u)
		&=
		\int_{\mathbb B_d}a(\mathbf v)
		\left[
		\int_{\mathcal X}K_0(\mathbf v,\mathbf x)\,d\mathrm P_0(\mathbf x)
		\right]
		d\mathrm U_d(\mathbf v)\\
		&=
		\int_{\mathbb B_d}a(\mathbf v)\,d\mathrm U_d(\mathbf v).
	\end{aligned}
	\]
	Thus \(\mathcal A_0=I-\mathcal M_0\) maps
	\(L^2_0(\mathbb B_d,\mathrm U_d)\) into itself. Since the kernel of
	\(\mathcal A_0\) consists only of constants, its restriction to
	\(L^2_0(\mathbb B_d,\mathrm U_d)\) has trivial kernel.
	
	Finally, \(\mathcal A_0=I-\mathcal M_0\) is the identity minus a compact
	operator, hence it is Fredholm of index zero. The same holds for its restriction
	to the invariant closed subspace \(L^2_0(\mathbb B_d,\mathrm U_d)\). Since the
	restricted operator has trivial kernel and index zero, its range has codimension
	zero. Therefore it is onto \(L^2_0(\mathbb B_d,\mathrm U_d)\). Hence
	\[
	\mathcal A_0:
	L^2_0(\mathbb B_d,\mathrm U_d)
	\to
	L^2_0(\mathbb B_d,\mathrm U_d)
	\]
	is bijective. By the bounded inverse theorem, its inverse is bounded.
\end{proof}

\section{Additional numerical details}

\subsection{Implementation details of the EOT-map test}

\begin{algorithm}[tbp]
	\caption{EOT-map two-sample test}
	\label{alg:eot-map-two-sample-test}
	\footnotesize\begin{algorithmic}[1]
		\REQUIRE Samples \(\mathbf X_1,\ldots,\mathbf X_n\), \(\mathbf Y_1,\ldots,\mathbf Y_m\);
		regularization parameter \(\varepsilon>0\); number of reference points
		\(N\); number of bootstrap repetitions \(B\); significance level
		\(\alpha\).
		
		\STATE Generate common reference points
		\(\mathbf U_1,\ldots,\mathbf U_N\overset{\mathrm{i.i.d.}}{\sim}\mathrm U_d\), and set
		\[
		\mathrm U_{d,N}=\frac1N\sum_{\ell=1}^N\delta_{\mathbf U_\ell},
		\qquad
		\mathrm P_n=\frac1n\sum_{i=1}^n\delta_{\mathbf X_i},
		\qquad
		\mathrm Q_m=\frac1m\sum_{j=1}^m\delta_{\mathbf Y_j}.
		\]
		
		\STATE Compute the empirical EOT maps
		\[
		\widehat T_{\varepsilon,n}^{\mathrm P}
		=
		T_{\varepsilon,(\mathrm U_{d,N},\mathrm P_n)},
		\qquad
		\widehat T_{\varepsilon,m}^{\mathrm Q}
		=
		T_{\varepsilon,(\mathrm U_{d,N},\mathrm Q_m)}
		\]
		at the reference points \(\mathbf U_1,\ldots,\mathbf U_N\).
		
		\STATE Compute
		\[
		\mathcal T_{n,m}
		=
		\frac{nm}{n+m}
		\frac1N
		\sum_{\ell=1}^N
		\left\|
		\widehat T_{\varepsilon,n}^{\mathrm P}(\mathbf U_\ell)
		-
		\widehat T_{\varepsilon,m}^{\mathrm Q}(\mathbf U_\ell)
		\right\|^2 .
		\]
		
		\FOR{\(b=1,\ldots,B\)}
		\STATE Generate independent nonnegative multiplier weights
		\(\xi_1^{(b)},\ldots,\xi_n^{(b)}\) and
		\(\zeta_1^{(b)},\ldots,\zeta_m^{(b)}\), normalized to have sample sums
		one within each sample.
		
		\STATE Form the weighted empirical measures
		\[
		\mathrm P_n^{*,b}
		=
		\sum_{i=1}^n w_{i}^{\mathrm P,b}\delta_{\mathbf X_i},
		\qquad
		\mathrm Q_m^{*,b}
		=
		\sum_{j=1}^m w_{j}^{\mathrm Q,b}\delta_{\mathbf Y_j},
		\]
		where
		\[
		w_i^{\mathrm P,b}
		=
		\frac{\xi_i^{(b)}}{\sum_{r=1}^n\xi_r^{(b)}},
		\qquad
		w_j^{\mathrm Q,b}
		=
		\frac{\zeta_j^{(b)}}{\sum_{r=1}^m\zeta_r^{(b)}} .
		\]
		
		\STATE Compute the bootstrap EOT maps
		\[
		\widehat T_{\varepsilon,n}^{\mathrm P,*,b}
		=
		T_{\varepsilon,(\mathrm U_{d,N},\mathrm P_n^{*,b})},
		\qquad
		\widehat T_{\varepsilon,m}^{\mathrm Q,*,b}
		=
		T_{\varepsilon,(\mathrm U_{d,N},\mathrm Q_m^{*,b})}.
		\]
		
		\STATE Compute the centered bootstrap statistic
		\[
		\mathcal T_{n,m}^{*,b}
		=
		\frac{nm}{n+m}
		\frac1N
		\sum_{\ell=1}^N
		\left\|
		\left(
		\widehat T_{\varepsilon,n}^{\mathrm P,*,b}
		-
		\widehat T_{\varepsilon,n}^{\mathrm P}
		\right)(\mathbf U_\ell)
		-
		\left(
		\widehat T_{\varepsilon,m}^{\mathrm Q,*,b}
		-
		\widehat T_{\varepsilon,m}^{\mathrm Q}
		\right)(\mathbf U_\ell)
		\right\|^2 .
		\]
		\ENDFOR
		
		\STATE Let \(\widehat c_{n,m,N,1-\alpha}^{*}\) be the empirical
		\((1-\alpha)\)-quantile of
		\(\mathcal T_{n,m}^{*,1},\ldots,\mathcal T_{n,m}^{*,B}\).
		
		\STATE Reject \(H_0:\mathrm P=\mathrm Q\) if
		\(\mathcal T_{n,m}>\widehat c_{n,m,N,1-\alpha}^{*}\).
	\end{algorithmic}
\end{algorithm}

\subsection{Additional benchmark comparison results}
\label{app:benchmark-comparison-details}

This appendix provides the full detailed results and extended discussion for the benchmark comparison in Section~\ref{subsec:comparison-existing-tests}. All experimental settings are identical to those described in the main text.

\noindent\textbf{Null calibration.}
We first evaluate empirical size under the Gaussian null
\(\mathrm P=\mathrm Q=N(0,I_d)\), the uniform null
\(\mathrm P=\mathrm Q=\operatorname{Unif}([-1,1]^d)\), and the Student-\(t\) null.

The results are summarized in Table~\ref{tab:benchmark-size} in the main text. Overall, all methods exhibit reasonable size control. The empirical sizes of the proposed EOT-map test range from $0.039$ to $0.064$, while those of OT-rank Energy range from $0.046$ to $0.066$. Both are close to the nominal level $\alpha=0.05$ given the Monte Carlo standard error. The benchmark tests show comparable calibration.

\noindent\textbf{Power under fixed alternatives.}
We next compare empirical power under five classes of fixed alternatives. In all cases, the baseline distribution is $\mathrm P=N(0,I_d)$. The alternatives include a location shift $\mathrm Q=N(\delta \mathbf e_1,I_d)$, a scale change $\mathrm Q=N(0,\sigma^2 I_d)$, a correlation change $\mathrm Q=N(0,\Sigma_\rho)$, where $\Sigma_\rho$ has unit diagonal entries and common off-diagonal entry $\rho$, and a symmetric Gaussian mixture $\mathrm Q = \frac12 N(-\delta \mathbf e_1,I_d) + \frac12 N(\delta \mathbf e_1,I_d)$. The mixture alternative preserves the mean but changes the distributional shape and increases the variance in the first coordinate.

We also consider a nonlinear local deformation. For $d\in\{2,5\}$, let $\mathbf X=(X_1,\ldots,X_d)^\top\sim N(0,I_d)$ and define \[\mathbf Y_\tau = \mathbf X+ \tau (\sin(X_2), 0.5\sin(X_1), 0, \ldots, 0)^\top.\] The alternative distribution is $\mathrm Q_\tau=\mathcal L(\mathbf Y_\tau)$, where $\mathcal L(\mathbf Y_\tau)$ denotes the law of $\mathbf Y_\tau$. For $d=2$, this is a two-dimensional nonlinear deformation; for $d=5$, the same deformation is embedded in the first two coordinates, with the remaining coordinates acting as nuisance dimensions.

Figure~\ref{fig:location-power-comparison} in the main text visualizes the empirical rejection rates under the location alternatives. The EOT-map test is among the strongest methods across the considered dimensions. In dimension $d=2$, its power is comparable to that of the energy, MMD, Wasserstein, and OT-rank Energy tests. In dimension $d=5$, the advantage of EOT-map becomes more visible for weak-to-moderate shifts. For example, around $\delta=0.2$, EOT-map has a higher rejection rate than OT-rank Energy and Wasserstein; by $\delta=0.4$, its rejection rate is close to $0.9$, while the corresponding curves for OT-rank Energy and Wasserstein remain visibly lower. Overall, the proposed method remains competitive across dimensions and is particularly effective when the location shift induces a coherent directional displacement between the two estimated maps from the common reference distribution.

Table~\ref{tab:benchmark-power-other-full} reports the empirical power under scale, correlation, symmetric mixture, and nonlinear deformation alternatives.

\begin{table}[!htbp]
	\centering
	\resizebox{0.98\textwidth}{!}{
		\begin{tabular}{lllccccccc}
			\toprule
			Alternative
			& Dimension
			& Parameter
			& Energy
			& MMD
			& Sinkhorn
			& \(k\)-NN
			& Wasserstein
			& OT-rank
			& EOT-map \\
			\midrule
			
			\multirow{8}{*}{Scale}
			& \multirow{4}{*}{\(d=2\)}
			& \(\sigma=1.1\) & 0.095 & 0.146 & 0.156 & 0.088 & 0.155 & 0.072 & 0.137 \\
			& & \(\sigma=1.2\) & 0.365 & 0.567 & 0.576 & 0.222 & 0.584 & 0.124 & 0.614 \\
			& & \(\sigma=1.4\) & 0.964 & 0.996 & 1.000 & 0.761 & 0.997 & 0.547 & 0.999 \\
			& & \(\sigma=1.6\) & 1.000 & 1.000 & 1.000 & 0.995 & 1.000 & 0.963 & 1.000 \\
			\cmidrule(lr){2-10}
			& \multirow{4}{*}{\(d=5\)}
			& \(\sigma=1.1\) & 0.136 & 0.183 & 0.170 & 0.108 & 0.143 & 0.055 & 0.110 \\
			& & \(\sigma=1.2\) & 0.632 & 0.797 & 0.645 & 0.351 & 0.579 & 0.096 & 0.414 \\
			& & \(\sigma=1.4\) & 1.000 & 1.000 & 1.000 & 0.973 & 1.000 & 0.292 & 0.997 \\
			& & \(\sigma=1.6\) & 1.000 & 1.000 & 1.000 & 1.000 & 1.000 & 0.692 & 1.000 \\
			
			\midrule
			
			\multirow{8}{*}{Correlation}
			& \multirow{4}{*}{\(d=2\)}
			& \(\rho=0.2\) & 0.072 & 0.086 & 0.173 & 0.087 & 0.174 & 0.071 & 0.168 \\
			& & \(\rho=0.4\) & 0.160 & 0.297 & 0.721 & 0.322 & 0.722 & 0.260 & 0.730 \\
			& & \(\rho=0.6\) & 0.634 & 0.879 & 0.999 & 0.871 & 0.998 & 0.832 & 0.998 \\
			& & \(\rho=0.8\) & 0.999 & 1.000 & 1.000 & 1.000 & 1.000 & 1.000 & 1.000 \\
			\cmidrule(lr){2-10}
			& \multirow{4}{*}{\(d=5\)}
			& \(\rho=0.2\) & 0.134 & 0.186 & 0.722 & 0.347 & 0.700 & 0.072 & 0.561 \\
			& & \(\rho=0.4\) & 0.735 & 0.935 & 1.000 & 0.974 & 1.000 & 0.251 & 1.000 \\
			& & \(\rho=0.6\) & 1.000 & 1.000 & 1.000 & 1.000 & 1.000 & 0.748 & 1.000 \\
			& & \(\rho=0.8\) & 1.000 & 1.000 & 1.000 & 1.000 & 1.000 & 1.000 & 1.000 \\
			
			\midrule
			
			\multirow{8}{*}{Mixture}
			& \multirow{4}{*}{\(d=2\)}
			& \(\delta=0.5\) & 0.078 & 0.102 & 0.126 & 0.060 & 0.121 & 0.071 & 0.120 \\
			& & \(\delta=1.0\) & 0.785 & 0.938 & 0.968 & 0.616 & 0.966 & 0.352 & 0.974 \\
			& & \(\delta=1.5\) & 1.000 & 1.000 & 1.000 & 1.000 & 1.000 & 1.000 & 1.000 \\
			& & \(\delta=2.0\) & 1.000 & 1.000 & 1.000 & 1.000 & 1.000 & 1.000 & 1.000 \\
			\cmidrule(lr){2-10}
			& \multirow{4}{*}{\(d=5\)}
			& \(\delta=0.5\) & 0.049 & 0.048 & 0.070 & 0.063 & 0.062 & 0.071 & 0.060 \\
			& & \(\delta=1.0\) & 0.301 & 0.440 & 0.739 & 0.444 & 0.708 & 0.065 & 0.588 \\
			& & \(\delta=1.5\) & 0.998 & 1.000 & 1.000 & 0.998 & 1.000 & 0.175 & 1.000 \\
			& & \(\delta=2.0\) & 1.000 & 1.000 & 1.000 & 1.000 & 1.000 & 0.462 & 1.000 \\
			
			\midrule
			
			\multirow{8}{*}{\shortstack{Nonlinear \\ deformation}}
			& \multirow{4}{*}{\(d=2\)}
			& \(\tau=0.25\) & 0.084 & 0.112 & 0.220 & 0.107 & 0.211 & 0.094 & 0.215 \\
			& & \(\tau=0.40\) & 0.156 & 0.292 & 0.636 & 0.297 & 0.638 & 0.233 & 0.639 \\
			& & \(\tau=0.60\) & 0.497 & 0.759 & 0.966 & 0.742 & 0.966 & 0.724 & 0.967 \\
			& & \(\tau=0.70\) & 0.723 & 0.926 & 0.997 & 0.905 & 0.998 & 0.921 & 0.996 \\
			\cmidrule(lr){2-10}
			& \multirow{4}{*}{\(d=5\)}
			& \(\tau=0.25\) & 0.072 & 0.077 & 0.128 & 0.092 & 0.124 & 0.068 & 0.098 \\
			& & \(\tau=0.40\) & 0.077 & 0.079 & 0.299 & 0.191 & 0.262 & 0.074 & 0.165 \\
			& & \(\tau=0.60\) & 0.141 & 0.176 & 0.739 & 0.510 & 0.691 & 0.104 & 0.444 \\
			& & \(\tau=0.70\) & 0.180 & 0.232 & 0.878 & 0.689 & 0.855 & 0.130 & 0.618 \\
			\bottomrule
		\end{tabular}
	}
	\caption{
		Empirical power under scale, correlation, mixture, and nonlinear
		deformation alternatives.
		The nominal level is \(\alpha=0.05\), \(n=m=200\), and each entry is based on
		\(R=1000\) Monte Carlo replications.
	}
	\label{tab:benchmark-power-other-full}
\end{table}

For scale alternatives, the EOT-map test is competitive in dimension $d=2$, especially for moderate-to-large scale changes. In dimension $d=5$, weak scale changes are better detected by MMD and transport-cost-based tests, while the EOT-map test approaches full power once the scale difference becomes stronger. OT-rank Energy is well calibrated but is less sensitive to scale changes, especially in dimension $d=5$. For example, at $\sigma=1.4$, its empirical power is $0.292$ in dimension $d=5$, compared with $0.997$ for EOT-map.

For correlation alternatives, Sinkhorn divergence and Wasserstein are particularly sensitive in the weak high-dimensional regime. The EOT-map test remains competitive and substantially outperforms Energy and MMD in several settings. In dimension $d=2$, it is essentially as powerful as the strongest transport-based competitors. OT-rank Energy gains power as the correlation signal increases, but is substantially less powerful than the transport-cost and EOT-map methods for weak and moderate correlation changes, especially in dimension $d=5$. At $\rho=0.2$, its power is $0.072$ in dimension $d=5$, whereas the EOT-map test reaches $0.561$.

For symmetric mixture alternatives, the experiment examines sensitivity to non-Gaussian shape changes with no mean shift. In dimension $d=2$, the EOT-map test is close to the strongest methods and is slightly more powerful for moderate mixture separation. In dimension $d=5$, it is less powerful than Sinkhorn divergence and Wasserstein at moderate separation, but remains more powerful than Energy, MMD, and $k$-NN. OT-rank Energy is comparatively less sensitive to this high-dimensional shape alternative: even at $\delta=2.0$, its power is $0.462$ in dimension $d=5$, while the other methods have essentially full power. This suggests that the rank-based energy statistic may lose some of the geometric information captured by the EOT-map discrepancy.

For nonlinear deformation alternatives, the transport-based methods are substantially more sensitive than Energy, MMD, and OT-rank Energy in dimension $d=2$. The EOT-map test has power nearly identical to Sinkhorn divergence and Wasserstein across the reported perturbation strengths. For example, at $\tau=0.40$, its empirical power is $0.639$, compared with $0.636$ for Sinkhorn divergence and $0.638$ for Wasserstein. In dimension $d=5$, where the same two-dimensional deformation is embedded in additional nuisance coordinates, Sinkhorn divergence and Wasserstein are more powerful. The EOT-map test is less sensitive than these transport-cost-based procedures, but its power still increases steadily with the perturbation strength and remains stronger than Energy, MMD, and OT-rank Energy for larger values of $\tau$. At $\tau=0.70$, for instance, EOT-map has power $0.618$, compared with $0.180$, $0.232$, and $0.130$ for Energy, MMD, and OT-rank Energy, respectively.

\noindent\textbf{Computation time.}

Table~\ref{tab:benchmark-runtime} reports the average computation time per test. In the present implementation, the EOT-map test is more computationally demanding than the energy, MMD, $k$-NN, Wasserstein, and OT-rank Energy tests, but is faster than the permutation-calibrated Sinkhorn divergence test.

\begin{table}[!htbp]
	\centering
	\resizebox{0.98\textwidth}{!}{
	\begin{tabular}{lccccccc}
		\toprule
		Dimension
		& Energy
		& MMD
		& Sinkhorn
		& \(k\)-NN
		& Wasserstein
		& OT-rank
		& EOT-map \\
		\midrule
		\(d=2\) & 0.19 & 0.39 & 12.12 & 0.02 & 1.37 & 0.09 & 2.14 \\
		\(d=5\) & 0.23 & 0.45 & 9.80  & 0.02 & 1.49 & 0.12 & 1.11 \\
		\bottomrule
	\end{tabular}
	}
	\caption{
		Average computation time per test, in seconds.
	}
	\label{tab:benchmark-runtime}
\end{table}

This advantage mainly comes from the calibration and debiasing costs of Sinkhorn divergence. Each evaluation of the debiased Sinkhorn divergence needs to solve the minimizer problem \eqref{eq:sinkhorncost} three times, and such calculations are performed repeatedly in permutation calibration. The self-transport terms may also be numerically demanding when $\varepsilon$ is small, because identical empirical supports can induce sharp near-diagonal couplings.

By contrast, the EOT-map test uses a common reference sample and weighted bootstrap calibration. Its computational cost can be directly controlled by the number $N$ of reference points. Since the sensitivity analysis shows that moderate values of $N$ already yield stable empirical size and power, the proposed method provides a practical way to reduce computation, especially in larger-sample settings. The runtime comparison should nevertheless be viewed as implementation- and sample-size-dependent rather than as a theoretical complexity ordering.

\noindent\textbf{Summary of the comparison.}
The benchmark results show that the proposed EOT-map test has well-calibrated finite-sample size and is competitive with existing two-sample procedures. It is particularly effective when the distributional difference admits a geometric or map-level interpretation. Under location alternatives, it is among the strongest methods considered, which is consistent with the fact that a location shift induces a coherent directional displacement between the two estimated maps from the common reference distribution.

The additional structural experiments further highlight the role of the map-based statistic. Under nonlinear deformation alternatives, the EOT-map test substantially outperforms Energy and MMD and achieves power essentially identical to Sinkhorn divergence. Under correlation alternatives, it is competitive with Sinkhorn, especially in dimension $d=2$, and remains substantially more powerful than Energy and MMD for several weak-to-moderate signals. Under symmetric mixture alternatives, the EOT-map test remains competitive, although Sinkhorn divergence and Wasserstein are more sensitive in the moderate high-dimensional setting and OT-rank Energy is substantially less powerful.

The results also show that the EOT-map test does not uniformly dominate all existing scalar discrepancy tests. Under some high-dimensional correlation or scale alternatives, Sinkhorn divergence or MMD can be more sensitive. Thus, the proposed method should be viewed as a competitive and interpretable complement to scalar discrepancy tests, rather than as a uniformly dominant replacement.

A distinctive advantage of the EOT-map test is that, in addition to a rejection decision, it provides the estimated vector field $\widehat T_{\varepsilon,n}^{\mathrm P}(\mathbf u) - \widehat T_{\varepsilon,m}^{\mathrm Q}(\mathbf u)$, which can be visualized to reveal whether the distributional difference is primarily a translation, a radial expansion, a dependence deformation, a local nonlinear distortion, or a more complex shape change.

\subsection{Sensitivity analyses and computation time}
\label{app:sensitivity-computation}

\noindent\textbf{Sensitivity to the regularization parameter.}

We next examine the sensitivity of the proposed test to the entropic
regularization parameter \(\varepsilon\), which controls the smoothing level of
the EOT problem and affects both statistical performance and numerical
stability. The goal is to assess whether the procedure is stable over a
reasonable range of regularization levels and to illustrate the associated
statistical--computational trade-off.

We consider two types of choices for \(\varepsilon\): fixed values
\(\varepsilon\in\{0.2,0.5,1.0\}\), and a scale-adaptive pilot rule with
\(c_\varepsilon\in\{0.01,0.1,0.2,0.5,1.0,10000\}\). The pilot rule uses a
Gaussian pilot sample of size \(N_{\mathrm{pilot}}=3000\), and the resulting
\(\varepsilon\) is fixed across all Monte Carlo replications for the
corresponding configuration. The very large value \(c_\varepsilon=10000\) is
included only as a stress test for excessive regularization.

The experiment is conducted for \(d\in\{2,5\}\) and
\(n=m\in\{200,500,1000\}\). We consider the Gaussian null
\(\mathrm P=\mathrm Q=N(0,I_d)\) and the mean-shift alternative
\(\mathrm P=N(0,I_d)\), \(\mathrm Q=N(0.4\mathbf e_1,I_d)\). For each configuration,
we use \(R=500\) Monte Carlo replications and \(B=300\) bootstrap repetitions,
set the nominal level to \(\alpha=0.05\), and take the number of reference
points to be \(N=n\). Tables~\ref{tab:epsilon-sensitivity-size-time}
and~\ref{tab:epsilon-sensitivity-power-time} report empirical rejection
probabilities, with the mean runtime per replication shown in parentheses.

The results show a clear statistical--computational trade-off. Very small
regularization can lead to size distortion and numerical instability,
especially in higher dimension. For example, when \(d=5\), the fixed choice
\(\varepsilon=0.2\) yields empirical sizes \(0.370\), \(0.646\), and \(0.714\)
for \(n=m=200,500,1000\), respectively. Similarly, the adaptive choice
\(c_\varepsilon=0.01\) produces empirical sizes \(0.190\), \(0.798\), and
\(0.994\) in the same settings. These values are far above the nominal level,
suggesting that overly small regularization can make finite-sample calibration
unreliable.

Moderate adaptive choices provide much more stable size control. In particular,
for \(d=5\), the choices \(c_\varepsilon=0.2\), \(0.5\), and \(1.0\) yield
empirical sizes close to the nominal level across all sample sizes. For example,
with \(c_\varepsilon=0.2\), the empirical sizes are \(0.048\), \(0.056\), and
\(0.058\) for \(n=m=200,500,1000\), respectively. The results in dimension
\(d=2\) are also broadly stable for moderate and large regularization levels.

Across the stable range of \(\varepsilon\), the test retains high power under
the mean-shift alternative. For instance, when \(d=5\) and
\(c_\varepsilon=0.2\), the empirical powers are \(0.870\), \(1.000\), and
\(1.000\) for \(n=m=200,500,1000\), respectively. Thus, moderate regularization
improves calibration without sacrificing power in this setting.

The stress case \(c_\varepsilon=10000\) should not be interpreted as evidence
that larger regularization is always preferable. Extremely large
\(\varepsilon\) heavily smooths the EOT map, so the resulting map discrepancy is
mainly driven by low-order features such as differences in means. This explains
why excessive regularization can still perform well for the mean-shift
alternative. However, it may suppress higher-order distributional features such
as scale, dependence, multimodality, or local shape differences.

To illustrate this point, we conducted an additional stress experiment under
the equal-mean mixture alternative
\[
\mathrm P=N(0,I_d),
\qquad
\mathrm Q=\frac12 N(-1.5\mathbf e_1,I_d)+\frac12 N(1.5\mathbf e_1,I_d),
\]
with \(d=5\) and \(n=m=1000\). Under the moderate adaptive choice
\(c_\varepsilon=0.5\), the empirical rejection probability was \(1.000\).
By contrast, under \(c_\varepsilon=10000\), it dropped to \(0.036\). This
confirms that excessive regularization can severely reduce power against
non-location alternatives, even when it appears effective for detecting mean
shifts.

Computation time is also strongly affected by \(\varepsilon\). Smaller
regularization generally requires more Sinkhorn iterations and is therefore
substantially more expensive. This is particularly visible for \(d=2\) and
\(n=m=1000\), where \(c_\varepsilon=0.01\) requires approximately \(807.9\)
seconds per replication under the null and \(934.3\) seconds under the
alternative. Moderate or large regularization levels substantially reduce the
runtime.

Overall, the experiment suggests that \(\varepsilon\) should be neither too
small nor excessively large. Very small values may lead to numerical
instability and size distortion, whereas extremely large values may oversmooth
the EOT maps and suppress distributional features beyond mean differences.
Moderate scale-adaptive choices, such as \(c_\varepsilon=0.2\) or
\(c_\varepsilon=0.5\), provide a favorable balance between calibration, power,
and computational efficiency.

\begin{table}[!htbp]
	\centering
	
	\resizebox{0.98\textwidth}{!}{
		\begin{tabular}{llccc}
			\toprule
			Dimension & Regularization
			& \(n=m=200\)
			& \(n=m=500\)
			& \(n=m=1000\) \\
			\midrule
			\multirow{9}{*}{\(d=2\)}
			& fixed \(\varepsilon=0.2\)
			& 0.044 (4.7) & 0.068 (24.3) & 0.048 (134.4) \\
			& fixed \(\varepsilon=0.5\)
			& 0.030 (2.3) & 0.052 (19.2) & 0.066 (81.9) \\
			& fixed \(\varepsilon=1.0\)
			& 0.048 (1.4) & 0.058 (18.8) & 0.068 (71.8) \\
			& adaptive \(c_\varepsilon=0.01\)
			& 0.022 (31.0) & 0.038 (93.4) & 0.034 (807.9) \\
			& adaptive \(c_\varepsilon=0.1\)
			& 0.044 (3.6) & 0.032 (24.5) & 0.038 (116.0) \\
			& adaptive \(c_\varepsilon=0.2\)
			& 0.048 (2.2) & 0.050 (18.6) & 0.052 (81.9) \\
			& adaptive \(c_\varepsilon=0.5\)
			& 0.056 (1.4) & 0.052 (17.6) & 0.056 (71.5) \\
			& adaptive \(c_\varepsilon=1.0\)
			& 0.076 (1.0) & 0.074 (16.8) & 0.050 (67.6) \\
			& adaptive \(c_\varepsilon=10000\)
			& 0.044 (0.5) & 0.060 (14.8) & 0.062 (60.9) \\
			\midrule
			\multirow{9}{*}{\(d=5\)}
			& fixed \(\varepsilon=0.2\)
			& 0.370 (5.4) & 0.646 (24.5) & 0.714 (130.2) \\
			& fixed \(\varepsilon=0.5\)
			& 0.114 (2.9) & 0.104 (18.8) & 0.080 (81.4) \\
			& fixed \(\varepsilon=1.0\)
			& 0.074 (2.2) & 0.058 (18.0) & 0.060 (73.8) \\
			& adaptive \(c_\varepsilon=0.01\)
			& 0.190 (12.2) & 0.798 (38.3) & 0.994 (254.0) \\
			& adaptive \(c_\varepsilon=0.1\)
			& 0.082 (1.8) & 0.046 (17.9) & 0.042 (74.2) \\
			& adaptive \(c_\varepsilon=0.2\)
			& 0.048 (1.3) & 0.056 (17.1) & 0.058 (69.9) \\
			& adaptive \(c_\varepsilon=0.5\)
			& 0.050 (1.3) & 0.048 (17.1) & 0.056 (69.8) \\
			& adaptive \(c_\varepsilon=1.0\)
			& 0.054 (1.3) & 0.058 (17.1) & 0.058 (69.6) \\
			& adaptive \(c_\varepsilon=10000\)
			& 0.060 (0.6) & 0.068 (14.8) & 0.050 (60.1) \\
			\bottomrule
		\end{tabular}
	}
	\caption{
		Sensitivity of empirical size and computation time to the regularization
		parameter. The nominal level is \(\alpha=0.05\). Results are empirical
		rejection probabilities under the Gaussian null; numbers in parentheses
		are mean computation times per replication, in seconds.
	}
	\label{tab:epsilon-sensitivity-size-time}
\end{table}

\begin{table}[!htbp]
	\centering
	
	\resizebox{0.98\textwidth}{!}{
		\begin{tabular}{llccc}
			\toprule
			Dimension & Regularization
			& \(n=m=200\)
			& \(n=m=500\)
			& \(n=m=1000\) \\
			\midrule
			\multirow{9}{*}{\(d=2\)}
			& fixed \(\varepsilon=0.2\)
			& 0.950 (4.9) & 1.000 (24.8) & 1.000 (155.3) \\
			& fixed \(\varepsilon=0.5\)
			& 0.938 (2.4) & 1.000 (20.1) & 1.000 (87.9) \\
			& fixed \(\varepsilon=1.0\)
			& 0.946 (1.5) & 1.000 (17.6) & 1.000 (72.0) \\
			& adaptive \(c_\varepsilon=0.01\)
			& 0.894 (32.1) & 1.000 (95.4) & 1.000 (934.3) \\
			& adaptive \(c_\varepsilon=0.1\)
			& 0.952 (3.5) & 1.000 (23.1) & 1.000 (123.6) \\
			& adaptive \(c_\varepsilon=0.2\)
			& 0.954 (2.3) & 1.000 (19.9) & 1.000 (83.2) \\
			& adaptive \(c_\varepsilon=0.5\)
			& 0.948 (1.4) & 1.000 (17.7) & 1.000 (72.4) \\
			& adaptive \(c_\varepsilon=1.0\)
			& 0.952 (1.0) & 1.000 (16.9) & 1.000 (68.8) \\
			& adaptive \(c_\varepsilon=10000\)
			& 0.966 (0.5) & 1.000 (14.8) & 1.000 (60.5) \\
			\midrule
			\multirow{9}{*}{\(d=5\)}
			& fixed \(\varepsilon=0.2\)
			& 0.934 (5.1) & 1.000 (23.7) & 1.000 (132.2) \\
			& fixed \(\varepsilon=0.5\)
			& 0.894 (2.6) & 1.000 (18.8) & 1.000 (81.3) \\
			& fixed \(\varepsilon=1.0\)
			& 0.852 (1.8) & 0.996 (17.9) & 1.000 (74.1) \\
			& adaptive \(c_\varepsilon=0.01\)
			& 0.830 (12.8) & 1.000 (39.1) & 1.000 (272.6) \\
			& adaptive \(c_\varepsilon=0.1\)
			& 0.856 (1.8) & 0.998 (17.9) & 1.000 (73.2) \\
			& adaptive \(c_\varepsilon=0.2\)
			& 0.870 (1.3) & 1.000 (17.2) & 1.000 (68.4) \\
			& adaptive \(c_\varepsilon=0.5\)
			& 0.892 (1.3) & 0.998 (17.1) & 1.000 (68.2) \\
			& adaptive \(c_\varepsilon=1.0\)
			& 0.894 (1.3) & 1.000 (17.1) & 1.000 (68.1) \\
			& adaptive \(c_\varepsilon=10000\)
			& 0.908 (0.6) & 1.000 (14.8) & 1.000 (60.1) \\
			\bottomrule
		\end{tabular}
	}
	\caption{
		Sensitivity of empirical power and computation time to the regularization
		parameter under the Gaussian mean-shift alternative. Numbers in
		parentheses are mean computation times per replication, in seconds.
	}
	\label{tab:epsilon-sensitivity-power-time}
\end{table}

\noindent\textbf{Reference-sample economy.} \label{para:sensitive_reference}
We next examine the effect of the number \(N\) of reference points used to
approximate the integral over the common reference domain \(\mathbb B_d\).
Unlike the data sample sizes \(n\) and \(m\), the quantity \(N\) controls only
the Monte Carlo approximation of the outer reference integral. It therefore
provides a direct computational tuning parameter for the proposed method.
The experiment is conducted with \(n=m=400\), \(d\in\{2,5\}\), nominal level
\(\alpha=0.05\), \(B=300\) bootstrap repetitions, and \(R=500\) Monte Carlo
replications. Within each dimension, the regularization parameter is fixed by
the pilot rule with \(c_\varepsilon=0.2\). We vary
\(N\in\{50,100,200,400,600,800\}\).

Table~\ref{tab:n-reference-sensitivity} shows that the empirical rejection
probabilities are stable over a wide range of reference sample sizes. Under the
Gaussian null, the empirical sizes remain close to the nominal level across all
choices of \(N\). They range from \(0.044\) to \(0.062\) for \(d=2\), and from
\(0.032\) to \(0.068\) for \(d=5\), which is consistent with the Monte Carlo
fluctuation expected from \(R=500\) replications. Under the mean-shift and
correlation alternatives, the power is already close to one for moderate values
of \(N\), and increasing \(N\) further does not lead to a systematic improvement.

In contrast, the computational cost increases substantially with \(N\). For
example, under the Gaussian null with \(d=2\), the average runtime per
replication increases from \(2.12\) seconds at \(N=100\) to \(11.69\) seconds at
\(N=400\), and further to \(22.30\) seconds at \(N=800\). These findings suggest
a useful reference-sample economy of the proposed method: the number of
reference points is a computational tuning parameter, and moderate choices such
as \(N=100\) or \(N=200\) can provide stable empirical calibration and power
while substantially reducing computation time.

\begin{table}[!htbp]
	\centering
	\resizebox{0.7\textwidth}{!}{
	\begin{tabular}{ccccc}
		\toprule
		\(d\) & \(N\) & Gaussian null & Mean shift & Correlation \\
		\midrule
		2 & 50
		& \(0.044\;(1.6)\)
		& \(1.000\;(1.6)\)
		& \(0.980\;(1.6)\) \\
		2 & 100
		& \(0.060\;(2.1)\)
		& \(0.998\;(2.1)\)
		& \(0.984\;(2.1)\) \\
		2 & 200
		& \(0.062\;(5.0)\)
		& \(1.000\;(5.1)\)
		& \(0.988\;(5.1)\) \\
		2 & 400
		& \(0.058\;(11.7)\)
		& \(1.000\;(11.8)\)
		& \(0.992\;(11.8)\) \\
		2 & 600
		& \(0.050\;(16.9)\)
		& \(0.996\;(16.9)\)
		& \(0.986\;(16.9)\) \\
		2 & 800
		& \(0.050\;(22.3)\)
		& \(0.998\;(22.3)\)
		& \(0.978\;(22.3)\) \\
		\midrule
		5 & 50
		& \(0.032\;(0.7)\)
		& \(0.994\;(0.7)\)
		& \(1.000\;(0.7)\) \\
		5 & 100
		& \(0.054\;(1.0)\)
		& \(0.994\;(1.0)\)
		& \(1.000\;(1.0)\) \\
		5 & 200
		& \(0.042\;(2.5)\)
		& \(0.998\;(2.5)\)
		& \(1.000\;(2.5)\) \\
		5 & 400
		& \(0.048\;(10.4)\)
		& \(0.996\;(10.4)\)
		& \(1.000\;(10.4)\) \\
		5 & 600
		& \(0.056\;(15.0)\)
		& \(0.998\;(15.0)\)
		& \(1.000\;(15.1)\) \\
		5 & 800
		& \(0.068\;(20.1)\)
		& \(0.998\;(20.2)\)
		& \(1.000\;(20.2)\) \\
		\bottomrule
	\end{tabular}
	}
	\caption{
		Sensitivity of the EOT-map test to the number of reference points \(N\).
		Each entry reports the empirical rejection probability, with the average runtime
		per Monte Carlo replication shown in parentheses. 
	}
	\label{tab:n-reference-sensitivity}
\end{table}

\end{document}